\documentclass[a4paper,twocolumn,notitlepage,nofootinbib,longbibliography,superscriptaddress,floatfix]{revtex4-2}
\usepackage{
    adjustbox,
    algorithm,
    algpseudocode,
    amsmath,
    amssymb,
    amsthm,
    booktabs,
    braket,
    csquotes,
    enumitem,
    graphicx,
    IEEEtrantools,
    mathtools,
    multirow,
    pifont,
    placeins,
    refcount,
    relsize,
    tabls,
    tikz-cd,
    xcolor
}

\usepackage[english]{babel}
\usepackage[normalem]{ulem}
\usepackage[caption=false]{subfig}
\usepackage{soul}

\definecolor{mainblue}{HTML}{1f77b4}
\definecolor{mainorange}{HTML}{ff7f0e}
\definecolor{maingreen}{HTML}{2cc92c}
\definecolor{mainred}{HTML}{DC3522}
\definecolor{mainpurple}{HTML}{9467bd}
\definecolor{mainpink}{HTML}{e377c2}

\usepackage[colorlinks=true,citecolor=maingreen,linkcolor=mainorange,urlcolor=mainpink]{hyperref}

\newcommand*{\dagg}{\ensuremath{^{\dagger}}}
\DeclareMathOperator{\Tr}{tr}

\DeclareMathOperator*{\argmin}{arg\,min}

\newcommand{\pder}[2][]{\frac{\partial#1}{\partial#2}}
\newcommand{\der}[2][]{\frac{\mathrm d#1}{\mathrm d#2}}

\providecommand{\he}{\ensuremath{\hat{e}}}

\providecommand{\to}{\ensuremath{\Tilde{o}}}

\providecommand{\tx}{\ensuremath{\Tilde{x}}}

\providecommand{\calF}{\ensuremath{\mathcal{F}}}

\providecommand{\calM}{\ensuremath{\mathcal{M}}}
\providecommand{\calN}{\ensuremath{\mathcal{N}}}
\providecommand{\calO}{\ensuremath{\mathcal{O}}}

\providecommand{\calS}{\ensuremath{\mathcal{S}}}

\providecommand{\calX}{\ensuremath{\mathcal{X}}}
\providecommand{\calY}{\ensuremath{\mathcal{Y}}}

\providecommand{\sfM}{\ensuremath{\mathsf{M}}}

\providecommand{\bbC}{\ensuremath{\mathbb{C}}}

\providecommand{\bbE}{\ensuremath{\mathbb{E}}}

\providecommand{\bbI}{\ensuremath{\mathbb{I}}}

\providecommand{\bbN}{\ensuremath{\mathbb{N}}}

\providecommand{\bbR}{\ensuremath{\mathbb{R}}}

\newcommand{\grad}{\texttt{\texttt{grad}}}
\newcommand{\smoothgrad}{\texttt{smooth\_grad}}
\newcommand{\gradxinput}{\texttt{grad{$\times$}input}}
\newcommand{\sensitivity}{\texttt{sensitivity}}
\newcommand{\taylorone}{\texttt{Taylor-$1$}}
\newcommand{\taylorinf}{\texttt{Taylor-$\infty$}}
\newcommand{\SV}{\texttt{SV}}
\newcommand{\IG}{\texttt{IG}}

\newcommand{\QLRP}[1]{\texttt{QLRP$#1$}}
\newcommand{\LRP}{\texttt{LRP}}

\newcommand{\cmark}{\ding{51}}
\newcommand{\xmark}{\ding{55}}

\newtheorem{theorem}{Theorem}
\newtheorem{lemma}[theorem]{Lemma}
\newtheorem{proposition}[theorem]{Proposition}

\usepackage{times}

\newcommand{\fu}{Dahlem Center for Complex Quantum Systems, Freie Universit\"{a}t Berlin, 14195 Berlin, Germany}
\newcommand{\hzb}{Helmholtz-Zentrum Berlin f{\"u}r Materialien und Energie, 14109 Berlin, Germany}
\newcommand{\tuelec}{Department of Electrical Engineering and Computer Science, Technische Universit\"at Berlin, 14109 Berlin, Germany}
\newcommand{\bifold}{BIFOLD -- Berlin Institute for the Foundations of Learning and Data, 14109 Berlin, Germany}
\newcommand{\hhi}{Fraunhofer Heinrich Hertz Institute, 10587 Berlin, Germany}
\newcommand{\fumore}{Department of Mathematics and Computer Science, Freie Universit\"{a}t Berlin, 14195 Berlin, Germany}

\begin{document}

\title{Opportunities and limitations of explaining quantum machine learning}

\author{Elies Gil-Fuster}
\affiliation{\fu}
\affiliation{\hhi}
\author{Jonas R.\ Naujoks}
\affiliation{\hhi}

\author{Grégoire Montavon}
\affiliation{\fumore}
\affiliation{\bifold}

\author{Thomas Wiegand}
\affiliation{\hhi}
\affiliation{\tuelec}
\affiliation{\bifold}

\author{Wojciech Samek}
\affiliation{\hhi}
\affiliation{\tuelec}
\affiliation{\bifold}

\author{Jens Eisert}
\affiliation{\fu}
\affiliation{\hhi}
\affiliation{\hzb}

\begin{abstract}
    A common trait of many machine learning models is that it is often difficult to understand and explain what caused the model to produce the given output.
    While the explainability of neural networks has been an active field of research in the last years, comparably little is known for quantum machine learning models.
    Despite a few recent works analyzing some specific aspects of explainability, as of now there is no clear big picture perspective as to what can be expected from quantum learning models in terms of explainability.
    In this work, we address this issue by identifying promising research avenues in this direction and lining out the expected future results.
    We additionally propose two explanation methods designed specifically for quantum machine learning models, as first of their kind to the best of our knowledge.
    Next to our pre-view of the field, we compare both existing and novel methods to explain the predictions of quantum learning models.
    By studying explainability in quantum machine learning, we can contribute to the sustainable development of the field, preventing trust issues in the future.
\end{abstract}

\maketitle

\section{Introduction}

\begin{figure}[t]
    \centering
    \includegraphics{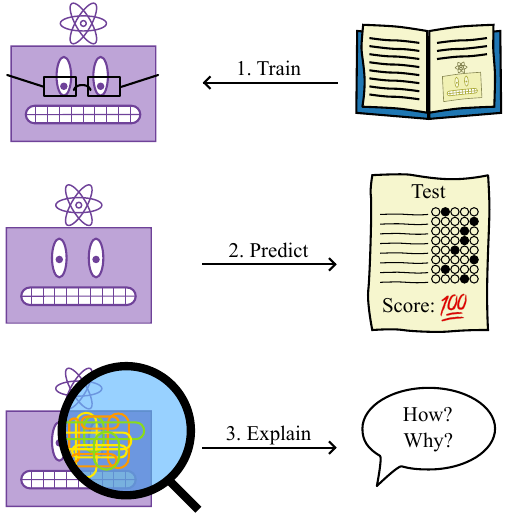}
    \caption{
        Framework of eXplainability in \emph{quantum machine learning}  (XQML) as presented in this work.
    }
    \label{fig:framework}
\end{figure}

Concomitant with the advent of enormously large databases, the striking development of available computing power, and
a rapid increase in the pace of research on data-driven approaches, systems of \emph{artificial intelligence} (AI) and \emph{machine learning} (ML) are reaching levels of performance unseen before, often outperforming human capabilities~\cite{lecun2012neural,Goodfellow2016deep}.
From chat systems passing the Turing test~\cite{nov2023putting}, to achieving super-human performance in strategy games~\cite{silver2016mastering}, and becoming ubiquitous in speech~\cite{deng2013new}, image~\cite{he2016deep}, and other pattern recognition tasks, AI algorithms have changed our world.

Given the importance of AI algorithms to our modern lives, researchers have been thinking intensely recently whether one could possibly improve such algorithms by resorting to \emph{quantum computers}.
While presently-available quantum devices are still noisy and medium-sized, the rate of development of quantum technologies prompts us to imagine a near future with large-scale quantum computers.
On some paradigmatic---albeit not yet practical---mathematically well-defined problems, quantum computers come within reach or outperform the fastest classical supercomputers available to date~\cite{arute2019quantum,hangleiter2022computational}.
Generically, a key question is still whether quantum computers can achieve better performance in learning tasks than classical algorithms, and how this would manifest in practice.
This is the core question of the aspiring field of \emph{quantum machine learning} (QML)~\cite{biamonte2017quantum, dunjko2018machine,carleo2019machine,McClean2016theory}. 
Such quantum advantages---as this state of affairs is often 
referred to---could refer to advantages in sample complexity, in 
computational complexity, or in generalization.
For highly structured problems, 
super-polynomial quantum advantages have actually already been proven~\cite{sweke2021quantum,liu2021rigorous,DensityModelling}, and it can only be a matter of time to have further evidence available to what extent quantum algorithms might possibly help dealing with real-world 
data.

One disadvantage of state-of-the-art ML models is the difficulty to understand the model's behavior as an emergent phenomenon from inspecting the individual constituent components~\cite{szegedy2014intriguing,lapuschkin2016analyzing}.
This hinders broader ML adoption in applications where human lives and 
social justice are at stake.
It is the goal of \emph{explainable AI} (XAI)~\cite{SamPIEEE21} and, 
more broadly, of reliable AI, to mitigate the potential negative impact of ML in society by unveiling the causes behind the behavior of modern learning models.
Notions of XAI provide answers to the question of what it actually means to explain a learning model.
One here asks questions of 
the type, \emph{what is the role of each of the parameters in the model?}, \emph{what pixels in this image are responsible for it to be labelled as a cat?}, 
or \emph{how can we design learning models which remain human-interpretable throughout learning?}
It is not a single method, but 
rather a portfolio of tools that capture various aspects of this question.

While this field of research is for good reason becoming well-established for classical AI algorithms~\cite{SamPIEEE21,xAIReview,ArrietaRSBTBGGM20}, the same cannot be said for quantum algorithms.
In fact, few steps have been taken so far to 
bring transparency to quantum ML~\cite{steinmuller2022explainable, lifshitz2022quantum, heese2023explaining, pira2024interpretability}.
This lack of a machinery gives rise to a grave omission, however:
All limitations of black-box classical ML (e.g.,  in terms of trustworthiness) also exist in black-box QML. Hence, it is important to find out what feature lets a quantum algorithm make a certain prediction, to ``open the 
black box'', metaphorically speaking. 

There is one major difference between the classical and quantum settings in terms of explainability.
Namely, QML research still has not realized the main goal in the field: to solve practically relevant problems faster or otherwise better than its classical counterpart.
This represents a timely opportunity for developing \emph{explainable quantum machine learning} (XQML) techniques: we have the chance to design QML models that are inherently explainable from the get-go.
Conversely, the development of high-performing models in classical ML pre-dates the development of the field of XAI, resulting in tensions between accuracy and interpretability~\cite{dziugaite2020enforcinginterpretabilitystatisticalimpacts}.
The cartoon in Fig.~\ref{fig:framework} introduces QML as a three-step process: after training and observing good tests results, we need to explain the behavior of the quantum learning agent.

In this work, we aim at presenting first steps towards establishing such a framework.
Our main contribution is then two-fold: 
\begin{itemize}
\item We offer a broad perspective of the incumbent field of explainability for QML: We do so by introducing and reviewing methods of explainability and what is specifically different for QML.
\item We also propose two concrete explanation techniques designed specifically for present-day QML models, based on \emph{parametrized quantum circuits} (PQCs).
\end{itemize}

\section{Preliminaries}\label{s:prelim}

    In this section, we attempt to bridge the fields of (classical) explainable AI and QML.
    We first lay out the formalism of parametrized quantum circuits, which are the workhorse in most of present-day QML.
    Our introduction focuses on the essential differences between classical and quantum ML that we inherit from the principles of quantum mechanics.
    This way, we do not dwell extensively on the QML practitioner's point of view.
    Next we present explainability holistically, singling out the types of questions one may ask in order to explain the behavior of a learning model.
    Here we start very broad and become progressively narrow until we reach the so-called post-hoc local explanations, which are the area of explainability we concentrate later on, in Section~\ref{s:methods}.

\subsection{Quantum function families for machine learning}\label{ss:QML}

    In this work we restrict ourselves to learning tasks with classical data.
    The predominant approach to QML for classical data relies on a view of quantum computing that is reminiscent of probabilistic classical computers.
    A common reading of quantum mechanics is as a generalized theory of probability, where qubits play the role of probabilistic bits. We keep the discussion here at a high level, so that an understanding of physics is not required to follow the exposition.
    
    A quantum state composed of many qubits can be thought of as a generalized probability distribution over bit-strings.
    Probability distributions over $n$-bit-strings can be represented as diagonal matrices
    \begin{align}
        D_p &= \sum_{i\in\{0,1\}^n} p_i \lvert i\rangle\!\langle i\rvert,
    \end{align}
    where the notation $\lvert i\rangle\!\langle j\rvert$ is the Euclidean basis of the vector space of matrices.
    Given that its diagonal is real, it follows that the matrix is Hermitian $D_p=D_p\dagg$.
    Moreover it is also \emph{positive semi-definite} (PSD) $D_p\geq0$, since all entries are non-negative. 
    Finally, due to normalization of any probability distribution $p$, it has unit trace $\Tr(D_p)=1$.
    The step from classical probability distributions to \emph{quantum states} is to just remove the requirement that the matrix be diagonal.
    This way, a quantum state $\rho$ over $n$ qubits is represented by a $2^n\times2^n$-dimensional matrix (same as $D_p$), which is Hermitian, PSD, and has unit trace.
    The off-diagonal entries of the matrices are in general complex valued, and we refer to them as \emph{coherences}.
    Classical probability distributions correspond to the special cases of quantum states with diagonal matrices.
    Then, a quantum state stores statistical information on its diagonal.
    We say it stores statistical information because the diagonal always represents a probability distribution, and we can sample this distribution by physically performing a measurement in the computational basis.
    The information in the off-diagonal corresponds to other correlations among different bit-strings that would influence further transformations of the state.

    In simple terms, a quantum circuit is a sequence of operations each of which takes a quantum state as input and returns another quantum state as output.
    Transformations of quantum states can be thought of as analogues of transformations of probability distributions.
    For example, a doubly stochastic matrix $S$ is a linear transformation that maps probability distributions onto other probability distributions $p'=Sp$.
    Similarly, the linear transformations that map quantum states onto quantum states are called \emph{quantum channels}.
    In general, these are completely positive, trace preserving maps, and, in particular, the action of any stochastic matrix can be seen as a quantum channel acting on the corresponding matrix $D_p\mapsto D_{p'}$ as a convex combination of permutations.
    In general quantum channels can correspond to exotic transformations of states, but in practice we mainly consider those that we can hope to implement on quantum hardware.
    Therefore, throughout this work, we only consider unitary transformations, so quantum channels that act as $\rho\mapsto U\rho U^\dagger$, where $U$ is a unitary matrix, fulfilling $U\dagg U=\bbI$ by definition.
    
    A common parametrization of the set of all unitary matrices is via a collection of rotation angles $\vartheta=(\vartheta_1,\ldots,\vartheta_M)\in[0,2\pi)^M$.
    Then, we usually talk about parametrized transformations $\rho\mapsto \rho'(\vartheta)=U(\vartheta)\rho U\dagg(\vartheta)$.
    Unitary transformations are also referred to as \emph{quantum gates} in the context of quantum computing, as they are complex-valued generalizations of invertible logic gates.

    Finally, a real-valued function that takes  bit-strings $b\colon\{0,1\}^n\to\bbR$ as input provides a natural way to extract information from a probability distribution via its expectation value $\langle b\rangle_p = \sum_{i\in\{0,1\}^n} b(i)p_i\in\bbR$.
    We call 
    \begin{align}
        B=\sum_{i\in\{0,1\}^n} b(i)\lvert i\rangle\!\langle i\rvert,
    \end{align}
    which is again a real-valued, diagonal, Hermitian matrix (just not necessarily unit trace nor PSD).
    Using this representation, the expectation value adopts the form of the Hilbert-Schmidt norm of the matrices $\langle b\rangle_p =\Tr\{D_pB\}$.
    The way via which we 
    retrieve information 
    from a quantum state $\rho$ is fully analogous: Given an observable $\calM$ (a complex-valued Hermitian matrix), we can compute its expectation value relative to the state as $\langle\calM\rangle_\rho=\Tr\{\rho\calM\}$, which is always real-valued due to both $\rho$ and $\calM$ being Hermitian.

    A typical quantum computation then takes the form of a three-step process:
    \begin{enumerate}
        \item \label{list:initial_state} Start from a fixed (easy-to-prepare) initial state $\rho_0$.
        \item \label{list:parametrized_gates} Apply a \emph{quantum circuit} specified as a sequence of (parametrized) unitary gates $U_1(\vartheta_1), \ldots, U_L(\vartheta_L)$, each of which acting as $\rho_j = U_j(\vartheta_j) \rho_{j-1} U_j\dagg(\vartheta_j)$, and resulting in the final state
        $\rho_L(\vartheta_1,\ldots,\vartheta_L)$.
        \item \label{list:observable} Estimate the expectation value of a fixed (easy-to-measure) observable\footnote{
            As a technical aside, talking about an evolved state with a fixed observable is called the \emph{Schr\"odinger picture}.
            We could have analogously talked about a fixed state and an evolved observable, placing ourselves in the \emph{Heisenberg picture}.
            Considering the space in-between, where both the state and the observable are evolved via some gates each is sometimes called the \emph{interaction picture}, and we exploit it later.
        } $\calM_0$, obtaining the real value $\Tr\{\rho_L(\vartheta_1,\ldots,\vartheta_L)\calM_0\}$.
    \end{enumerate}
    The real-valued result of this quantum computation is
    \begin{align}
        f(\rho_0,(U_j(\vartheta_j))_j,\calM_0) &=\Tr\{\rho_L(\vartheta_1,\ldots,\vartheta_L)\calM_0\}.
    \end{align}
    We call a \emph{parametrized quantum circuit} (PQC) a fixed choice of initial state $\rho_0$, parametrized gates $(U_j)_j$ (not their parameters), and final observable $\calM_0$.
    Each PQC gives rise to a real function 
    \begin{align}
        (\vartheta_j)_j &\mapsto g_{\rho_0, (U_j)_j, \calM_0}((\vartheta_j)_j) = f(\rho_0,(U_j(\vartheta_j))_j,\calM_0)
    \end{align}
    of the parameters.
    We abuse notation slightly and still use $f$ instead of $g_{\rho_0, (U_j)_j, \calM_0}$ when the specific PQC is irrelevant or clear by context.
    We refer to functions arising from PQCs in this way as \emph{quantum functions}.
    Appendix~\ref{a:QMLvsPML} contains a brief description of an analogous family of classical functions defined from parametrized probability distributions.

    To perform machine learning it is not enough to have a single function, though.
    If we call $\calX$ the data domain, and $\calY$ the label co-domain, we must furbish a family of functions $\calF\subseteq\calY^\calX$, to be used as the hypothesis class.
    One straightforward approach would be to, given a PQC, consider a set of possible initial states $\{\rho(x)\,|\, x\in\calX\}$ instead of a single one, or equivalently a data-dependent observable $\calM(x)$.
    This has been a popular approach since the start of PQC-based QML, and it knocks on the door of the fundamental question \enquote{what is the best way to upload classical data into a quantum computation?}.
    In practice, having a data-dependent state amounts to keeping the structure introduced above, and just fixing the rotation angles of the first gates to be equal to the input data.
    One commonly distinguishes between an \emph{encoding} gate $E(x)$, and the \emph{trainable} (or variational) gates $U_j(\vartheta_j)_j$.
    With this notation the data-dependent state is $\rho(x) = E(x)\rho_0 E\dagg(x)$.
    We refer to $V_L(\vartheta)= U_L(\vartheta_L)\ldots U_1(\vartheta_1)$ as the quantum gate containing all the trainable gates, and we call $\Theta$ the set of all possible parameter values $\vartheta\in\Theta$.
    In this prescription, a fixed PQC (where the encoding gate is also fixed as of step~\ref{list:parametrized_gates} above) gives rise to a function 
    \begin{align}
        f(x, \vartheta) &= \Tr\{V(\vartheta)\rho(x)V\dagg(\vartheta)\calM_0\}
    \end{align}
    of both the input $x$ and the variational parameters $\vartheta$.
    The division between encoding and trainable gates allows us to consider the set of all functions 
    \begin{align}
        \calF &\coloneqq \{ x\mapsto f(x, \vartheta) \, |\, \vartheta\in\Theta\}
    \end{align}
    that arise from different parameter values for a given PQC.
    In general terms, these are the quantum function families 
    that we consider for QML.
    Each choice of PQC amounts to a different hypothesis 
    class in the exact same way as different neural network (NN) architectures result in different ML models.
    Specifically, the milestone works~\cite{GilVidal2020,Schuld2021effect} showed that PQCs give rise to generalized trigonometric polynomials.
    That is, the parametrized function families realized by such circuits are always of the form
    \begin{align}
        f_\vartheta(x) &= \sum_{\omega\in\Omega}
        \left( 
        a_\omega(\vartheta) \cos\langle\omega,x\rangle + b_\omega(\vartheta)\sin\langle\omega,x\rangle
        \right)
        .
    \end{align}
    Here, the values of the coefficients $(a_\omega(\vartheta),b_\omega(\vartheta))_\omega$ depend only on the trainable parameters $\vartheta$ of the PQC, in a computationally intricate way.
    Further, the frequency spectrum $\Omega$ depends only on which are the encoding gates.
    Here $\langle\,\cdot\,,\cdot\,\rangle$ denotes the usual Euclidean inner product of vectors.

    The study of PQC-based QML models, which we consider in this work,
    has been pursued within the quest for practical quantum advantages in
    meaningful learning tasks, as potential quantum analogues to neural networks.
    Although both approaches rely on parametrized building blocks and share a layered structure, the principles of quantum mechanics prevent much deeper similarity:
    \begin{itemize}
        \item
            The \emph{no-cloning-theorem} -- that shows that unknown quantum states cannot be copied so that 
            both copies are statistically indistinguishable 
            from the input -- prevents parametrized gates from fanning out the same way as artificial neurons do.
            Namely, the signal a classical neuron transmits can be copied freely and used as input for multiple neurons in the next layer.
            For PQCs, a leading framework to replicate this feature is data re-uploading~\cite{perez2020data}\footnote{
                In data re-uploading, the input is uploaded several times to introduce a redundancy which is present by default in NNs.
            }.
        \item
            In general, the matrix representation of quantum states requires an exponential dimension, which means we cannot efficiently 
            store the intermediate state of the computer 
            at each step.
        \item
            Storing intermediate information (even if it is not the entire state) necessarily destroys the computation, which prevents information re-use and the existence of fast training algorithms like back-propagation~\cite{abbas2024quantum}.
        \item
            In practice, we estimate 
            expectation values via samples from the quantum computer.
            This means that our read-out precision is only polynomial in time, not exponential, which results in distinguishability problems and can mask the signal inside the statistical noise.
            This problem is discussed also in the context of concentration phenomena and vanishing gradients~\cite{mcclean2018barren,larocca2024review}.
        \item 
            The application of unitary gates is a linear transformation (non-linear transformations are forbidden in quantum mechanics).
            The only possible source of non-linearity comes from the parametrization of the gates $\alpha\mapsto U(\alpha)$, and not from the nested composition of information-processing layers.
            This allows us to understand the expressivity of 
            PQC-based QML models as trigonometric series~\cite{Schuld2021effect}.
        \item 
            The connectivity graph of the underlying quantum computing platform can severely restrict the unitary gates we can implement in practice.
            In classical NNs sparsely-connected layers are used among other reasons because doing so has been empirically found to boost performance.
            Conversely, for PQC designs to remain realistic we should actually only consider unitary gates that act on at most $2$ qubits at once, and ideally they should be adjacent.
            In spite of the field of quantum compilation being by now mature, the hardware limitations in the short- and mid-term must be kept in mind.
    \end{itemize}
    There exist ways to circumvent some of these \enquote{problems} by considering only restricted families of quantum circuits.
    Yet, these characteristics are fundamental to quantum mechanics: by aiming to remove them even partially, we may inadvertently fall in the regime of \emph{efficient classical simulability}, in which there can be no quantum advantage~\cite{Bravyi2016improved, Tang2019,schreiber2022classical, landman2023classically,Sweke2023Potential,shin2024dequantizing, rudolph2023classical,cerezo2023does,dias2024classical}.
    Indeed, we must take into account the classical and quantum time and computational complexity of both the QML models themselves and the explainability methods we consider later.

\subsection{Explainability} \label{ss:explainability}
    
    In addition to building models that perform accurate predictions, it is often desirable to ensure they operate with some level of transparency, in particular, ensuring that their decisions can be explained in a way that is understandable for a human. This aspect has been treated extensively within the field of \emph{explainable AI}  (XAI)~\cite{GunningA19,ArrietaRSBTBGGM20,SamPIEEE21}. A central example motivating the usefulness of explainability is detecting when a learning model is performing well \enquote{for the wrong reasons}, the so-called clever Hans effect~\cite{lapuschkin_unmasking_2019, unsupervised-ch} or \enquote{shortcut learning}~\cite{geirhos_shortcut_2020}. A flawed prediction strategy, e.g., 
    one relying on spuriously correlated artifacts, may stop working on new data, exposing the environment in which the model operates to great risks (see, e.g., Refs.\  \cite{DegraveJL21,unsupervised-ch}). Explainable AI techniques can precisely pinpoint these flawed decision strategies, and also provides a starting point to produce more robust models by aligning the model with human expert knowledge \cite{RossHD17,AndersWNSML22,ChormaiHMM24}. The same Explainable AI methods can also enable a reverse alignment where a ML model trained on vast amounts of data can serve as an exploration tool for humans, enabling for example, the generation of new insights or scientific hypotheses \cite{Krenn2022,SchramowskiSTBH20}.

    The field of XAI encompasses a multitude of diverse techniques addressing different levels of interpretability as well as models and data to which they apply. Explanation techniques can be organized along several conceptual axes:
    
    \paragraph{Global versus local explanations.}
    A local explanation focuses 
    on unraveling a model's prediction proper to a single data point. Such explanations 
    then result in properties 
    defined for the input pattern, like the relevance of each of its dimensions. A global explanation on the other hand is one that aims at elucidating a given model's overarching reasoning strategy.
    This includes recognizing prototypical behavior of the model without inherently providing an interface to understand a given single data point.
    A representative global explanation method is \emph{spectral relevance analysis}~\cite{lapuschkin_unmasking_2019,Dreyer2024CVPR}, which looks at the behavior of a model on a dataset by aggregating the output of a large collection of local explanations, e.g.,  into clusters that summarize the complex multifaceted decision strategy of the model.
    
    \paragraph{Feature-wise versus concept-wise explanations.} Another way of distinguishing between explanation techniques is with respect to what are the actual units of explanation. While classical explanations typically explain decision of the model in terms of their input features (e.g.,  pixels \cite{bach2015pixel}) or combinations of them (see, e.g., Ref.\  \cite{SchnakeHigherOrderExplanationsGraph2022}), a more recent line of research has sought to support the explanation by more abstract features, so-called concepts, which can be, for example, activations in intermediate layers, or linear transformations of those activations \cite{KimWGCWVS18,achtibat_attribution_2023,ChormaiHMM24,VielhabenLMS24}. Concept-based explanations, because of their intrinsic lower dimensionality, are also more amenable to a global (dataset-wide) analysis of the model.
            
\paragraph{Ante-hoc versus post-hoc explanations.} 
An additional line of disambiguation of explainability techniques concerns the \emph{stage} at which explanations are generated.
The field of ante-hoc explainability deals with the design and study of learning models that are intrinsically interpretable, specifically, they embed units of activation that are directly inspectable and that carry 
intrinsic meaning.
    This encompasses generalized additive models~\cite{lou2012intelligible}, decision trees, or neural networks with specific topologies (e.g., global average pooling layers \cite{ZhouKLOT16} or attention layers~\cite{vaswani2017attention}).
    On a related note, data-driven physical theories, e.g., implemented via neural ordinary differential equations \cite{ChenRBD18}, constitute another instance of ante-hoc explanations: For example, before we fit the viscosity parameter to the measured data, we already know that the role of that parameter will be the viscosity coefficient, and thus the viscosity model is inherently interpretable. Enforcing such a structure in a learning model limits its capabilities compared to black-box models. This concept is known as the \emph{performance-explainability trade-off} \cite{ArrietaRSBTBGGM20}. 
    Nowadays, the most powerful ML models, especially deep neural networks, fall under the category of opaque or black-box models: they are powerful but their inner workings are not easily interpretable. In this case, post-hoc explainability is used. A post-hoc explanation is extracted from an already trained model.
    
    \paragraph{Attribution versus synthesis.}
    We refer to an explanation as \emph{attribution} when the focus is on which features are measurably responsible for the predicted output. Many attribution methods (see, e.g., Refs.\  \cite{bach2015pixel,sundararajan2017axiomatic,strumbelj2010efficient}) are designed to assign a score to a particular feature that can be interpreted as the share by which that feature contributes to the output of the model.
    In contrast, \emph{synthesis} methods rely on generating new data points which carry insight, either on their own or in relation to the original data point.
    For example, given a correctly classified input, one can generate a counterfactual~\cite{dombrowski2023diffeomorphic}, specifically, the smallest natural perturbation to that input that causes a change in the predicted class. One could also, without relying on any particular data point, synthesize a so-called class prototype that     maximizes the score for a given class \cite{nguyen2016synthesizing}. This synthesis approach to explanation is also the basis of the \emph{DeepDream} technique \cite{google2015deepdream} (which was adapted to quantum circuits in Ref.\ \cite{lifshitz2022quantum}).

\subsection{Post-hoc local attribution methods}\label{ss:posthoc}

    We focus specifically on \emph{post-hoc local attribution} methods, in which 
    explanations consist of assigning a relevance score to each input dimension of a given pattern.
    The use cases we present later all belong to this class of explanations.
    This 
    presupposes that the input features or relevance attributions thereof are 
    themselves interpretable 
    to the end user.
    In a classification setting, an ML model typically outputs a score for each of the $K$ classes $h(x)=(h_1(x), \ldots, h_K(x))$.
    To explain a single input, 
    though, we call $f$ the score for the class that we would like to investigate. Often, one chooses the class that corresponds to the true label for each input.
    So, explicitly, the human-interpretable property we explain in this case is \enquote{how much does each input dimension contribute to the correct class?}
    For local attribution methods, the final explanation is sometimes called a 
    \emph{relevance heat-map} (as for image-recognition tasks that is what it looks like). Many attribution techniques have been proposed, which we organize in the following into three types of approaches.

\paragraph{Occlusion-based approaches.}
A first approach to attribution consists of scoring input features according to the effect of adding or removing them on the output of the model. This approach is instantiated by the method of \emph{Shapley values} (\SV{})~\cite{shapley1953value}, which originates from coalition game theory. In its original formulation, the method addresses the question of assigning the individual contribution of a given set of players to a global payoff function. The method comes with strong theoretical foundation (it has a unique axiomatic characterization \cite{shapley1953value}) and has been extended to attribute a machine learning model (see, e.g.,
Refs.\ \cite{strumbelj2010efficient,lundbergUnifiedApproachInterpreting2017}), where the players forming the coalition are the input features, and the payoff to be distributed correspond to the output of the model to be explained. 

    In this work, we consider the \emph{baseline Shapley values} approach \cite{sundararajanManyShapleyValues2020}, where such features are set to a baseline $\tx$.
    Starting from the original datum $x$, and given the baseline $\tx$, one systematically replaces each possible combination of input components of $x$ by those of $\tx$.
    Such an explanation is formalized as
    \begin{align}
        E_i(x) &= \sum_{\calS\subseteq([d]\setminus{\{i\}})} \frac{\lvert\calS\rvert (d-1-\lvert\calS\rvert)!}{d!} \left( f(x_{\calS\cup\{i\}}) - f(x_{\calS})\right),
        \label{eq:shapley}
    \end{align}
    where $[d]=\{1,\ldots,d\}$, and for any set $A\subseteq[d]$, we use $x_A$ to denote the modification of $x$ whose component indices in $A$ are left untouched.
    So the $i$th component of $x_A$ is $x_i$ if $i\in A$ and $\tx_i$ otherwise.
    Computing Shapley values exactly requires computational resources that scale combinatorially, so in practice they are most often replaced by a sampling approximation (like \emph{Shapley value sampling}~\cite{strumbelj2010efficient}). There exist variations for 
    specific model classes where the computation of Shapley values 
    scales polynomially \cite{CastroPolynomialCalculationShapley2009, AnconaExplainingDeepNeural2019a}.

\paragraph{Gradient-based approaches.}\label{sss:gradient-based}
    An alternative to evaluating the effect of removing each features individually consists of relying on the gradient of the model, which can be interpreted as a local effect of each individual feature.
    Gradients can usually be efficiently computed, for instance via the backpropagation algorithm for neural networks~\cite{Rumelhart1986}.
    A particular way of turning a gradient into an attribution is to use it to perform a first-order Taylor expansion of the function $f$ around some baseline $\tx$, evaluate it at the point of interest $x$, and identifying the first-order terms bound to each input dimension, i.e.,
    \begin{align}
        E_i(x) &=\partial_if(\tx)(x_i-\tx_i).
        \label{eq:taylor}
    \end{align}
    For the method to achieve a meaningful attribution of the model output, the first-order terms must dominate over the higher-order terms locally, which essentially requires the function to be locally linear. This often requires to adapt the root point to the actual data point, i.e.\ define a function $\tx(x)$. We refer to this explanation method as \taylorone{}. 
    
    Another gradient-based approach to compute explanations is \emph{integrated gradient} (\IG{}) \cite{sundararajan2017axiomatic}, where rather than evaluating the gradient only once at $\tx$, we integrate it over a path connecting $\tx$ and $x$. If choosing the path to be a straight line, one gets the integral $\int_{[0,1]} \partial_i f(tx+(1-t)\tx) \mathrm dt$, which we can substitute to the gradient in the equation above. In practice, the integral is approximated by a discretization. The main advantage of the Integrated Gradients approach over Taylor expansion is its ability to produce meaningful explanation even in the presence of strong non-linearities. This broader applicability of Integrated Gradients comes however at the cost of a higher computational cost.

    Another simple yet commonly used gradient-based explanation technique is \gradxinput{}: $E_i(x) = \partial_if(x) x_i$.
    For specific functions, like homogeneous linear functions, \gradxinput{} can be shown to coincide with the other gradient-based approaches.

\paragraph{Propagation-based approaches.}\label{sss:divide}
    This last family of methods differ from the ones above by viewing the prediction not as a function but as the output of a computational graph. It then seeks to attribute the prediction to the input features by reverse propagating in this graph.
    Such an approach leads to explanations that are specific to the computational graph.
    This is in contrast to explanations that only require black-box access to the model, which we refer to as \emph{moodel-agnostic} explanations.
    The approach is exemplified by the \emph{layer-wise relevance propagation} (LRP)~\cite{bach2015pixel} approach. Originally proposed in the context of convolutional neural networks and bag-of-words models, LRP was later 
    extended to a broad range of models (e.g.,  unsupervised models \cite{MontavonKSM20} or transformers \cite{AliSEMMW22}). In its basic form, LRP assumes that the function $f(x)$ of interest can be expressed as a sequence of mappings, e.g.\ in a deep neural network the mapping onto its successive layers:
    \begin{align}
        \big\{\,a_{l} \mapsto a_{l+1}\,\big\}_{l=1}^{L-1},
    \end{align}
    where $a_1$ and $a_L$ are the input and output of the model to explain.
    Explanation proceeds by iteratively attributing the model output to the diverse layers of the network, starting with the top layers and moving backward towards the input. This layer-wise attribution is implemented by purposely defined propagation rules. Denoting by $R_l$ and $R_{l+1}$ the attribution on the representation $a_l$ and $a_{l+1}$, respectively, the LRP propagation rule can be generically represented by the collection of maps:
    \begin{align}
        \big\{\,(a_l, R_{l+1}) \mapsto R_l\,\big\}_{l=1}^{L-1}
        \label{eq:lrp}
    \end{align}
    at each layer, defining the redistribution strategy. The process starts with setting $R_L = f(x)$, applying the maps above at each layer in reverse order, and retrieving the vector $R_1$, which forms the desired explanation $E(x)$.
    We illustrate this process with pictorial diagrams in Appendix~\ref{a:LRP}. The main advantage of the LRP approach over occlusion and gradient-based methods discussed above is its higher flexibility, allowing to finely control the forming of the explanation at each layer, addressing its specific nonlinearities, so that a robust explanation can be extracted in only a single forward-backward pass.
    For a more detailed introduction of the LRP technique, including explicit propagation rules, we refer the reader to Ref.\ \cite{montavonLayerWiseRelevancePropagation2019}.

    \begin{table*}[t]
        \centering
        \caption{
            Evaluation of surveyed local attribution methods in 
            terms of formal (data-set unspecific) criteria of explanation 
            correctness and other more general usefulness criteria.
        }
        \label{tab:black-box}
        \begin{tabular}{lcccc} \hline\hline
          \multirow{4}{*}{Local Attribution Method} &  \multicolumn{3}{c}{Explanation correctness} & \\\cmidrule(lr){2-4} 
           & \multirow{2}{*}{Conservative?} & \multirow{2}{*}{Continuous?} & Implementation & Computational \\
           &   &  & invariant? & efficiency? \\\hline\hline
          \SV{} & \cmark & \cmark & \cmark & \xmark \\ 
          \taylorone{} & \cmark & \xmark & \cmark & \cmark \\
          \IG{} & \cmark & \cmark & \cmark & \xmark \\
          \gradxinput{} & \xmark & \xmark & \cmark & \cmark \\
          \LRP{} & \cmark & \cmark & \xmark & \cmark \\\hline\hline
        \end{tabular}
    \end{table*}

\subsection{Evaluating explanation techniques}

Having introduced various approaches to explainability, as well as a variety of technique to achieve post-hoc local attributions, one question remains, namely which explanation technique is most appropriate in a given context. The first and most important criterion to be considered is \emph{explanation correctness}, that is, if the explanation faithfully reflects the underlying strategy the model has used to achieve its prediction. Unlike the prediction function itself, there are typically no ground-truth explanations to be evaluated against. The strategy used by the model to be explained is by definition unknown. Explanation correctness can thus only be assessed via \emph{indirect} means.

A first property an explanation should satisfy is \emph{conservation}: Conservation requires that, if one is to interpret the score assigned to each input features to be the share they contribute to the function output, one should have that $\sum_i E_i(x)\approx f(x)$. Another property an explanation should satisfy is \emph{continuity}, requiring that similar inputs fed into a continuous prediction function should receive similar explanations. A further property to be satisfied is \emph{implementation invariance}, requiring that for two different models implementing the same function $f$, the explanation for a given data point should be the same or similar. It is noteworthy that the fulfillment of these properties can be assessed to some extent by a direct inspection of the explanation technique without requiring experimentation with a real model trained on real data. For example, one can demonstrate that \SV{} gives rise to continuous explanations, being a weighted sum of continuous functions. Conversely, gradient-based method may produce discontinuous explanation if we consider that the gradient of a continuous prediction function may be discontinuous. An analysis of presented methods in light of these three criteria is given in Table \ref{tab:black-box} (left). These formal properties of an explanation may however be insufficient to fully characterize its correctness. For example, a simple explanation technique that would uniformly redistribute the function output to the input features would fill all correctness criteria above, yet fail to highlight features that contribute the most to the prediction.

A common way of assessing the ability of an explanation consists to retrieve the most relevant features from all input features is the \enquote{pixel-flipping} evaluation procedure \cite{bach2015pixel,SamekBMLM17}. Pixel-flipping removes features from most to least relevant and keep track of how quickly the output of the model decreases. The faster the decrease the better the explanation. An alternate way of testing a explanation method's ability to extract correct features is to devise synthetic datasets and models for which the ground-truth explanation is known and test the degree of consistency the explanation technique offer with respect to these ground-truth explanations~\cite{Arras2022clevrxai}.
We propose specific such metrics in Appendix~\ref{a:evaluation}, which we use in our numerical experiments in Section~\ref{s:results}.
Acknowledging the multifaceted nature of the task of evaluating explanations, works like Ref.~\cite{hedstromQuantusExplainableAI2023} have contributed series of explanation correctness tests together with their software implementation, which can be used for benchmarking purposes.

Besides explanation correctness, further aspects need to be considered for an explanation to fulfill its practical purpose. Listing desirable properties for human-interpretable explanations is an enterprise that pre-dates modern machine learning: remarkable work was already laid out in Ref.~\cite{swartout1993explanation}. Among the list of possible desiderata, one can mention \emph{computational efficiency} which requires that the explanations are computable within ideally one single evaluation of the machine learning model (the computational efficiency of the presented method is listed in Table~\ref{tab:black-box}. Other desiderata such as human interpretability or sufficiency of the explanation pertain more to the class of explanation techniques, for example, whether the explanation is returned in terms of input features or in terms of more abstract human-readable features, or whether the scoring of those feature (feature attribution) provides sufficient information for delivering useful insights and make the explanation actionable.

\subsection{Explaining PQC-based QML models}

    \begin{table*}[t]
        \centering
        \caption{Summary of main differences between NNs and QML models, and their effects on explainability.}
        \label{tab:XAIvsXQML}
        \begin{tabular}{p{.18\linewidth}|p{.2\linewidth}|p{.2\linewidth}|p{.3\linewidth}}\hline\hline
             & Neural Networks & PQC-based QML models & Effect on explainability \\\hline\hline
            Intermediate information & Stored by default & No cloning theorem and exponential classical memory requirements & Model-specific explanations using intermediate information are more difficult to find for QML models. \\\hline
            Available precision & Exponential & Polynomial & Limited available quantum hypothesis families. \\\hline
            Linearity & Nested composition of non-linear activation functions & Non-linear parametrization of high-dimensional linear operations & The linearity of quantum mechanics can result in straightforward explainability. \\\hline
            Computational graph & Usual structural depiction & Not usual structural depiction & In general not possible to directly port model-specific explanations from NNs to QML. For instance some divide-and-explain ideas like common forms of LRP cannot be immediately applied. \\\hline
            Continuity of gradients & Piece-wise continuous & Analytic everywhere & Continuity of explanations is less problematic for QML than for NNs. \\\hline\hline
        \end{tabular}
    \end{table*}

    We close this section by briefly discussing whether and which explanation techniques can be carried over directly into QML.
    We start by revisiting the list of fundamental limitations inherited from quantum mechanics from Section~\ref{ss:QML} and, for illustrative purposes, how they relate to the models and properties introduced in Sections~\ref{ss:explainability} and~\ref{ss:posthoc}.
    We summarize the main points of this section in Table~\ref{tab:XAIvsXQML}.

    The no-cloning-theorem mostly limits the QML models themselves, and not so much what can be explained about them. Just every time the model is evaluated, a the quantum circuit needs to be ran from scratch, without storing any intermediate information.
    The exponential dimension of the Hilbert space of quantum states and observables indicates that naive uses of intermediate information may be possible, but at the cost of exponential memory requirements and runtime, which present a problem for scalability.
    The problem of storing intermediate information exactly on classical memory can be alleviated by considering quantum-time-efficient approximate representations, for which there are several proposals~\cite{Shadows,RevModPhys.93.045003,AreaReview,PhysRevA.57.127,Sergi,rudolph2023classical}.
    The finite-shot noise is another relevant practical problem: we are only able to evaluate quantum functions (expectation values of PQCs) to precision polynomial in the total runtime, not exponential like in classical computers.
    The precision we are able to achieve in estimating the model in turn limits the precision of gradients (this is an important current obstacle for training QML models~\cite{mcclean2018barren,larocca2024review,chinzei2024tradeoff,gilfuster2024relation}).
    Simply put, we are unable to distinguish between exponentially small values and zero.
    Therefore, all XQML methods that require evaluating gradients might be rendered inapplicable for the common cases where typical gradient magnitudes are exponentially small.
    Next, the fact that the only source of non-linearity in PQCs is the parametrization of the unitary gates tells us that we shall often be able to exploit the linear nature of quantum computations to our advantage, since explaining linear models is relatively straightforward.
    All in all, the main factors that play a role in the explainability of QML models are the storage of intermediate information, the finite precision in estimating hypotheses and their gradients, and the fundamentally linear nature of quantum computation.

    Further differences arise when comparing 
    specifically QML models to NNs from the point of view of explainability.
    One key fact is that, even though the graphical depictions of PQCs often resemble (at least structurally) those of NNs, in truth these are objects of essentially different nature.
    The usual sketches of NNs as subsequent layers of neurons and their connections are actually the computational graph of the model itself.
    Conversely, the usual sketches of PQCs as subsequent layers of unitary gates applied to several qubits capture the step-by-step physical implementation of these circuits in a laboratory, but not the computational graph of a classical representation of the same process.
    This difference in computational graphs may already be enough to point out that model-specific techniques that have been derived for NNs shall not be in general directly applicable to PQCs.
    
    An operational difference between QML and NNs results from the different sources of non-linearity.
    Networks that use \emph{rectified linear units} (ReLU) give rise to functions that are only \emph{piece-wise} analytic (their gradients are only piece-wise continuous), which can give rise to so-called \enquote{shattered gradients} for points close to the non-analytic ones.
    Conversely, QML models are sums of trigonometric functions with bounded frequency, which gives rise to functions that are analytic everywhere.
    The intuition behind the importance of explanation continuity is that small perturbations of the input should lead to small perturbations of the explanation.
    The fact that NNs give rise to non-analytic functions results in explanations not being continuous by default.
    Conversely, all XQML methods considered so far are always continuous because the labeling functions themselves are analytic.

    In extending explanation methods to quantum models, we can also discuss the expected efficiency considerations.
    In principle, black-box explanations that only require evaluating the hypotheses should not incur efficiency issues, the total quantum complexity should be the total number of calls to the quantum model times the complexity of evaluating the model.
    Reasonably speaking, we shall only attempt to produce explanations for quantum circuits which are efficient to evaluate in the first place, so the complexity of producing black-box explanations for quantum models should also be computationally efficient.
    It should also not be a problem if the black-box model requires evaluating gradients.
    Issues might arise from the required precision, as discussed above, but that is a different notion of efficiency.
    For model-specific explanations we further use information about the structure of the model.
    If the used information is only the succinct description of the circuit, that should also not be problematic.
    Nevertheless, if the information required to evaluate the model-specific explanation involves storing intermediate information, that might incur exponential classical space overheads.
    Below we propose two novel model-specific explanations that illustrate precisely this distinction: we provide first a black-box explanation that is quantum efficient, and second an explanation that requires storing intermediate states in classical memory, which uses intractable amounts of space.
    As discussed, there could exist model-specific XQML methods which are efficient with a smart strategy using the same quantum circuit but with mid-circuit measurements.
    Such strategies might result in explanations that display a quantum-classical separation, if one can prove that the same explanation could not be realized without access to the quantum circuit.

\subsection{Related work on XQML}

    So far, a few incursions have already been performed toward XQML.
    For instance Ref.~\cite{steinmuller2022explainable} performed a numerical study for simple PQCs using explanations like \IG{} and \SV{} and recorded the effects of noise in their explanation performance.
    Also, Ref.~\cite{heese2023explaining} extended the definition of Shapley values to be well-defined for functions defined as expectation values of random variables and then proposed several applications of the resulting \SV{}-like explanations.
    The authors of this work did not focus only on QML applications, like local attribution methods, but also they used the mindset of coalition games to explain other properties of PQCs, like the effect of each individual trainable gate in the circuit.
    Finally, Refs.~\cite{lifshitz2022quantum,pira2024interpretability} proposed approaches to interpretability (related to explainability).
    In Ref.~\cite{lifshitz2022quantum} the technique of \emph{deep dreaming} is used to improve the design of quantum circuits and to extract human-interpretable properties.
    Improving on a popular method consisting on building a local interpretable surrogate of the model called LIME~\cite{ribeiro2016why}, Ref.~\cite{pira2024interpretability} develops a method in which PQCs not only produce a prediction, but also a confidence value for it.
        
\section{Novel XQML approaches}\label{s:methods}

    Here we propose two novel explanation techniques designed specifically for PQC-based QML models.
    To the best of our knowledge, these are the first model-specific explanations for quantum learning models.
    The first of the two, which we call \taylorinf{} is a black-box explanation that exploits the Fourier picture of PQCs~\cite{Schuld2021effect}.
    The second, which we call \emph{quantum layerwise relevance propagation} (\QLRP{ })
    is our proposal to adopt the \emph{divide and explain} philosophy behind the LRP explanation technique~\cite{montavon2018methods}.
    In this section we introduce the main ideas behind these techniques, a full derivation together with implementation-oriented explanation can be found in Appendices~\ref{a:taylor}, \ref{a:twiNN}, and \ref{a:QLRP}.

\subsection{\taylorinf{} explanation}\label{ss:taylorinf}

    To arrive at \taylorinf{} we recall from Refs.~\cite{GilVidal2020,Schuld2021effect} that PQCs where each input component is encoded only once give rise to degree-$1$ trigonometric polynomials.
    That is, the functions realized by such circuits are always of the form
    \begin{align}
        f_\vartheta(x) &= \sum_{\omega\in\Omega} \left( a_\omega(\vartheta) \cos\langle\omega,x\rangle + b_\omega(\vartheta)\sin\langle\omega,x\rangle
        \right) 
        ,
    \end{align}
    where the values of the coefficients $(a_\omega(\vartheta),b_\omega(\vartheta))_\omega$ depend only on the trainable parameters of the PQC.
    Further the frequency spectrum $\Omega\subseteq\{0,\pm1\}^d$ is such that $\Omega\cup -\Omega=\{0,\pm1\}^d$, $\Omega\cap -\Omega=\{0\}^d$.
    Here $\langle\,\cdot\,,\cdot\,\rangle$ denotes the usual Euclidean inner product of vectors.

    We propose the explanation \taylorinf{} as a generalization of \taylorone{} in which we take not only the first-order terms $\partial_i f(\tx)(x_i-\tx_i)$, but rather all single-component contributions $\partial_i^kf(\tx) (x_i-\tx_i)^k$.
    Indeed, the \taylorinf{} explanation is based on the expression
    \begin{align}
        f(x) &= f(\tx) + \sum_{i=1}^d \left(\sum_{k=1}^\infty \frac{\partial_i^kf(\tx)(x_i-\tx_i)^k}{k!}\right)  + \varepsilon.
    \end{align}
    Now $\varepsilon$ contains only higher-order crossed terms, including partial derivatives with respect to different components.
    We denote the infinite series by $T_i(x,\tx)$, and exploiting usual trigonometric identities, we reach the defining formula 
    \begin{align}
        T_i(x,\tx) &= \sin(x_i-\tx_i) \partial_if(\tx) \\
        &\hphantom{=}+ (1-\cos(x_i-\tx_i)) \partial_i^2f(\tx)\nonumber
    \end{align}
    for \taylorinf{}.
    Together with an appropriate root point $\tx$ 
    for which both $f(\tx)$ and $\varepsilon$ become negligible, we reach the approximation
    \begin{align}
        f(x) &\approx \sum_{i=1}^d T_i(x,\tx).
    \end{align}
    We consequently define the explanation \taylorinf{} as $E_i(x)=T_i(x,\tx)$.
    Again, the question of finding suitable root points remains open.

    Note that \taylorinf{} is a strict improvement over \taylorone{} in terms of conservation for degree-$1$ trigonometric polynomials.
    A generalization of \taylorinf{} to higher-degree polynomials is left for future research.
    Note also that, \taylorinf{} ultimately being a black-box explanation, it could also be used for other types of functions, but in general the good approximation condition (conservation) cannot be expected to hold in general beyond degree-$1$ trigonometric polynomials.

\subsection{Quantum layerwise relevance propagation}\label{ss:qlrp}

    We start by regrouping the PQC scheme introduced above into a two-step process.
    \begin{enumerate}
        \item Prepare a data-dependent state: $x\mapsto\rho(x)$.
        \item Measure a task-dependent observable: $\calM(\vartheta) \mapsto \langle\calM(\vartheta)\rangle_{\rho(x)}$.
    \end{enumerate}
    In this \emph{interaction picture}, both the initial state and the measurement observable have been evolved, the former only under data-dependent unitary gates, and the latter only under the trainable ones. This perspective yields the simple computational graph
    \begin{align}
        x \xmapsto{\text{Encoding step}} \rho(x) \xmapsto{\text{Linear step}} f(x) = \Tr\{\rho(x)\calM(\vartheta)\}.
    \end{align}
    First the \emph{encoding step} is in general a non-linear map of $x$, then the \emph{linear step} corresponds to the quantum measurement, and is indeed a linear map of $\rho(x)$.
    For ease of implementation we introduce a \emph{twin neural network} (twiNN) in Appendix~\ref{a:twiNN} that mirrors these two steps.

    Following the divide-and-explain principle underlying the LRP method, we consider a two step explanation, following the computational graph in reverse order.
    In the first step, we produce the intermediate relevance for $\rho(x)$, which we denote by $R(\rho)$, and which depends on both the value of the function $f(x)$ and the entries of $\rho(x)$ itself: $(f(x),\rho(x))\mapsto R(\rho)$.
    In the second step, we reach the final explanation $E(x)$ by using information only from $x, \rho$, and the relevance of $\rho$: $(x, \rho, R(\rho))\mapsto E(x)$.
    To fully specify this approach one needs only provide concrete formulas for $R(\rho)$ and $E(x)$, which we call the \emph{linear rule} and the \emph{encoding rule}, respectively.

    In the XQML method we propose, the linear rule exploits the fact that $f(x)$ is a linear function 
    \begin{align}
        f(x) &= \Tr\{\rho(x)\calM(\vartheta)\} = \sum_{i,j} \rho_{i,j}(x) \calM_{i,j}(\vartheta)
    \end{align}
    of the entries of $\rho(x)$.
    The linear rule then is simply\footnote{
        For ease of implementation, we actually consider a real-valued version of every complex-valued matrix involved, as we explain in Appendix~\ref{a:twiNN}.
    } $R_{i,j}(\rho) = \rho_{i,j}(x)\calM_{i,j}(\vartheta)$, which is always conservative by construction, and it corresponds to \gradxinput{} introduced above.
    The encoding rule is in turn a bit more involved, as not only is the encoding step non-linear, but also this rule must take the relevance $R(\rho)$ into account.
    For the PQCs we consider (where each input component is encoded only once) each entry of $\rho(x)$ is itself a degree-$1$ polynomial in each of the components of $x$.
    That means that we can consider again the expansion 
    \begin{align}
        \rho_{i,j}(x) = \rho_{i,j}(\tx_{i,j}) + \sum_{k=1}^d T^{i,j}_k(x,\tx_{i,j}) + \varepsilon
    \end{align}
    we have used for the derivation of the \taylorinf{} relevance.
    Assuming we find good root points $\tx_{i,j}$ for each entry of $\rho(x)$, we need to combine these $T_k^{i,j}$ functions with the intermediate explanation we obtained from the linear rule.
    We propose the encoding rule
    \begin{align}
        E_k(x) &= \sum_{i,j} T_k^{i,j}(x,\tx_{i,j}) \frac{R_{i,j}(\rho)}{\rho_{i,j}},
    \end{align}
    with the goal of optimizing for conservation.
    For this formula only, we use the convention $0/0=0$, as for the zero entries $\rho_{i,j}(x)=0$ it holds also that $R_{i,j}(x)=0$ from the linear rule, and we need just remove those entries from the sum.
    One can readily see that, assuming we have good root points, the identity $\sum_{k=1}^d E_k(x)\approx f(x)$ holds, and we call the resulting explanation technique \QLRP{ }, which stands for \emph{quantum layerwise relevance propagation}.
    In Appendix~\ref{a:QLRP}, we give an algorithm for finding root points.

    The \QLRP{ } algorithm we propose involves several steps which require exponential space on a classical computer.
    The number of independent entries of $\rho(x)$ and $\calM(\vartheta)$ in general is $\calO(\exp(n))$, which accounts for the exponential 
    complexity of the linear rule.
    The runtime of the algorithms we propose to find the root points for each $\rho_{i,j}(x)$ is polynomial in $n$, but since there are exponentially many of them, the encoding rule also has exponential complexity.
    This way our proposed \QLRP{ } requires full classical simulation of the quantum circuit and storing the intermediate data-dependent state.
    In its current formulation, access to the PQC that evaluates the function does not improve the total runtime, as the storage requirements are still exponential throughout.

    This explanation technique represents the first adaption of the divide-and-explain principle underlying the LRP explanation technique to QML, albeit with some practical limitations.
    On the one hand, this is precisely in line with the improvement we wanted to achieve with respect to black-box methods.
    On the other hand, though, it means that this method is not scalable to larger problems with more qubits involved.
    
    Two important questions remain open.
    First, it would be interesting to establish how applicable this method is to strictly classically-simulable circuits (like circuits where all states involved are either low-entanglement or low-magic states).
    Second, one could adapt \QLRP{ } to take advantage of quantum hardware which, together with intermediate measurements if necessary, would reduce the storage requirements to be classically tractable, possibly at the cost of only approximately accurate explanations.

    \begin{figure}[t]
        \centering
        \includegraphics[width=\linewidth]{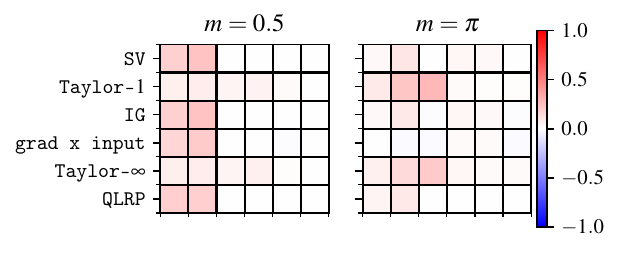}
        \caption{
        Average predicted explanation for data in the first class, for two synthetic learning tasks.
        For this class, the relevant components should be the first three.
        For the right task, with $m=\pi$, we observe that most explanation methods have average relevance close to $0$ for all components.}
        \label{fig:explanations}
    \end{figure}

\section{Numerical results}\label{s:results}

    \begin{figure*}[t!]
        \centering
        \includegraphics{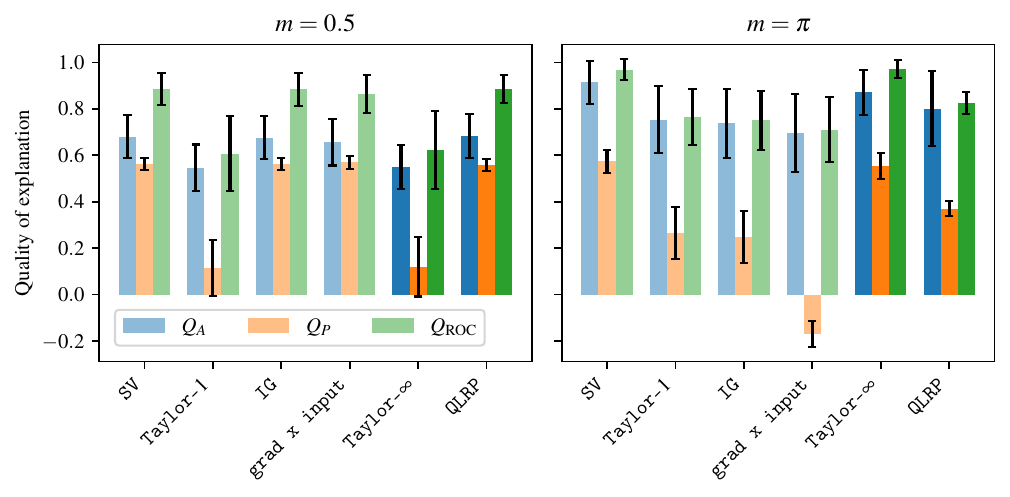}
        \caption{
        Quality of explanation for methods covered in this work, full report can be found in Appendix~\ref{a:experiments} and Fig.~\ref{fig:full-eval}.
        The difference in opacity differentiates between the existing local attribution methods introduced in Section~\ref{ss:posthoc} and the ones we introduce for quantum learning models in Section~\ref{s:methods}.
        }
        \label{fig:explanation_quality}
    \end{figure*}

    In order to showcase our proposed methods, \taylorinf{} and \QLRP{ }, we devise a proof of concept classification experiment.
    In this experiment, we use synthetic data for which we can evaluate not only the performance of a PQC, but also for which the explanation quality is easily quantifiable.
    Our goal is to display the practical application of our protocols, not to claim that quantum models have an advantage over classical algorithms in terms of explainability, nor that our proposed explanation techniques are superior to previously-established ones.
    With the aim of offering a point of comparison, we benchmark a subset of the model-agnostic explanations introduced in Section~\ref{ss:posthoc} together with our novel model-specific techniques.

    We propose a classification task using $6$-dimensional data sorted into $4$ different classes.
    The specifics of the task can be found in Appendix~\ref{a:experiments}.
    The main feature of the data generation is that three components are relevant for each class.
    The classification task is rather simple, in that points from different classes cluster together and are close to linearly separable.
    The relevant components are sampled from a normal distribution centered away from the origin, and the irrelevant components are sampled uniformly random.
    With this, for each class, we immediately know which components should be given high relevance, and which ones should receive zero attribution.

    We train a $6$-qubit PQC to solve the multi-class task, where four different fixed observables are measured at the end, one producing a relative score for each class.
    A relatively simple PQC where each component is encoded only once is already able to solve the task satisfactorily.
    Then, we produce local explanations for individual input patterns following all the ideas introduced above: both the established ones from Section~\ref{ss:posthoc} and our novel ones from Section~\ref{s:methods}.
    In Fig.~\ref{fig:explanations}, we can observe the average explanations produced by several different methods for elements of the first class.
    We note that, although we manually 
    set three components as relevant for each class, from the explainability analysis 
    we can infer that our simple model only makes use of two such components.
    A similar behavior 
    has been reported in Ref.~\cite{steinmuller2022explainable}.
    We also see a dichotomy between the experiment repetitions: the explanations produced by both \taylorone{} and \taylorinf{} differ from the other methods in that the relevant components are more starkly highlighted in the $m=\pi$ repetition.

    In order to quantify the quality of explanation, we focus on quality metrics that presuppose known information about the problem.
    We use a \emph{mask} $M(x)$, which has the same shape as $x$, and whose entries flag the actually relevant features of the input.
    Quality evaluation metrics correspond to different similarity measures between the explanation $E(x)$ and the mask $M(x)$ for each given input $x$. In Appendix~\ref{a:evaluation}, we provide a precise formulation of the quality metrics considered in this work.
    In the first measure, called \emph{explanation alignment} $Q_A(x)$, we simply compute the share of relevance that is distributed to the correct features.
    This takes the form of a normalized inner product between the elementwise magnitudes of the explanation and the mask.
    Next, given a set of inputs $(x_i)_i$, another measure is the \emph{Pearson correlation} between the explanations and the masks: $Q_P(x)=\operatorname{Corr}_{\{x_i\}_i}(E(x_i),M(x_i))$.
    Finally, borrowing from other corners of machine learning, we include a third metric $Q_\text{ROC}$ based on the \emph{receiver operator characteristics} (ROC), which relies on a binary classification of explanations based on the mask.
    The quality of explanation is depicted in Fig.~\ref{fig:explanation_quality}.

    As explained in Appendix~\ref{a:experiments}, we perform three different experiment repetitions, labeled by a hyperparameter we call $m$.
    We consider three values of $m$, and we observe that the larger $m$ is, the more difficult the classification becomes.
    In Fig.~\ref{fig:explanation_quality} we show only some of the explanation methods for the intermediate and higher values of the hyperparameter: $m=0.5$ and $m=\pi$ respectively.
    In both cases the achieved test accuracy is above $80\%$, and in the intermediate case it is close to optimal.
    What we observe in the plots is a rich range of qualities, both for different methods and for different metrics.
    Without a clear best choice across the board, the messages to be extracted from these experiments are:
    \begin{itemize}
        \item There is no obvious benefit in using model-specific explanations over model-agnostic ones.
            This may be an intrinsic fact of XQML or an artifact of the simplicity of our synthetic classification task.
        \item If the goal is to quantify the explainability of a QML model for a specific data-set, then efforts must be invested in validating each XQML method for that set.
        \item If the goal is to quantify the explainability of a QML model in general, then further research is needed.
            In this sense, there is no free lunch also for explainability.
    \end{itemize}

    These messages are in line with what we see in classical ML, where explanation methods start distinguishing themselves only with increasing complexity of the model and difficulty of the task.
    For instance, for very deep convolutional NNs or for transformers gradient-based approaches start having problems (while they may work fine for simple settings)~\cite{AchICML24}.

    When comparing the left and right plots in Fig.~\ref{fig:explanation_quality}, we observe that most explanations achieve higher scores for both the explanation alignment $Q_A$ and the ROC-based $Q_\text{ROC}$ evaluation metrics for the $m=\pi$ task.
    For the XQML methods introduced in this work, we see a reverse in the performance ranking: \QLRP{ } outperforms \taylorinf{} for $m=0.5$, but the converse is true for $m=\pi$.
    This may well be from specific artifacts in the particular toy problem we employ.
    Structurally, both experiment repetitions are very similar, with the only noticeable difference that the simple classifier achieves better accuracy for $m=0.5$ than for $m=\pi$.
    From this, we hypothesize that the second task being more complex, the training algorithm settles for a slightly simpler model that in the first task, thus accounting for the lower accuracy.
    If this were the case, a possible explanation for the better performance of the gradient-based explanation methods could be that the second model is closer to a linear one.
    As put fort in Ref.~\cite{AchICML24} for the study of explainability in transformers, linearity could account for a generic improvement in the explanation quality of gradient-based methods.
    Nevertheless, this particular experiment is synthetic in nature, and there is no strong evidence to believe the trends showed here would generalize to larger, more realistic scenarios.

    Next to measuring the quality of explanation across several metrics, we recall that the main difference between \taylorone{} and \taylorinf{} ought to be an improvement in conservation.
    To assess potential differences, we consider the \emph{relative approximation error} 
    \begin{align}
        \frac{\lvert\sum_{i=1}^d E_i(x) - f(x)\rvert}{\lvert f(x)\rvert}
    \end{align}
    of a specific explanation $E$.
    In Table~\ref{tab:conservation} we report the relative approximation error of both \taylorone{} and \taylorinf{} averaged over the data-sets corresponding to each of the two experiments.
    \begin{table}[h]
        \centering
        \caption{Relative approximation errors achieved by \taylorone{} and \taylorinf{} on the synthetic classification tasks.
        }
        \label{tab:conservation}
        \begin{tabular}{ccc}\hline\hline
            XQML Method & $m=0.5$ & $m=\pi$ \\\hline\hline
            \taylorone{} & $0.032\pm0.005$ & $0.40\pm0.03$ \\\hline
            \taylorinf{} & $0.028\pm0.004$ & $0.52\pm0.04$ \\\hline\hline
        \end{tabular}
    \end{table}

    We briefly recall that one of the guiding principles for the derivation of \taylorinf{} was to include more terms in the Taylor expansion of the function.
    By exploiting the structure of low-degree trigonometric functions, one would have expected a high degree of conservation in the resulting explanations.
    This is not what we observe in these experiments, where both methods produce similar results, and \taylorone{} is even more conservative for the $m=\pi$ experiment.
    The hypothesis introduced above that the classifier in $m=\pi$ was closer to linear would also explain this disparity.
    While \taylorinf{} contains strictly more terms of the Taylor expansion than \taylorone{}, it could be that the inclusion of these extra terms is harmful if the underlying function is close to linear, as for example if only positive contributions are taken from a balanced sum.
    Again, we leave this as a conjecture, and we do not spend further efforts in characterizing behaviors that could be specific to this synthetic toy problem.

\section{Summary and outlook}
    
    In this work, we have charted what the field of \emph{explainable quantum machine learning} (XQML) could look like, as a ``preview'' of the field.
    In order to provide such a preview, 
    we have discussed the importance of explainability in general for machine learning systems of all kinds.
    After reviewing the main concepts from the fields of \emph{explainable artificial intelligence} (XAI) and \emph{quantum machine learning} (QML), we have noted there is one significant disparity between QML and classical ML:
    while XAI is substantially more developed than XQML, advancements in explainability have lagged behind breakthroughs in performance for deep learning models.
    This is not yet the case for QML models, and this offers an opportunity for avoiding a trade-off between explainability and performance in QML.
    By recognizing the depth of potential impact of explainability in the development of QML, we can, therefore, decide to prioritize the design of inherently explainable QML models.

    Along the way, we have remarked that the very core features of quantum mechanics may render impossible the direct application of certain XAI ideas to QML models.
    The hardness of efficiently representing quantum states classically, the impossibility of copying quantum information, and the destruction of quantum coherence when measuring intermediate steps of a quantum computation force us to be imaginative in the design of XQML techniques.
    Next to laying out the formalism of XQML, in this work we have also taken the first steps forward, by proposing quantum versions of well-established classical XAI techniques.

    In our work, we have focused on post-hoc local attribution methods on the one hand, and on QML models based on \emph{parametrized quantum circuits} (PQCs) on the other hand.
    We have introduced two novel explanation techniques to be applicable for PQC-based QML models.
    Referencing the classical XAI ideas from where they stem, we called them \taylorinf{} and \emph{quantum layerwise relevance propagation}, \QLRP.
    \taylorinf{} leverages the Taylor decomposition of the labeling function, and allows us to retain a higher-degree of information than its classical counterpart \taylorone{}.
    \taylorinf{} is a black-box explanation, but it exploits the trigonometric structure of usual QML labeling functions~\cite{Schuld2021effect}.
    \QLRP{ } adapts the divide-and-explain mindset adopted in \emph{layerwise relevance propagation} (\LRP{})~\cite{bach2015pixel} to PQC-based QML models.
    By considering a PQC as a linear model 
    on a high-dimensional Hilbert space, in \QLRP{ } we propose a two-step explanation, taking into 
    account the structure of the intermediate quantum feature map.

    We present these novel explanation techniques with derivations substantially expanding on their classical counterparts.
    We accompany these analytical results with numerical experiments using a synthetic learning task.
    The goal of these experiments is not so much to 
   present state-of-the-art simulations at the forefront of classical simulations, but have been set up mainly to transparently illustrate the pipeline of implementation and evaluation of XQML techniques.
    
    By reporting the performance of different explanation techniques, we provide a comparison between existing model-agnostic explanations and our novel model-specific ones.
    We leave for future work a full characterization of the potential advantages and merits of the different techniques.
    We decide on this humbler stance for two important reasons:
    First, it is our understanding that strong evidence for QML models that perform well on practically-relevant tasks is still missing \cite{gilfuster2024relation}, which means there is a lack of understanding of which features will prove eventually relevant in the design of QML models.
    Second, the large-scale benchmark from Ref.~\cite{bowles2024betterclassicalsubtleart} raised the alarm that cherry-picking is a common issue in QML research.
    The conclusion from our numerics is that different methods perform differently for different tasks and under different metrics, and so further research is needed in (1) characterizing all XQML techniques further, and (2) understanding which features and metrics will be relevant when eventually designing practically-relevant QML models.
    To be precise: we do not search for quantum advantage in XAI, we rather contribute to the foundations of XQML and take first meaningful steps.

    Arguably the main challenge we encountered is that of scalability.
    To design an explanation that makes use of the specific structure of a PQC, the straightforward approach of fully characterizing the Hilbert space of quantum computation quickly becomes classically intractable.
    A generic issue in QML also arises in XQML: the choice of encoding is critical~\cite{Schuld2021effect}.
    Namely, independently of the achieved performance in solving a learning task, if the chosen encoding introduces too much complexity, explanation becomes critically difficult.
    We identified a trade-off that is also present in XAI for deep learning: the less structured the learning model, the more complex the explanation method must be.

    The challenges we encountered additionally dictate the potential ways forward in XQML.
    The scalability issue must be addressed by designing smart quantum explanation algorithms that are quantum-time- and classical-space-efficient.
    Such algorithms would involve intermediate measurements to extract partial information about the quantum computation.
    To better grasp the effect of the encoding on explainability, a promising direction is studying the explainability of kernel-based learning models, since QML has a deep connection to kernel methods~\cite{schuld2021supervised}.
    Additionally, the structure-complexity trade-off could dictate the design of QML models: to have the specific structures that would make explainability easier (for example, as proposed in Ref.~\cite{cybinski2024speakphysicistunderstandyou}.

    Other promising research directions include the generalization of our explanation techniques to other encodings, including data re-uploading circuits~\cite{perez2020data}, 
    or QML models taking quantum states as input data.
    Implementing larger-scale experiments would help us towards understanding what desirable features arise from different explanation techniques and evaluation metrics.
    From both these considerations, a further impactful contribution would be the compilation of earnest explainability-based guidelines to the design of 
    PQCs.
    With these guidelines, the difficulty of finding specific application domains for QML could be partially alleviated.

    More generally, other bridges could be built between XAI and quantum computing.
    For instance, one could further study the possibility of explaining parts of the quantum circuit, instead of only the produced explanations for a learning task~\cite{heese2023explaining}.
    This direction includes questions around mechanistic interpretability, meaning the study of relevance at the level of individual neurons.
    In the opposite direction, an interesting contribution would be to better chart the space of potential quantum advantage in XAI by establishing the computational complexity of the different steps in the current XAI pipelines for deep learning.
    Finally, our work could be expanded toward the direction of hybrid quantum-classical ML models, where PQCs are directly combined with classical neural networks.

\subsection*{Acknowledgements}

The authors thank Anna Dawid for insightful comments in an earlier version of this draft.
This work has been supported by the BMBF (Hybrid), the Munich Quantum Valley (K-4 and K-8), the BMWK (EniQmA), the Quantum Flagship (PasQuans2, Millenion), QuantERA (HQCC), the Cluster of Excellence MATH+, the DFG (CRC 183 and Research Unit KI-FOR 5363 $-$ project ID: 459422098), the Einstein Foundation (Einstein Research Unit on Quantum Devices), Berlin Quantum, and the ERC (DebuQC). For the Einstein Research Unit, this has been the result of an important joint-node collaboration. EGF is supported by a Google PhD Fellowship and acknowledges travel support from the European Union’s Horizon 2020 research and innovation programme under grant agreement No 951847.

\bibliography{sources.bib}

\begin{thebibliography}{96}%
\makeatletter
\providecommand \@ifxundefined [1]{%
 \@ifx{#1\undefined}
}%
\providecommand \@ifnum [1]{%
 \ifnum #1\expandafter \@firstoftwo
 \else \expandafter \@secondoftwo
 \fi
}%
\providecommand \@ifx [1]{%
 \ifx #1\expandafter \@firstoftwo
 \else \expandafter \@secondoftwo
 \fi
}%
\providecommand \natexlab [1]{#1}%
\providecommand \enquote  [1]{``#1''}%
\providecommand \bibnamefont  [1]{#1}%
\providecommand \bibfnamefont [1]{#1}%
\providecommand \citenamefont [1]{#1}%
\providecommand \href@noop [0]{\@secondoftwo}%
\providecommand \href [0]{\begingroup \@sanitize@url \@href}%
\providecommand \@href[1]{\@@startlink{#1}\@@href}%
\providecommand \@@href[1]{\endgroup#1\@@endlink}%
\providecommand \@sanitize@url [0]{\catcode `\\12\catcode `\$12\catcode
  `\&12\catcode `\#12\catcode `\^12\catcode `\_12\catcode `\%12\relax}%
\providecommand \@@startlink[1]{}%
\providecommand \@@endlink[0]{}%
\providecommand \url  [0]{\begingroup\@sanitize@url \@url }%
\providecommand \@url [1]{\endgroup\@href {#1}{\urlprefix }}%
\providecommand \urlprefix  [0]{URL }%
\providecommand \Eprint [0]{\href }%
\providecommand \doibase [0]{https://doi.org/}%
\providecommand \selectlanguage [0]{\@gobble}%
\providecommand \bibinfo  [0]{\@secondoftwo}%
\providecommand \bibfield  [0]{\@secondoftwo}%
\providecommand \translation [1]{[#1]}%
\providecommand \BibitemOpen [0]{}%
\providecommand \bibitemStop [0]{}%
\providecommand \bibitemNoStop [0]{.\EOS\space}%
\providecommand \EOS [0]{\spacefactor3000\relax}%
\providecommand \BibitemShut  [1]{\csname bibitem#1\endcsname}%
\let\auto@bib@innerbib\@empty
\bibitem [{\citenamefont {LeCun}\ \emph {et~al.}(2012)\citenamefont {LeCun},
  \citenamefont {Bottou}, \citenamefont {Orr},\ and\ \citenamefont
  {M{\"u}ller}}]{lecun2012neural}%
  \BibitemOpen
  \bibfield  {author} {\bibinfo {author} {\bibfnamefont {Y.~A.}\ \bibnamefont
  {LeCun}}, \bibinfo {author} {\bibfnamefont {L.}~\bibnamefont {Bottou}},
  \bibinfo {author} {\bibfnamefont {G.~B.}\ \bibnamefont {Orr}},\ and\ \bibinfo
  {author} {\bibfnamefont {K.-R.}\ \bibnamefont {M{\"u}ller}},\ }in\ \href
  {https://doi.org/10.1007/978-3-642-35289-8} {\emph {\bibinfo {booktitle}
  {Neural networks: Tricks of the trade}}}\ (\bibinfo  {publisher} {Springer},\
  \bibinfo {address} {Berlin},\ \bibinfo {year} {2012})\ pp.\ \bibinfo {pages}
  {9--48}\BibitemShut {NoStop}%
\bibitem [{\citenamefont {Goodfellow}\ \emph {et~al.}(2016)\citenamefont
  {Goodfellow}, \citenamefont {Bengio},\ and\ \citenamefont
  {Courville}}]{Goodfellow2016deep}%
  \BibitemOpen
  \bibfield  {author} {\bibinfo {author} {\bibfnamefont {I.}~\bibnamefont
  {Goodfellow}}, \bibinfo {author} {\bibfnamefont {Y.}~\bibnamefont {Bengio}},\
  and\ \bibinfo {author} {\bibfnamefont {A.}~\bibnamefont {Courville}},\ }\href
  {http://www.deeplearningbook.org} {\emph {\bibinfo {title} {Deep learning}}}\
  (\bibinfo  {publisher} {MIT Press},\ \bibinfo {address} {Cambridge (MA)},\
  \bibinfo {year} {2016})\BibitemShut {NoStop}%
\bibitem [{\citenamefont {Nov}\ \emph {et~al.}(2023)\citenamefont {Nov},
  \citenamefont {Singh},\ and\ \citenamefont {Mann}}]{nov2023putting}%
  \BibitemOpen
  \bibfield  {author} {\bibinfo {author} {\bibfnamefont {O.}~\bibnamefont
  {Nov}}, \bibinfo {author} {\bibfnamefont {N.}~\bibnamefont {Singh}},\ and\
  \bibinfo {author} {\bibfnamefont {D.}~\bibnamefont {Mann}},\ }\bibfield
  {title} {\bibinfo {title} {{Putting ChatGPT’s medical advice to the
  (Turing) test: survey study}},\ }\href {https://doi.org/10.2196/46939}
  {\bibfield  {journal} {\bibinfo  {journal} {JMIR Med. Educ.}\ }\textbf
  {\bibinfo {volume} {9}},\ \bibinfo {pages} {e46939} (\bibinfo {year}
  {2023})}\BibitemShut {NoStop}%
\bibitem [{\citenamefont {Silver}\ \emph {et~al.}(2016)\citenamefont {Silver},
  \citenamefont {Huang}, \citenamefont {Maddison}, \citenamefont {Guez},
  \citenamefont {Sifre}, \citenamefont {Driessche}, \citenamefont
  {Schrittwieser}, \citenamefont {Antonoglou}, \citenamefont {Panneershelvam},
  \citenamefont {Lanctot}, \citenamefont {Dieleman}, \citenamefont {Grewe},
  \citenamefont {Nham}, \citenamefont {Kalchbrenner}, \citenamefont
  {Sutskever}, \citenamefont {Graepel}, \citenamefont {Lillicrap},
  \citenamefont {Leach}, \citenamefont {Kavukcuoglu},\ and\ \citenamefont
  {Hassabis}}]{silver2016mastering}%
  \BibitemOpen
  \bibfield  {author} {\bibinfo {author} {\bibfnamefont {D.}~\bibnamefont
  {Silver}}, \bibinfo {author} {\bibfnamefont {A.}~\bibnamefont {Huang}},
  \bibinfo {author} {\bibfnamefont {C.~J.}\ \bibnamefont {Maddison}}, \bibinfo
  {author} {\bibfnamefont {A.}~\bibnamefont {Guez}}, \bibinfo {author}
  {\bibfnamefont {L.}~\bibnamefont {Sifre}}, \bibinfo {author} {\bibfnamefont
  {G.~V.~D.}\ \bibnamefont {Driessche}}, \bibinfo {author} {\bibfnamefont
  {J.}~\bibnamefont {Schrittwieser}}, \bibinfo {author} {\bibfnamefont
  {I.}~\bibnamefont {Antonoglou}}, \bibinfo {author} {\bibfnamefont
  {V.}~\bibnamefont {Panneershelvam}}, \bibinfo {author} {\bibfnamefont
  {M.}~\bibnamefont {Lanctot}}, \bibinfo {author} {\bibfnamefont
  {S.}~\bibnamefont {Dieleman}}, \bibinfo {author} {\bibfnamefont
  {D.}~\bibnamefont {Grewe}}, \bibinfo {author} {\bibfnamefont
  {J.}~\bibnamefont {Nham}}, \bibinfo {author} {\bibfnamefont {N.}~\bibnamefont
  {Kalchbrenner}}, \bibinfo {author} {\bibfnamefont {I.}~\bibnamefont
  {Sutskever}}, \bibinfo {author} {\bibfnamefont {T.}~\bibnamefont {Graepel}},
  \bibinfo {author} {\bibfnamefont {T.}~\bibnamefont {Lillicrap}}, \bibinfo
  {author} {\bibfnamefont {M.}~\bibnamefont {Leach}}, \bibinfo {author}
  {\bibfnamefont {K.}~\bibnamefont {Kavukcuoglu}},\ and\ \bibinfo {author}
  {\bibfnamefont {D.}~\bibnamefont {Hassabis}},\ }\bibfield  {title} {\bibinfo
  {title} {Mastering the game of go with deep neural networks and tree
  search},\ }\href {https://doi.org/10.1038/nature16961} {\bibfield  {journal}
  {\bibinfo  {journal} {Nature}\ }\textbf {\bibinfo {volume} {529}},\ \bibinfo
  {pages} {484} (\bibinfo {year} {2016})}\BibitemShut {NoStop}%
\bibitem [{\citenamefont {Deng}\ \emph {et~al.}(2013)\citenamefont {Deng},
  \citenamefont {Hinton},\ and\ \citenamefont {Kingsbury}}]{deng2013new}%
  \BibitemOpen
  \bibfield  {author} {\bibinfo {author} {\bibfnamefont {L.}~\bibnamefont
  {Deng}}, \bibinfo {author} {\bibfnamefont {G.}~\bibnamefont {Hinton}},\ and\
  \bibinfo {author} {\bibfnamefont {B.}~\bibnamefont {Kingsbury}},\ }\bibfield
  {title} {\bibinfo {title} {New types of deep neural network learning for
  speech recognition and related applications: An overview},\ }\href
  {https://doi.org/10.1109/ICASSP.2013.6639344} {\bibfield  {journal} {\bibinfo
   {journal} {IEEE Int. Conf. Ac. Speech and Signal Proc. (ICASSP)}\ ,\
  \bibinfo {pages} {8599}} (\bibinfo {year} {2013})}\BibitemShut {NoStop}%
\bibitem [{\citenamefont {He}\ \emph {et~al.}(2016)\citenamefont {He},
  \citenamefont {Zhang}, \citenamefont {Ren},\ and\ \citenamefont
  {Sun}}]{he2016deep}%
  \BibitemOpen
  \bibfield  {author} {\bibinfo {author} {\bibfnamefont {K.}~\bibnamefont
  {He}}, \bibinfo {author} {\bibfnamefont {X.}~\bibnamefont {Zhang}}, \bibinfo
  {author} {\bibfnamefont {S.}~\bibnamefont {Ren}},\ and\ \bibinfo {author}
  {\bibfnamefont {J.}~\bibnamefont {Sun}},\ }\bibfield  {title} {\bibinfo
  {title} {Deep residual learning for image recognition},\ }\href
  {https://doi.org/10.1109/CVPR.2016.90} {\bibfield  {journal} {\bibinfo
  {journal} {Proc. IEEE Conf. Comp. Vision and Pattern Rec. (CVPR)}\ ,\
  \bibinfo {pages} {770}} (\bibinfo {year} {2016})}\BibitemShut {NoStop}%
\bibitem [{\citenamefont {Arute}\ \emph {et~al.}(2019)\citenamefont {Arute}
  \emph {et~al.}}]{arute2019quantum}%
  \BibitemOpen
  \bibfield  {author} {\bibinfo {author} {\bibfnamefont {F.}~\bibnamefont
  {Arute}} \emph {et~al.},\ }\bibfield  {title} {\bibinfo {title} {Quantum
  supremacy using a programmable superconducting processor},\ }\href
  {https://doi.org/10.1038/s41586-019-1666-5} {\bibfield  {journal} {\bibinfo
  {journal} {Nature}\ }\textbf {\bibinfo {volume} {574}},\ \bibinfo {pages}
  {505} (\bibinfo {year} {2019})}\BibitemShut {NoStop}%
\bibitem [{\citenamefont {Hangleiter}\ and\ \citenamefont
  {Eisert}(2023)}]{hangleiter2022computational}%
  \BibitemOpen
  \bibfield  {author} {\bibinfo {author} {\bibfnamefont {D.}~\bibnamefont
  {Hangleiter}}\ and\ \bibinfo {author} {\bibfnamefont {J.}~\bibnamefont
  {Eisert}},\ }\bibfield  {title} {\bibinfo {title} {Computational advantage of
  quantum random sampling},\ }\href
  {https://doi.org/10.1103/RevModPhys.95.035001} {\bibfield  {journal}
  {\bibinfo  {journal} {Rev. Mod. Phys.}\ }\textbf {\bibinfo {volume} {95}},\
  \bibinfo {pages} {035001} (\bibinfo {year} {2023})}\BibitemShut {NoStop}%
\bibitem [{\citenamefont {Biamonte}\ \emph {et~al.}(2017)\citenamefont
  {Biamonte}, \citenamefont {Wittek}, \citenamefont {Pancotti}, \citenamefont
  {Rebentrost}, \citenamefont {Wiebe},\ and\ \citenamefont
  {Lloyd}}]{biamonte2017quantum}%
  \BibitemOpen
  \bibfield  {author} {\bibinfo {author} {\bibfnamefont {J.}~\bibnamefont
  {Biamonte}}, \bibinfo {author} {\bibfnamefont {P.}~\bibnamefont {Wittek}},
  \bibinfo {author} {\bibfnamefont {N.}~\bibnamefont {Pancotti}}, \bibinfo
  {author} {\bibfnamefont {P.}~\bibnamefont {Rebentrost}}, \bibinfo {author}
  {\bibfnamefont {N.}~\bibnamefont {Wiebe}},\ and\ \bibinfo {author}
  {\bibfnamefont {S.}~\bibnamefont {Lloyd}},\ }\bibfield  {title} {\bibinfo
  {title} {Quantum machine learning},\ }\href
  {https://doi.org/10.1038/nature23474} {\bibfield  {journal} {\bibinfo
  {journal} {Nature}\ }\textbf {\bibinfo {volume} {549}},\ \bibinfo {pages}
  {195} (\bibinfo {year} {2017})}\BibitemShut {NoStop}%
\bibitem [{\citenamefont {Dunjko}\ and\ \citenamefont
  {Briegel}(2018)}]{dunjko2018machine}%
  \BibitemOpen
  \bibfield  {author} {\bibinfo {author} {\bibfnamefont {V.}~\bibnamefont
  {Dunjko}}\ and\ \bibinfo {author} {\bibfnamefont {H.~J.}\ \bibnamefont
  {Briegel}},\ }\bibfield  {title} {\bibinfo {title} {Machine learning {\&}
  artificial intelligence in the quantum domain: a review of recent progress},\
  }\href {https://doi.org/10.1088/1361-6633/aab406} {\bibfield  {journal}
  {\bibinfo  {journal} {Rep. Prog. Phys.}\ }\textbf {\bibinfo {volume} {81}},\
  \bibinfo {pages} {074001} (\bibinfo {year} {2018})}\BibitemShut {NoStop}%
\bibitem [{\citenamefont {Carleo}\ \emph {et~al.}(2019)\citenamefont {Carleo},
  \citenamefont {Cirac}, \citenamefont {Cranmer}, \citenamefont {Daudet},
  \citenamefont {Schuld}, \citenamefont {Tishby}, \citenamefont
  {Vogt-Maranto},\ and\ \citenamefont {Zdeborov\'a}}]{carleo2019machine}%
  \BibitemOpen
  \bibfield  {author} {\bibinfo {author} {\bibfnamefont {G.}~\bibnamefont
  {Carleo}}, \bibinfo {author} {\bibfnamefont {I.}~\bibnamefont {Cirac}},
  \bibinfo {author} {\bibfnamefont {K.}~\bibnamefont {Cranmer}}, \bibinfo
  {author} {\bibfnamefont {L.}~\bibnamefont {Daudet}}, \bibinfo {author}
  {\bibfnamefont {M.}~\bibnamefont {Schuld}}, \bibinfo {author} {\bibfnamefont
  {N.}~\bibnamefont {Tishby}}, \bibinfo {author} {\bibfnamefont
  {L.}~\bibnamefont {Vogt-Maranto}},\ and\ \bibinfo {author} {\bibfnamefont
  {L.}~\bibnamefont {Zdeborov\'a}},\ }\bibfield  {title} {\bibinfo {title}
  {Machine learning and the physical sciences},\ }\href
  {https://doi.org/10.1103/RevModPhys.91.045002} {\bibfield  {journal}
  {\bibinfo  {journal} {Rev. Mod. Phys.}\ }\textbf {\bibinfo {volume} {91}},\
  \bibinfo {pages} {045002} (\bibinfo {year} {2019})}\BibitemShut {NoStop}%
\bibitem [{\citenamefont {McClean}\ \emph {et~al.}(2016)\citenamefont
  {McClean}, \citenamefont {Romero}, \citenamefont {Babbush},\ and\
  \citenamefont {Aspuru-Guzik}}]{McClean2016theory}%
  \BibitemOpen
  \bibfield  {author} {\bibinfo {author} {\bibfnamefont {J.~R.}\ \bibnamefont
  {McClean}}, \bibinfo {author} {\bibfnamefont {J.}~\bibnamefont {Romero}},
  \bibinfo {author} {\bibfnamefont {R.}~\bibnamefont {Babbush}},\ and\ \bibinfo
  {author} {\bibfnamefont {A.}~\bibnamefont {Aspuru-Guzik}},\ }\bibfield
  {title} {\bibinfo {title} {The theory of variational hybrid quantum-classical
  algorithms},\ }\href {https://doi.org/10.1088/1367-2630/18/2/023023}
  {\bibfield  {journal} {\bibinfo  {journal} {New J. Phys.}\ }\textbf {\bibinfo
  {volume} {18}},\ \bibinfo {pages} {023023} (\bibinfo {year}
  {2016})}\BibitemShut {NoStop}%
\bibitem [{\citenamefont {Sweke}\ \emph {et~al.}(2021)\citenamefont {Sweke},
  \citenamefont {Seifert}, \citenamefont {Hangleiter},\ and\ \citenamefont
  {Eisert}}]{sweke2021quantum}%
  \BibitemOpen
  \bibfield  {author} {\bibinfo {author} {\bibfnamefont {R.}~\bibnamefont
  {Sweke}}, \bibinfo {author} {\bibfnamefont {J.-P.}\ \bibnamefont {Seifert}},
  \bibinfo {author} {\bibfnamefont {D.}~\bibnamefont {Hangleiter}},\ and\
  \bibinfo {author} {\bibfnamefont {J.}~\bibnamefont {Eisert}},\ }\bibfield
  {title} {\bibinfo {title} {On the quantum versus classical learnability of
  discrete distributions},\ }\href {https://doi.org/10.22331/q-2021-03-23-417}
  {\bibfield  {journal} {\bibinfo  {journal} {Quantum}\ }\textbf {\bibinfo
  {volume} {5}},\ \bibinfo {pages} {417} (\bibinfo {year} {2021})}\BibitemShut
  {NoStop}%
\bibitem [{\citenamefont {Liu}\ \emph {et~al.}(2021)\citenamefont {Liu},
  \citenamefont {Arunachalam},\ and\ \citenamefont {Temme}}]{liu2021rigorous}%
  \BibitemOpen
  \bibfield  {author} {\bibinfo {author} {\bibfnamefont {Y.}~\bibnamefont
  {Liu}}, \bibinfo {author} {\bibfnamefont {S.}~\bibnamefont {Arunachalam}},\
  and\ \bibinfo {author} {\bibfnamefont {K.}~\bibnamefont {Temme}},\ }\bibfield
   {title} {\bibinfo {title} {A rigorous and robust quantum speed-up in
  supervised machine learning},\ }\href
  {https://doi.org/10.1038/s41567-021-01287-z} {\bibfield  {journal} {\bibinfo
  {journal} {Nature Phys.}\ }\textbf {\bibinfo {volume} {17}},\ \bibinfo
  {pages} {1013} (\bibinfo {year} {2021})}\BibitemShut {NoStop}%
\bibitem [{\citenamefont {Pirnay}\ \emph {et~al.}(2023)\citenamefont {Pirnay},
  \citenamefont {Sweke}, \citenamefont {Eisert},\ and\ \citenamefont
  {Seifert}}]{DensityModelling}%
  \BibitemOpen
  \bibfield  {author} {\bibinfo {author} {\bibfnamefont {N.}~\bibnamefont
  {Pirnay}}, \bibinfo {author} {\bibfnamefont {R.}~\bibnamefont {Sweke}},
  \bibinfo {author} {\bibfnamefont {J.}~\bibnamefont {Eisert}},\ and\ \bibinfo
  {author} {\bibfnamefont {J.-P.}\ \bibnamefont {Seifert}},\ }\bibfield
  {title} {\bibinfo {title} {A super-polynomial quantum-classical separation
  for density modelling},\ }\href {https://doi.org/10.1103/PhysRevA.107.042416}
  {\bibfield  {journal} {\bibinfo  {journal} {Phys. Rev. A}\ }\textbf {\bibinfo
  {volume} {107}},\ \bibinfo {pages} {042416} (\bibinfo {year}
  {2023})}\BibitemShut {NoStop}%
\bibitem [{\citenamefont {Szegedy}\ \emph {et~al.}(2014)\citenamefont
  {Szegedy}, \citenamefont {Zaremba}, \citenamefont {Sutskever}, \citenamefont
  {Bruna}, \citenamefont {Erhan}, \citenamefont {Goodfellow},\ and\
  \citenamefont {Fergus}}]{szegedy2014intriguing}%
  \BibitemOpen
  \bibfield  {author} {\bibinfo {author} {\bibfnamefont {C.}~\bibnamefont
  {Szegedy}}, \bibinfo {author} {\bibfnamefont {W.}~\bibnamefont {Zaremba}},
  \bibinfo {author} {\bibfnamefont {I.}~\bibnamefont {Sutskever}}, \bibinfo
  {author} {\bibfnamefont {J.}~\bibnamefont {Bruna}}, \bibinfo {author}
  {\bibfnamefont {D.}~\bibnamefont {Erhan}}, \bibinfo {author} {\bibfnamefont
  {I.~J.}\ \bibnamefont {Goodfellow}},\ and\ \bibinfo {author} {\bibfnamefont
  {R.}~\bibnamefont {Fergus}},\ }\bibfield  {title} {\bibinfo {title}
  {Intriguing properties of neural networks},\ }\href
  {https://openreview.net/forum?id=kklr_MTHMRQjG} {\bibfield  {journal}
  {\bibinfo  {journal} {{ICLR} (Poster)}\ } (\bibinfo {year}
  {2014})}\BibitemShut {NoStop}%
\bibitem [{\citenamefont {Lapuschkin}\ \emph {et~al.}(2016)\citenamefont
  {Lapuschkin}, \citenamefont {Binder}, \citenamefont {Montavon}, \citenamefont
  {M{\"u}ller},\ and\ \citenamefont {Samek}}]{lapuschkin2016analyzing}%
  \BibitemOpen
  \bibfield  {author} {\bibinfo {author} {\bibfnamefont {S.}~\bibnamefont
  {Lapuschkin}}, \bibinfo {author} {\bibfnamefont {A.}~\bibnamefont {Binder}},
  \bibinfo {author} {\bibfnamefont {G.}~\bibnamefont {Montavon}}, \bibinfo
  {author} {\bibfnamefont {K.-R.}\ \bibnamefont {M{\"u}ller}},\ and\ \bibinfo
  {author} {\bibfnamefont {W.}~\bibnamefont {Samek}},\ }\bibfield  {title}
  {\bibinfo {title} {Analyzing classifiers: Fisher vectors and deep neural
  networks},\ }in\ \href
  {https://openaccess.thecvf.com/content_cvpr_2016/html/Bach_Analyzing_Classifiers_Fisher_CVPR_2016_paper.html}
  {\emph {\bibinfo {booktitle} {Proceedings of the IEEE Conference on Computer
  Vision and Pattern Recognition (CVPR)}}}\ (\bibinfo {year} {2016})\ pp.\
  \bibinfo {pages} {2912--2920}\BibitemShut {NoStop}%
\bibitem [{\citenamefont {Samek}\ \emph {et~al.}(2021)\citenamefont {Samek},
  \citenamefont {Montavon}, \citenamefont {Lapuschkin}, \citenamefont
  {Anders},\ and\ \citenamefont {M{\"u}ller}}]{SamPIEEE21}%
  \BibitemOpen
  \bibfield  {author} {\bibinfo {author} {\bibfnamefont {W.}~\bibnamefont
  {Samek}}, \bibinfo {author} {\bibfnamefont {G.}~\bibnamefont {Montavon}},
  \bibinfo {author} {\bibfnamefont {S.}~\bibnamefont {Lapuschkin}}, \bibinfo
  {author} {\bibfnamefont {C.~J.}\ \bibnamefont {Anders}},\ and\ \bibinfo
  {author} {\bibfnamefont {K.-R.}\ \bibnamefont {M{\"u}ller}},\ }\bibfield
  {title} {\bibinfo {title} {Explaining deep neural networks and beyond: A
  review of methods and applications},\ }\href
  {https://doi.org/10.1109/JPROC.2021.3060483} {\bibfield  {journal} {\bibinfo
  {journal} {Proc. IEEE}\ }\textbf {\bibinfo {volume} {109}},\ \bibinfo {pages}
  {247} (\bibinfo {year} {2021})}\BibitemShut {NoStop}%
\bibitem [{\citenamefont {Minh}\ \emph {et~al.}(2022)\citenamefont {Minh},
  \citenamefont {Wang}, \citenamefont {Li},\ and\ \citenamefont
  {Nguyen}}]{xAIReview}%
  \BibitemOpen
  \bibfield  {author} {\bibinfo {author} {\bibfnamefont {D.}~\bibnamefont
  {Minh}}, \bibinfo {author} {\bibfnamefont {H.}~\bibnamefont {Wang}}, \bibinfo
  {author} {\bibfnamefont {Y.}~\bibnamefont {Li}},\ and\ \bibinfo {author}
  {\bibfnamefont {T.~N.}\ \bibnamefont {Nguyen}},\ }\bibfield  {title}
  {\bibinfo {title} {Explainable artificial intelligence: a comprehensive
  review},\ }\href {https://doi.org/10.1007/s10462-021-10088-y} {\bibfield
  {journal} {\bibinfo  {journal} {Artif. Intell. Rev.}\ }\textbf {\bibinfo
  {volume} {55}},\ \bibinfo {pages} {3503} (\bibinfo {year}
  {2022})}\BibitemShut {NoStop}%
\bibitem [{\citenamefont {Arrieta}\ \emph {et~al.}(2020)\citenamefont
  {Arrieta}, \citenamefont {Rodr{\'{\i}}guez}, \citenamefont {Ser},
  \citenamefont {Bennetot}, \citenamefont {Tabik}, \citenamefont {Barbado},
  \citenamefont {Garc{\'{\i}}a}, \citenamefont {Gil{-}Lopez}, \citenamefont
  {Molina}, \citenamefont {Benjamins}, \citenamefont {Chatila},\ and\
  \citenamefont {Herrera}}]{ArrietaRSBTBGGM20}%
  \BibitemOpen
  \bibfield  {author} {\bibinfo {author} {\bibfnamefont {A.~B.}\ \bibnamefont
  {Arrieta}}, \bibinfo {author} {\bibfnamefont {N.~D.}\ \bibnamefont
  {Rodr{\'{\i}}guez}}, \bibinfo {author} {\bibfnamefont {J.~D.}\ \bibnamefont
  {Ser}}, \bibinfo {author} {\bibfnamefont {A.}~\bibnamefont {Bennetot}},
  \bibinfo {author} {\bibfnamefont {S.}~\bibnamefont {Tabik}}, \bibinfo
  {author} {\bibfnamefont {A.}~\bibnamefont {Barbado}}, \bibinfo {author}
  {\bibfnamefont {S.}~\bibnamefont {Garc{\'{\i}}a}}, \bibinfo {author}
  {\bibfnamefont {S.}~\bibnamefont {Gil{-}Lopez}}, \bibinfo {author}
  {\bibfnamefont {D.}~\bibnamefont {Molina}}, \bibinfo {author} {\bibfnamefont
  {R.}~\bibnamefont {Benjamins}}, \bibinfo {author} {\bibfnamefont
  {R.}~\bibnamefont {Chatila}},\ and\ \bibinfo {author} {\bibfnamefont
  {F.}~\bibnamefont {Herrera}},\ }\bibfield  {title} {\bibinfo {title}
  {Explainable artificial intelligence {(XAI):} concepts, taxonomies,
  opportunities and challenges toward responsible {AI}},\ }\href
  {https://doi.org/10.1016/j.inffus.2019.12.012} {\bibfield  {journal}
  {\bibinfo  {journal} {Inf. Fusion}\ }\textbf {\bibinfo {volume} {58}},\
  \bibinfo {pages} {82} (\bibinfo {year} {2020})}\BibitemShut {NoStop}%
\bibitem [{\citenamefont {Steinmüller}\ \emph {et~al.}(2022)\citenamefont
  {Steinmüller}, \citenamefont {Schulz}, \citenamefont {Graf},\ and\
  \citenamefont {Herr}}]{steinmuller2022explainable}%
  \BibitemOpen
  \bibfield  {author} {\bibinfo {author} {\bibfnamefont {P.}~\bibnamefont
  {Steinmüller}}, \bibinfo {author} {\bibfnamefont {T.}~\bibnamefont
  {Schulz}}, \bibinfo {author} {\bibfnamefont {F.}~\bibnamefont {Graf}},\ and\
  \bibinfo {author} {\bibfnamefont {D.}~\bibnamefont {Herr}},\ }\bibfield
  {title} {\bibinfo {title} {{eXplainable AI for quantum machine learning}},\
  }\href {https://arxiv.org/abs/2211.01441} {\bibfield  {journal} {\bibinfo
  {journal} {arXiv:2211.01441}\ } (\bibinfo {year} {2022})}\BibitemShut
  {NoStop}%
\bibitem [{\citenamefont {Lifshitz}(2022)}]{lifshitz2022quantum}%
  \BibitemOpen
  \bibfield  {author} {\bibinfo {author} {\bibfnamefont {R.}~\bibnamefont
  {Lifshitz}},\ }\bibfield  {title} {\bibinfo {title} {Quantum deep dreaming: A
  novel approach for quantum circuit design},\ }\href
  {https://arxiv.org/abs/2211.04343} {\bibfield  {journal} {\bibinfo  {journal}
  {arXiv:2211.04343}\ } (\bibinfo {year} {2022})}\BibitemShut {NoStop}%
\bibitem [{\citenamefont {Heese}\ \emph {et~al.}(2023)\citenamefont {Heese},
  \citenamefont {Gerlach}, \citenamefont {Mücke}, \citenamefont {Müller},
  \citenamefont {Jakobs},\ and\ \citenamefont
  {Piatkowski}}]{heese2023explaining}%
  \BibitemOpen
  \bibfield  {author} {\bibinfo {author} {\bibfnamefont {R.}~\bibnamefont
  {Heese}}, \bibinfo {author} {\bibfnamefont {T.}~\bibnamefont {Gerlach}},
  \bibinfo {author} {\bibfnamefont {S.}~\bibnamefont {Mücke}}, \bibinfo
  {author} {\bibfnamefont {S.}~\bibnamefont {Müller}}, \bibinfo {author}
  {\bibfnamefont {M.}~\bibnamefont {Jakobs}},\ and\ \bibinfo {author}
  {\bibfnamefont {N.}~\bibnamefont {Piatkowski}},\ }\bibfield  {title}
  {\bibinfo {title} {Explaining quantum circuits with shapley values: Towards
  explainable quantum machine learning},\ }\href
  {https://arxiv.org/abs/2301.09138} {\bibfield  {journal} {\bibinfo  {journal}
  {arXiv:2301.09138}\ } (\bibinfo {year} {2023})}\BibitemShut {NoStop}%
\bibitem [{\citenamefont {Pira}\ and\ \citenamefont
  {Ferrie}(2024)}]{pira2024interpretability}%
  \BibitemOpen
  \bibfield  {author} {\bibinfo {author} {\bibfnamefont {L.}~\bibnamefont
  {Pira}}\ and\ \bibinfo {author} {\bibfnamefont {C.}~\bibnamefont {Ferrie}},\
  }\bibfield  {title} {\bibinfo {title} {On the interpretability of quantum
  neural networks},\ }\href {https://doi.org/10.1007/s42484-024-00191-y}
  {\bibfield  {journal} {\bibinfo  {journal} {Quant. Mach. Int.}\ }\textbf
  {\bibinfo {volume} {6}},\ \bibinfo {pages} {52} (\bibinfo {year}
  {2024})}\BibitemShut {NoStop}%
\bibitem [{\citenamefont {Dziugaite}\ \emph {et~al.}(2020)\citenamefont
  {Dziugaite}, \citenamefont {Ben-David},\ and\ \citenamefont
  {Roy}}]{dziugaite2020enforcinginterpretabilitystatisticalimpacts}%
  \BibitemOpen
  \bibfield  {author} {\bibinfo {author} {\bibfnamefont {G.~K.}\ \bibnamefont
  {Dziugaite}}, \bibinfo {author} {\bibfnamefont {S.}~\bibnamefont
  {Ben-David}},\ and\ \bibinfo {author} {\bibfnamefont {D.~M.}\ \bibnamefont
  {Roy}},\ }\bibfield  {title} {\bibinfo {title} {Enforcing interpretability
  and its statistical impacts: Trade-offs between accuracy and
  interpretability},\ }\href {https://arxiv.org/abs/2010.13764} {\bibfield
  {journal} {\bibinfo  {journal} {arXiv:2010.13764}\ } (\bibinfo {year}
  {2020})}\BibitemShut {NoStop}%
\bibitem [{\citenamefont {Gil~Vidal}\ and\ \citenamefont
  {Theis}(2020)}]{GilVidal2020}%
  \BibitemOpen
  \bibfield  {author} {\bibinfo {author} {\bibfnamefont {F.~J.}\ \bibnamefont
  {Gil~Vidal}}\ and\ \bibinfo {author} {\bibfnamefont {D.~O.}\ \bibnamefont
  {Theis}},\ }\bibfield  {title} {\bibinfo {title} {Input redundancy for
  parameterized quantum circuits},\ }\href
  {https://doi.org/10.3389/fphy.2020.00297} {\bibfield  {journal} {\bibinfo
  {journal} {Frontiers in Physics}\ }\textbf {\bibinfo {volume} {8}},\ \bibinfo
  {pages} {13} (\bibinfo {year} {2020})}\BibitemShut {NoStop}%
\bibitem [{\citenamefont {Schuld}\ \emph {et~al.}(2021)\citenamefont {Schuld},
  \citenamefont {Sweke},\ and\ \citenamefont {Meyer}}]{Schuld2021effect}%
  \BibitemOpen
  \bibfield  {author} {\bibinfo {author} {\bibfnamefont {M.}~\bibnamefont
  {Schuld}}, \bibinfo {author} {\bibfnamefont {R.}~\bibnamefont {Sweke}},\ and\
  \bibinfo {author} {\bibfnamefont {J.~J.}\ \bibnamefont {Meyer}},\ }\bibfield
  {title} {\bibinfo {title} {Effect of data encoding on the expressive power of
  variational quantum-machine-learning models},\ }\href
  {https://doi.org/10.1103/PhysRevA.103.032430} {\bibfield  {journal} {\bibinfo
   {journal} {Phys. Rev. A}\ }\textbf {\bibinfo {volume} {103}},\ \bibinfo
  {pages} {032430} (\bibinfo {year} {2021})}\BibitemShut {NoStop}%
\bibitem [{\citenamefont {P{\'e}rez-Salinas}\ \emph {et~al.}(2020)\citenamefont
  {P{\'e}rez-Salinas}, \citenamefont {Cervera-Lierta}, \citenamefont
  {Gil-Fuster},\ and\ \citenamefont {Latorre}}]{perez2020data}%
  \BibitemOpen
  \bibfield  {author} {\bibinfo {author} {\bibfnamefont {A.}~\bibnamefont
  {P{\'e}rez-Salinas}}, \bibinfo {author} {\bibfnamefont {A.}~\bibnamefont
  {Cervera-Lierta}}, \bibinfo {author} {\bibfnamefont {E.}~\bibnamefont
  {Gil-Fuster}},\ and\ \bibinfo {author} {\bibfnamefont {J.~I.}\ \bibnamefont
  {Latorre}},\ }\bibfield  {title} {\bibinfo {title} {{Data re-uploading for a
  universal quantum classifier}},\ }\href
  {https://doi.org/10.22331/q-2020-02-06-226} {\bibfield  {journal} {\bibinfo
  {journal} {Quantum}\ }\textbf {\bibinfo {volume} {4}},\ \bibinfo {pages}
  {226} (\bibinfo {year} {2020})}\BibitemShut {NoStop}%
\bibitem [{\citenamefont {Abbas}\ \emph {et~al.}(2024)\citenamefont {Abbas},
  \citenamefont {King}, \citenamefont {Huang}, \citenamefont {Huggins},
  \citenamefont {Movassagh}, \citenamefont {Gilboa},\ and\ \citenamefont
  {McClean}}]{abbas2024quantum}%
  \BibitemOpen
  \bibfield  {author} {\bibinfo {author} {\bibfnamefont {A.}~\bibnamefont
  {Abbas}}, \bibinfo {author} {\bibfnamefont {R.}~\bibnamefont {King}},
  \bibinfo {author} {\bibfnamefont {H.-Y.}\ \bibnamefont {Huang}}, \bibinfo
  {author} {\bibfnamefont {W.~J.}\ \bibnamefont {Huggins}}, \bibinfo {author}
  {\bibfnamefont {R.}~\bibnamefont {Movassagh}}, \bibinfo {author}
  {\bibfnamefont {D.}~\bibnamefont {Gilboa}},\ and\ \bibinfo {author}
  {\bibfnamefont {J.}~\bibnamefont {McClean}},\ }\bibfield  {title} {\bibinfo
  {title} {On quantum backpropagation, information reuse, and cheating
  measurement collapse},\ }\href {https://doi.org/10.48550/arXiv.2305.13362}
  {\bibfield  {journal} {\bibinfo  {journal} {Adv. Neur. Inf. Proc. Syst.}\
  }\textbf {\bibinfo {volume} {2026}},\ \bibinfo {pages} {36} (\bibinfo {year}
  {2024})}\BibitemShut {NoStop}%
\bibitem [{\citenamefont {McClean}\ \emph {et~al.}(2018)\citenamefont
  {McClean}, \citenamefont {Boixo}, \citenamefont {Smelyanskiy}, \citenamefont
  {Babbush},\ and\ \citenamefont {Neven}}]{mcclean2018barren}%
  \BibitemOpen
  \bibfield  {author} {\bibinfo {author} {\bibfnamefont {J.~R.}\ \bibnamefont
  {McClean}}, \bibinfo {author} {\bibfnamefont {S.}~\bibnamefont {Boixo}},
  \bibinfo {author} {\bibfnamefont {V.~N.}\ \bibnamefont {Smelyanskiy}},
  \bibinfo {author} {\bibfnamefont {R.}~\bibnamefont {Babbush}},\ and\ \bibinfo
  {author} {\bibfnamefont {H.}~\bibnamefont {Neven}},\ }\bibfield  {title}
  {\bibinfo {title} {Barren plateaus in quantum neural network training
  landscapes},\ }\href {https://doi.org/10.1038/s41467-018-07090-4} {\bibfield
  {journal} {\bibinfo  {journal} {Nature Comm.}\ }\textbf {\bibinfo {volume}
  {9}},\ \bibinfo {pages} {4812} (\bibinfo {year} {2018})}\BibitemShut
  {NoStop}%
\bibitem [{\citenamefont {Larocca}\ \emph {et~al.}(2024)\citenamefont
  {Larocca}, \citenamefont {Thanasilp}, \citenamefont {Wang}, \citenamefont
  {Sharma}, \citenamefont {Biamonte}, \citenamefont {Coles}, \citenamefont
  {Cincio}, \citenamefont {McClean}, \citenamefont {Holmes},\ and\
  \citenamefont {Cerezo}}]{larocca2024review}%
  \BibitemOpen
  \bibfield  {author} {\bibinfo {author} {\bibfnamefont {M.}~\bibnamefont
  {Larocca}}, \bibinfo {author} {\bibfnamefont {S.}~\bibnamefont {Thanasilp}},
  \bibinfo {author} {\bibfnamefont {S.}~\bibnamefont {Wang}}, \bibinfo {author}
  {\bibfnamefont {K.}~\bibnamefont {Sharma}}, \bibinfo {author} {\bibfnamefont
  {J.}~\bibnamefont {Biamonte}}, \bibinfo {author} {\bibfnamefont {P.~J.}\
  \bibnamefont {Coles}}, \bibinfo {author} {\bibfnamefont {L.}~\bibnamefont
  {Cincio}}, \bibinfo {author} {\bibfnamefont {J.~R.}\ \bibnamefont {McClean}},
  \bibinfo {author} {\bibfnamefont {Z.}~\bibnamefont {Holmes}},\ and\ \bibinfo
  {author} {\bibfnamefont {M.}~\bibnamefont {Cerezo}},\ }\bibfield  {title}
  {\bibinfo {title} {A review of barren plateaus in variational quantum
  computing},\ }\href {https://arxiv.org/abs/2405.00781} {\bibfield  {journal}
  {\bibinfo  {journal} {arXiv:2405.00781}\ } (\bibinfo {year}
  {2024})}\BibitemShut {NoStop}%
\bibitem [{\citenamefont {Bravyi}\ and\ \citenamefont
  {Gosset}(2016)}]{Bravyi2016improved}%
  \BibitemOpen
  \bibfield  {author} {\bibinfo {author} {\bibfnamefont {S.}~\bibnamefont
  {Bravyi}}\ and\ \bibinfo {author} {\bibfnamefont {D.}~\bibnamefont
  {Gosset}},\ }\bibfield  {title} {\bibinfo {title} {{Improved classical
  simulation of quantum circuits dominated by Clifford gates}},\ }\href
  {https://doi.org/10.1103/PhysRevLett.116.250501} {\bibfield  {journal}
  {\bibinfo  {journal} {Phys. Rev. Lett.}\ }\textbf {\bibinfo {volume} {116}},\
  \bibinfo {pages} {250501} (\bibinfo {year} {2016})}\BibitemShut {NoStop}%
\bibitem [{\citenamefont {Tang}(2019)}]{Tang2019}%
  \BibitemOpen
  \bibfield  {author} {\bibinfo {author} {\bibfnamefont {E.}~\bibnamefont
  {Tang}},\ }\bibfield  {title} {\bibinfo {title} {A quantum-inspired classical
  algorithm for recommendation systems},\ }in\ \href
  {https://doi.org/10.1145/3313276.3316310} {\emph {\bibinfo {booktitle}
  {Proceedings of the 51st Annual ACM SIGACT Symposium on Theory of
  Computing}}},\ \bibinfo {series and number} {STOC ’19}\ (\bibinfo
  {publisher} {ACM},\ \bibinfo {year} {2019})\BibitemShut {NoStop}%
\bibitem [{\citenamefont {Schreiber}\ \emph {et~al.}(2023)\citenamefont
  {Schreiber}, \citenamefont {Eisert},\ and\ \citenamefont
  {Meyer}}]{schreiber2022classical}%
  \BibitemOpen
  \bibfield  {author} {\bibinfo {author} {\bibfnamefont {F.~J.}\ \bibnamefont
  {Schreiber}}, \bibinfo {author} {\bibfnamefont {J.}~\bibnamefont {Eisert}},\
  and\ \bibinfo {author} {\bibfnamefont {J.~J.}\ \bibnamefont {Meyer}},\
  }\bibfield  {title} {\bibinfo {title} {Classical surrogates for quantum
  learning models},\ }\href {https://doi.org/10.1103/PhysRevLett.131.100803}
  {\bibfield  {journal} {\bibinfo  {journal} {Phys. Rev. Lett.}\ }\textbf
  {\bibinfo {volume} {131}},\ \bibinfo {pages} {100803} (\bibinfo {year}
  {2023})}\BibitemShut {NoStop}%
\bibitem [{\citenamefont {Landman}\ \emph {et~al.}(2023)\citenamefont
  {Landman}, \citenamefont {Thabet}, \citenamefont {Dalyac}, \citenamefont
  {Mhiri},\ and\ \citenamefont {Kashefi}}]{landman2023classically}%
  \BibitemOpen
  \bibfield  {author} {\bibinfo {author} {\bibfnamefont {J.}~\bibnamefont
  {Landman}}, \bibinfo {author} {\bibfnamefont {S.}~\bibnamefont {Thabet}},
  \bibinfo {author} {\bibfnamefont {C.}~\bibnamefont {Dalyac}}, \bibinfo
  {author} {\bibfnamefont {H.}~\bibnamefont {Mhiri}},\ and\ \bibinfo {author}
  {\bibfnamefont {E.}~\bibnamefont {Kashefi}},\ }\bibfield  {title} {\bibinfo
  {title} {Classically approximating variational quantum machine learning with
  random fourier features},\ }in\ \href
  {https://openreview.net/forum?id=ymFhZxw70uz} {\emph {\bibinfo {booktitle}
  {The Eleventh International Conference on Learning Representations}}}\
  (\bibinfo {year} {2023})\BibitemShut {NoStop}%
\bibitem [{\citenamefont {Sweke}\ \emph {et~al.}(2023)\citenamefont {Sweke},
  \citenamefont {Recio}, \citenamefont {Jerbi}, \citenamefont {Gil-Fuster},
  \citenamefont {Fuller}, \citenamefont {Eisert},\ and\ \citenamefont
  {Meyer}}]{Sweke2023Potential}%
  \BibitemOpen
  \bibfield  {author} {\bibinfo {author} {\bibfnamefont {R.}~\bibnamefont
  {Sweke}}, \bibinfo {author} {\bibfnamefont {E.}~\bibnamefont {Recio}},
  \bibinfo {author} {\bibfnamefont {S.}~\bibnamefont {Jerbi}}, \bibinfo
  {author} {\bibfnamefont {E.}~\bibnamefont {Gil-Fuster}}, \bibinfo {author}
  {\bibfnamefont {B.}~\bibnamefont {Fuller}}, \bibinfo {author} {\bibfnamefont
  {J.}~\bibnamefont {Eisert}},\ and\ \bibinfo {author} {\bibfnamefont {J.~J.}\
  \bibnamefont {Meyer}},\ }\bibfield  {title} {\bibinfo {title} {{Potential and
  limitations of random Fourier features for dequantizing quantum machine
  learning}},\ }\href {https://arxiv.org/abs/2309.11647} {\bibfield  {journal}
  {\bibinfo  {journal} {arXiv:2309.11647}\ } (\bibinfo {year}
  {2023})}\BibitemShut {NoStop}%
\bibitem [{\citenamefont {Shin}\ \emph {et~al.}(2024)\citenamefont {Shin},
  \citenamefont {Teo},\ and\ \citenamefont {Jeong}}]{shin2024dequantizing}%
  \BibitemOpen
  \bibfield  {author} {\bibinfo {author} {\bibfnamefont {S.}~\bibnamefont
  {Shin}}, \bibinfo {author} {\bibfnamefont {Y.~S.}\ \bibnamefont {Teo}},\ and\
  \bibinfo {author} {\bibfnamefont {H.}~\bibnamefont {Jeong}},\ }\bibfield
  {title} {\bibinfo {title} {Dequantizing quantum machine learning models using
  tensor networks},\ }\href {https://doi.org/10.1103/PhysRevResearch.6.023218}
  {\bibfield  {journal} {\bibinfo  {journal} {Phys. Rev. Res.}\ }\textbf
  {\bibinfo {volume} {6}},\ \bibinfo {pages} {023218} (\bibinfo {year}
  {2024})}\BibitemShut {NoStop}%
\bibitem [{\citenamefont {Rudolph}\ \emph {et~al.}(2023)\citenamefont
  {Rudolph}, \citenamefont {Fontana}, \citenamefont {Holmes},\ and\
  \citenamefont {Cincio}}]{rudolph2023classical}%
  \BibitemOpen
  \bibfield  {author} {\bibinfo {author} {\bibfnamefont {M.~S.}\ \bibnamefont
  {Rudolph}}, \bibinfo {author} {\bibfnamefont {E.}~\bibnamefont {Fontana}},
  \bibinfo {author} {\bibfnamefont {Z.}~\bibnamefont {Holmes}},\ and\ \bibinfo
  {author} {\bibfnamefont {L.}~\bibnamefont {Cincio}},\ }\bibfield  {title}
  {\bibinfo {title} {{Classical surrogate simulation of quantum systems with
  LOWESA}},\ }\href {https://arxiv.org/abs/2308.09109} {\bibfield  {journal}
  {\bibinfo  {journal} {arXiv:2308.09109}\ } (\bibinfo {year}
  {2023})}\BibitemShut {NoStop}%
\bibitem [{\citenamefont {Cerezo}\ \emph {et~al.}(2023)\citenamefont {Cerezo},
  \citenamefont {Larocca}, \citenamefont {Garc{\'\i}a-Mart{\'\i}n},
  \citenamefont {Diaz}, \citenamefont {Braccia}, \citenamefont {Fontana},
  \citenamefont {Rudolph}, \citenamefont {Bermejo}, \citenamefont {Ijaz},
  \citenamefont {Thanasilp} \emph {et~al.}}]{cerezo2023does}%
  \BibitemOpen
  \bibfield  {author} {\bibinfo {author} {\bibfnamefont {M.}~\bibnamefont
  {Cerezo}}, \bibinfo {author} {\bibfnamefont {M.}~\bibnamefont {Larocca}},
  \bibinfo {author} {\bibfnamefont {D.}~\bibnamefont
  {Garc{\'\i}a-Mart{\'\i}n}}, \bibinfo {author} {\bibfnamefont {N.~L.}\
  \bibnamefont {Diaz}}, \bibinfo {author} {\bibfnamefont {P.}~\bibnamefont
  {Braccia}}, \bibinfo {author} {\bibfnamefont {E.}~\bibnamefont {Fontana}},
  \bibinfo {author} {\bibfnamefont {M.~S.}\ \bibnamefont {Rudolph}}, \bibinfo
  {author} {\bibfnamefont {P.}~\bibnamefont {Bermejo}}, \bibinfo {author}
  {\bibfnamefont {A.}~\bibnamefont {Ijaz}}, \bibinfo {author} {\bibfnamefont
  {S.}~\bibnamefont {Thanasilp}}, \emph {et~al.},\ }\bibfield  {title}
  {\bibinfo {title} {Does provable absence of barren plateaus imply classical
  simulability? or, why we need to rethink variational quantum computing},\
  }\href {https://arxiv.org/abs/2312.09121} {\bibfield  {journal} {\bibinfo
  {journal} {arXiv:2312.09121}\ } (\bibinfo {year} {2023})}\BibitemShut
  {NoStop}%
\bibitem [{\citenamefont {Dias}\ and\ \citenamefont
  {Koenig}(2024)}]{dias2024classical}%
  \BibitemOpen
  \bibfield  {author} {\bibinfo {author} {\bibfnamefont {B.}~\bibnamefont
  {Dias}}\ and\ \bibinfo {author} {\bibfnamefont {R.}~\bibnamefont {Koenig}},\
  }\bibfield  {title} {\bibinfo {title} {{Classical simulation of non-Gaussian
  fermionic circuits}},\ }\href {https://doi.org/0.22331/q-2024-05-21-1350}
  {\bibfield  {journal} {\bibinfo  {journal} {Quantum}\ }\textbf {\bibinfo
  {volume} {8}},\ \bibinfo {pages} {1350} (\bibinfo {year} {2024})}\BibitemShut
  {NoStop}%
\bibitem [{\citenamefont {Gunning}\ and\ \citenamefont
  {Aha}(2019)}]{GunningA19}%
  \BibitemOpen
  \bibfield  {author} {\bibinfo {author} {\bibfnamefont {D.}~\bibnamefont
  {Gunning}}\ and\ \bibinfo {author} {\bibfnamefont {D.~W.}\ \bibnamefont
  {Aha}},\ }\bibfield  {title} {\bibinfo {title} {{DARPA's explainable
  artificial intelligence {(XAI)} program}},\ }\href
  {https://doi.org/10.1609/aimag.v40i2.2850} {\bibfield  {journal} {\bibinfo
  {journal} {{AI} Mag.}\ }\textbf {\bibinfo {volume} {40}},\ \bibinfo {pages}
  {44} (\bibinfo {year} {2019})}\BibitemShut {NoStop}%
\bibitem [{\citenamefont {Lapuschkin}\ \emph {et~al.}(2019)\citenamefont
  {Lapuschkin}, \citenamefont {Wäldchen}, \citenamefont {Binder},
  \citenamefont {Montavon}, \citenamefont {Samek},\ and\ \citenamefont
  {Müller}}]{lapuschkin_unmasking_2019}%
  \BibitemOpen
  \bibfield  {author} {\bibinfo {author} {\bibfnamefont {S.}~\bibnamefont
  {Lapuschkin}}, \bibinfo {author} {\bibfnamefont {S.}~\bibnamefont
  {Wäldchen}}, \bibinfo {author} {\bibfnamefont {A.}~\bibnamefont {Binder}},
  \bibinfo {author} {\bibfnamefont {G.}~\bibnamefont {Montavon}}, \bibinfo
  {author} {\bibfnamefont {W.}~\bibnamefont {Samek}},\ and\ \bibinfo {author}
  {\bibfnamefont {K.-R.}\ \bibnamefont {Müller}},\ }\bibfield  {title}
  {\bibinfo {title} {Unmasking {clever} {Hans} predictors and assessing what
  machines really learn},\ }\href {https://doi.org/10.1038/s41467-019-08987-4}
  {\bibfield  {journal} {\bibinfo  {journal} {Nature Comm.}\ }\textbf {\bibinfo
  {volume} {10}},\ \bibinfo {pages} {1096} (\bibinfo {year}
  {2019})}\BibitemShut {NoStop}%
\bibitem [{\citenamefont {Kauffmann}\ \emph {et~al.}(2024)\citenamefont
  {Kauffmann}, \citenamefont {Dippel}, \citenamefont {Ruff}, \citenamefont
  {Samek}, \citenamefont {M{\"{u}}ller},\ and\ \citenamefont
  {Montavon}}]{unsupervised-ch}%
  \BibitemOpen
  \bibfield  {author} {\bibinfo {author} {\bibfnamefont {J.~R.}\ \bibnamefont
  {Kauffmann}}, \bibinfo {author} {\bibfnamefont {J.}~\bibnamefont {Dippel}},
  \bibinfo {author} {\bibfnamefont {L.}~\bibnamefont {Ruff}}, \bibinfo {author}
  {\bibfnamefont {W.}~\bibnamefont {Samek}}, \bibinfo {author} {\bibfnamefont
  {K.}~\bibnamefont {M{\"{u}}ller}},\ and\ \bibinfo {author} {\bibfnamefont
  {G.}~\bibnamefont {Montavon}},\ }\bibfield  {title} {\bibinfo {title} {{The
  clever Hans effect in unsupervised learning}},\ }\href
  {https://arxiv.org/abs/2408.08041} {\bibfield  {journal} {\bibinfo  {journal}
  {arXiv:2408.08041}\ } (\bibinfo {year} {2024})}\BibitemShut {NoStop}%
\bibitem [{\citenamefont {Geirhos}\ \emph {et~al.}(2020)\citenamefont
  {Geirhos}, \citenamefont {Jacobsen}, \citenamefont {Michaelis}, \citenamefont
  {Zemel}, \citenamefont {Brendel}, \citenamefont {Bethge},\ and\ \citenamefont
  {Wichmann}}]{geirhos_shortcut_2020}%
  \BibitemOpen
  \bibfield  {author} {\bibinfo {author} {\bibfnamefont {R.}~\bibnamefont
  {Geirhos}}, \bibinfo {author} {\bibfnamefont {J.-H.}\ \bibnamefont
  {Jacobsen}}, \bibinfo {author} {\bibfnamefont {C.}~\bibnamefont {Michaelis}},
  \bibinfo {author} {\bibfnamefont {R.}~\bibnamefont {Zemel}}, \bibinfo
  {author} {\bibfnamefont {W.}~\bibnamefont {Brendel}}, \bibinfo {author}
  {\bibfnamefont {M.}~\bibnamefont {Bethge}},\ and\ \bibinfo {author}
  {\bibfnamefont {F.~A.}\ \bibnamefont {Wichmann}},\ }\bibfield  {title}
  {\bibinfo {title} {Shortcut learning in deep neural networks},\ }\href
  {https://doi.org/10.1038/s42256-020-00257-z} {\bibfield  {journal} {\bibinfo
  {journal} {Nature Mach. Intell.}\ }\textbf {\bibinfo {volume} {2}},\ \bibinfo
  {pages} {665–673} (\bibinfo {year} {2020})}\BibitemShut {NoStop}%
\bibitem [{\citenamefont {DeGrave}\ \emph {et~al.}(2021)\citenamefont
  {DeGrave}, \citenamefont {Janizek},\ and\ \citenamefont {Lee}}]{DegraveJL21}%
  \BibitemOpen
  \bibfield  {author} {\bibinfo {author} {\bibfnamefont {A.~J.}\ \bibnamefont
  {DeGrave}}, \bibinfo {author} {\bibfnamefont {J.~D.}\ \bibnamefont
  {Janizek}},\ and\ \bibinfo {author} {\bibfnamefont {S.}~\bibnamefont {Lee}},\
  }\bibfield  {title} {\bibinfo {title} {{AI} for radiographic {COVID-19}
  detection selects shortcuts over signal},\ }\href
  {https://doi.org/10.1038/s42256-021-00338-7} {\bibfield  {journal} {\bibinfo
  {journal} {Nature Mach. Intell.}\ }\textbf {\bibinfo {volume} {3}},\ \bibinfo
  {pages} {610} (\bibinfo {year} {2021})}\BibitemShut {NoStop}%
\bibitem [{\citenamefont {Ross}\ \emph {et~al.}(2017)\citenamefont {Ross},
  \citenamefont {Hughes},\ and\ \citenamefont {Doshi{-}Velez}}]{RossHD17}%
  \BibitemOpen
  \bibfield  {author} {\bibinfo {author} {\bibfnamefont {A.~S.}\ \bibnamefont
  {Ross}}, \bibinfo {author} {\bibfnamefont {M.~C.}\ \bibnamefont {Hughes}},\
  and\ \bibinfo {author} {\bibfnamefont {F.}~\bibnamefont {Doshi{-}Velez}},\
  }\bibfield  {title} {\bibinfo {title} {Right for the right reasons: Training
  differentiable models by constraining their explanations},\ }in\ \href
  {https://doi.org/10.24963/ijcai.2017/371} {\emph {\bibinfo {booktitle}
  {{IJCAI}}}}\ (\bibinfo  {publisher} {ijcai.org},\ \bibinfo {year} {2017})\
  pp.\ \bibinfo {pages} {2662--2670}\BibitemShut {NoStop}%
\bibitem [{\citenamefont {Anders}\ \emph {et~al.}(2022)\citenamefont {Anders},
  \citenamefont {Weber}, \citenamefont {Neumann}, \citenamefont {Samek},
  \citenamefont {M{\"{u}}ller},\ and\ \citenamefont
  {Lapuschkin}}]{AndersWNSML22}%
  \BibitemOpen
  \bibfield  {author} {\bibinfo {author} {\bibfnamefont {C.~J.}\ \bibnamefont
  {Anders}}, \bibinfo {author} {\bibfnamefont {L.}~\bibnamefont {Weber}},
  \bibinfo {author} {\bibfnamefont {D.}~\bibnamefont {Neumann}}, \bibinfo
  {author} {\bibfnamefont {W.}~\bibnamefont {Samek}}, \bibinfo {author}
  {\bibfnamefont {K.}~\bibnamefont {M{\"{u}}ller}},\ and\ \bibinfo {author}
  {\bibfnamefont {S.}~\bibnamefont {Lapuschkin}},\ }\bibfield  {title}
  {\bibinfo {title} {{Finding and removing clever Hans: Using explanation
  methods to debug and improve deep models}},\ }\href
  {https://doi.org/10.1016/j.inffus.2021.07.015} {\bibfield  {journal}
  {\bibinfo  {journal} {Inf. Fusion}\ }\textbf {\bibinfo {volume} {77}},\
  \bibinfo {pages} {261} (\bibinfo {year} {2022})}\BibitemShut {NoStop}%
\bibitem [{\citenamefont {Chormai}\ \emph {et~al.}(2024)\citenamefont
  {Chormai}, \citenamefont {Herrmann}, \citenamefont {M{\"{u}}ller},\ and\
  \citenamefont {Montavon}}]{ChormaiHMM24}%
  \BibitemOpen
  \bibfield  {author} {\bibinfo {author} {\bibfnamefont {P.}~\bibnamefont
  {Chormai}}, \bibinfo {author} {\bibfnamefont {J.}~\bibnamefont {Herrmann}},
  \bibinfo {author} {\bibfnamefont {K.}~\bibnamefont {M{\"{u}}ller}},\ and\
  \bibinfo {author} {\bibfnamefont {G.}~\bibnamefont {Montavon}},\ }\bibfield
  {title} {\bibinfo {title} {Disentangled explanations of neural network
  predictions by finding relevant subspaces},\ }\href
  {https://doi.org/10.1109/TPAMI.2024.3388275} {\bibfield  {journal} {\bibinfo
  {journal} {{IEEE} Trans. Pattern Anal. Mach. Intell.}\ }\textbf {\bibinfo
  {volume} {46}},\ \bibinfo {pages} {7283} (\bibinfo {year}
  {2024})}\BibitemShut {NoStop}%
\bibitem [{\citenamefont {Krenn}\ \emph {et~al.}(2022)\citenamefont {Krenn},
  \citenamefont {Pollice}, \citenamefont {Guo}, \citenamefont {Aldeghi},
  \citenamefont {Cervera-Lierta}, \citenamefont {Friederich}, \citenamefont
  {dos Passos~Gomes}, \citenamefont {H\"{a}se}, \citenamefont {Jinich},
  \citenamefont {Nigam}, \citenamefont {Yao},\ and\ \citenamefont
  {Aspuru-Guzik}}]{Krenn2022}%
  \BibitemOpen
  \bibfield  {author} {\bibinfo {author} {\bibfnamefont {M.}~\bibnamefont
  {Krenn}}, \bibinfo {author} {\bibfnamefont {R.}~\bibnamefont {Pollice}},
  \bibinfo {author} {\bibfnamefont {S.~Y.}\ \bibnamefont {Guo}}, \bibinfo
  {author} {\bibfnamefont {M.}~\bibnamefont {Aldeghi}}, \bibinfo {author}
  {\bibfnamefont {A.}~\bibnamefont {Cervera-Lierta}}, \bibinfo {author}
  {\bibfnamefont {P.}~\bibnamefont {Friederich}}, \bibinfo {author}
  {\bibfnamefont {G.}~\bibnamefont {dos Passos~Gomes}}, \bibinfo {author}
  {\bibfnamefont {F.}~\bibnamefont {H\"{a}se}}, \bibinfo {author}
  {\bibfnamefont {A.}~\bibnamefont {Jinich}}, \bibinfo {author} {\bibfnamefont
  {A.}~\bibnamefont {Nigam}}, \bibinfo {author} {\bibfnamefont
  {Z.}~\bibnamefont {Yao}},\ and\ \bibinfo {author} {\bibfnamefont
  {A.}~\bibnamefont {Aspuru-Guzik}},\ }\bibfield  {title} {\bibinfo {title} {On
  scientific understanding with artificial intelligence},\ }\href
  {https://doi.org/10.1038/s42254-022-00518-3} {\bibfield  {journal} {\bibinfo
  {journal} {Nature Rev. Phys.}\ }\textbf {\bibinfo {volume} {4}},\ \bibinfo
  {pages} {761–769} (\bibinfo {year} {2022})}\BibitemShut {NoStop}%
\bibitem [{\citenamefont {Schramowski}\ \emph {et~al.}(2020)\citenamefont
  {Schramowski}, \citenamefont {Stammer}, \citenamefont {Teso}, \citenamefont
  {Brugger}, \citenamefont {Herbert}, \citenamefont {Shao}, \citenamefont
  {Luigs}, \citenamefont {Mahlein},\ and\ \citenamefont
  {Kersting}}]{SchramowskiSTBH20}%
  \BibitemOpen
  \bibfield  {author} {\bibinfo {author} {\bibfnamefont {P.}~\bibnamefont
  {Schramowski}}, \bibinfo {author} {\bibfnamefont {W.}~\bibnamefont
  {Stammer}}, \bibinfo {author} {\bibfnamefont {S.}~\bibnamefont {Teso}},
  \bibinfo {author} {\bibfnamefont {A.}~\bibnamefont {Brugger}}, \bibinfo
  {author} {\bibfnamefont {F.}~\bibnamefont {Herbert}}, \bibinfo {author}
  {\bibfnamefont {X.}~\bibnamefont {Shao}}, \bibinfo {author} {\bibfnamefont
  {H.}~\bibnamefont {Luigs}}, \bibinfo {author} {\bibfnamefont
  {A.}~\bibnamefont {Mahlein}},\ and\ \bibinfo {author} {\bibfnamefont
  {K.}~\bibnamefont {Kersting}},\ }\bibfield  {title} {\bibinfo {title} {Making
  deep neural networks right for the right scientific reasons by interacting
  with their explanations},\ }\href {https://doi.org/10.1038/s42256-020-0212-3}
  {\bibfield  {journal} {\bibinfo  {journal} {Nature Mach. Intell.}\ }\textbf
  {\bibinfo {volume} {2}},\ \bibinfo {pages} {476} (\bibinfo {year}
  {2020})}\BibitemShut {NoStop}%
\bibitem [{\citenamefont {Dreyer}\ \emph {et~al.}(2024)\citenamefont {Dreyer},
  \citenamefont {Achtibat}, \citenamefont {Samek},\ and\ \citenamefont
  {Lapuschkin}}]{Dreyer2024CVPR}%
  \BibitemOpen
  \bibfield  {author} {\bibinfo {author} {\bibfnamefont {M.}~\bibnamefont
  {Dreyer}}, \bibinfo {author} {\bibfnamefont {R.}~\bibnamefont {Achtibat}},
  \bibinfo {author} {\bibfnamefont {W.}~\bibnamefont {Samek}},\ and\ \bibinfo
  {author} {\bibfnamefont {S.}~\bibnamefont {Lapuschkin}},\ }\bibfield  {title}
  {\bibinfo {title} {Understanding the (extra-)ordinary: Validating deep model
  decisions with prototypical concept-based explanations},\ }in\ \href
  {https://openaccess.thecvf.com/content/CVPR2024W/SAIAD/html/Dreyer_Understanding_the_Extra-Ordinary_Validating_Deep_Model_Decisions_with_Prototypical_Concept-based_CVPRW_2024_paper.html}
  {\emph {\bibinfo {booktitle} {Proceedings of the IEEE/CVF Conference on
  Computer Vision and Pattern Recognition (CVPR) Workshops}}}\ (\bibinfo {year}
  {2024})\ pp.\ \bibinfo {pages} {3491--3501}\BibitemShut {NoStop}%
\bibitem [{\citenamefont {Bach}\ \emph {et~al.}(2015)\citenamefont {Bach},
  \citenamefont {Binder}, \citenamefont {Montavon}, \citenamefont {Klauschen},
  \citenamefont {M{\"u}ller},\ and\ \citenamefont {Samek}}]{bach2015pixel}%
  \BibitemOpen
  \bibfield  {author} {\bibinfo {author} {\bibfnamefont {S.}~\bibnamefont
  {Bach}}, \bibinfo {author} {\bibfnamefont {A.}~\bibnamefont {Binder}},
  \bibinfo {author} {\bibfnamefont {G.}~\bibnamefont {Montavon}}, \bibinfo
  {author} {\bibfnamefont {F.}~\bibnamefont {Klauschen}}, \bibinfo {author}
  {\bibfnamefont {K.-R.}\ \bibnamefont {M{\"u}ller}},\ and\ \bibinfo {author}
  {\bibfnamefont {W.}~\bibnamefont {Samek}},\ }\bibfield  {title} {\bibinfo
  {title} {On pixel-wise explanations for non-linear classifier decisions by
  layer-wise relevance propagation},\ }\href
  {https://doi.org/10.1371/journal.pone.0130140} {\bibfield  {journal}
  {\bibinfo  {journal} {PloS one}\ }\textbf {\bibinfo {volume} {10}},\ \bibinfo
  {pages} {e0130140} (\bibinfo {year} {2015})}\BibitemShut {NoStop}%
\bibitem [{\citenamefont {Schnake}\ \emph {et~al.}(2022)\citenamefont
  {Schnake}, \citenamefont {Eberle}, \citenamefont {Lederer}, \citenamefont
  {Nakajima}, \citenamefont {Schutt}, \citenamefont {Muller},\ and\
  \citenamefont {Montavon}}]{SchnakeHigherOrderExplanationsGraph2022}%
  \BibitemOpen
  \bibfield  {author} {\bibinfo {author} {\bibfnamefont {T.}~\bibnamefont
  {Schnake}}, \bibinfo {author} {\bibfnamefont {O.}~\bibnamefont {Eberle}},
  \bibinfo {author} {\bibfnamefont {J.}~\bibnamefont {Lederer}}, \bibinfo
  {author} {\bibfnamefont {S.}~\bibnamefont {Nakajima}}, \bibinfo {author}
  {\bibfnamefont {K.~T.}\ \bibnamefont {Schutt}}, \bibinfo {author}
  {\bibfnamefont {K.-R.}\ \bibnamefont {Muller}},\ and\ \bibinfo {author}
  {\bibfnamefont {G.}~\bibnamefont {Montavon}},\ }\bibfield  {title} {\bibinfo
  {title} {Higher-{{order explanations}} of {{graph neural networks}} via
  {{relevant walks}}},\ }\href {https://doi.org/10.1109/TPAMI.2021.3115452}
  {\bibfield  {journal} {\bibinfo  {journal} {IEEE Trans. Pattern Anal. Mach.
  Intell.}\ }\textbf {\bibinfo {volume} {44}},\ \bibinfo {pages} {7581}
  (\bibinfo {year} {2022})}\BibitemShut {NoStop}%
\bibitem [{\citenamefont {Kim}\ \emph {et~al.}(2018)\citenamefont {Kim},
  \citenamefont {Wattenberg}, \citenamefont {Gilmer}, \citenamefont {Cai},
  \citenamefont {Wexler}, \citenamefont {Vi{\'{e}}gas},\ and\ \citenamefont
  {Sayres}}]{KimWGCWVS18}%
  \BibitemOpen
  \bibfield  {author} {\bibinfo {author} {\bibfnamefont {B.}~\bibnamefont
  {Kim}}, \bibinfo {author} {\bibfnamefont {M.}~\bibnamefont {Wattenberg}},
  \bibinfo {author} {\bibfnamefont {J.}~\bibnamefont {Gilmer}}, \bibinfo
  {author} {\bibfnamefont {C.~J.}\ \bibnamefont {Cai}}, \bibinfo {author}
  {\bibfnamefont {J.}~\bibnamefont {Wexler}}, \bibinfo {author} {\bibfnamefont
  {F.~B.}\ \bibnamefont {Vi{\'{e}}gas}},\ and\ \bibinfo {author} {\bibfnamefont
  {R.}~\bibnamefont {Sayres}},\ }\bibfield  {title} {\bibinfo {title}
  {Interpretability beyond feature attribution: Quantitative testing with
  concept activation vectors {(TCAV)}},\ }in\ \href
  {https://proceedings.mlr.press/v80/kim18d.html} {\emph {\bibinfo {booktitle}
  {{ICML}}}},\ \bibinfo {series} {Proceedings of Machine Learning Research},
  Vol.~\bibinfo {volume} {80}\ (\bibinfo  {publisher} {{PMLR}},\ \bibinfo
  {year} {2018})\ pp.\ \bibinfo {pages} {2673--2682}\BibitemShut {NoStop}%
\bibitem [{\citenamefont {Achtibat}\ \emph {et~al.}(2023)\citenamefont
  {Achtibat}, \citenamefont {Dreyer}, \citenamefont {Eisenbraun}, \citenamefont
  {Bosse}, \citenamefont {Wiegand}, \citenamefont {Samek},\ and\ \citenamefont
  {Lapuschkin}}]{achtibat_attribution_2023}%
  \BibitemOpen
  \bibfield  {author} {\bibinfo {author} {\bibfnamefont {R.}~\bibnamefont
  {Achtibat}}, \bibinfo {author} {\bibfnamefont {M.}~\bibnamefont {Dreyer}},
  \bibinfo {author} {\bibfnamefont {I.}~\bibnamefont {Eisenbraun}}, \bibinfo
  {author} {\bibfnamefont {S.}~\bibnamefont {Bosse}}, \bibinfo {author}
  {\bibfnamefont {T.}~\bibnamefont {Wiegand}}, \bibinfo {author} {\bibfnamefont
  {W.}~\bibnamefont {Samek}},\ and\ \bibinfo {author} {\bibfnamefont
  {S.}~\bibnamefont {Lapuschkin}},\ }\bibfield  {title} {\bibinfo {title} {From
  attribution maps to human-understandable explanations through concept
  relevance propagation},\ }\href {https://doi.org/10.1038/s42256-023-00711-8}
  {\bibfield  {journal} {\bibinfo  {journal} {Nature Mach. Intell.}\ }\textbf
  {\bibinfo {volume} {5}},\ \bibinfo {pages} {1006} (\bibinfo {year}
  {2023})}\BibitemShut {NoStop}%
\bibitem [{\citenamefont {Vielhaben}\ \emph {et~al.}(2024)\citenamefont
  {Vielhaben}, \citenamefont {Lapuschkin}, \citenamefont {Montavon},\ and\
  \citenamefont {Samek}}]{VielhabenLMS24}%
  \BibitemOpen
  \bibfield  {author} {\bibinfo {author} {\bibfnamefont {J.}~\bibnamefont
  {Vielhaben}}, \bibinfo {author} {\bibfnamefont {S.}~\bibnamefont
  {Lapuschkin}}, \bibinfo {author} {\bibfnamefont {G.}~\bibnamefont
  {Montavon}},\ and\ \bibinfo {author} {\bibfnamefont {W.}~\bibnamefont
  {Samek}},\ }\bibfield  {title} {\bibinfo {title} {Explainable {AI} for time
  series via virtual inspection layers},\ }\href
  {https://doi.org/10.1016/j.patcog.2024.110309} {\bibfield  {journal}
  {\bibinfo  {journal} {Pattern Recognit.}\ }\textbf {\bibinfo {volume}
  {150}},\ \bibinfo {pages} {110309} (\bibinfo {year} {2024})}\BibitemShut
  {NoStop}%
\bibitem [{\citenamefont {Lou}\ \emph {et~al.}(2012)\citenamefont {Lou},
  \citenamefont {Caruana},\ and\ \citenamefont {Gehrke}}]{lou2012intelligible}%
  \BibitemOpen
  \bibfield  {author} {\bibinfo {author} {\bibfnamefont {Y.}~\bibnamefont
  {Lou}}, \bibinfo {author} {\bibfnamefont {R.}~\bibnamefont {Caruana}},\ and\
  \bibinfo {author} {\bibfnamefont {J.}~\bibnamefont {Gehrke}},\ }\bibfield
  {title} {\bibinfo {title} {Intelligible models for classification and
  regression},\ }in\ \href {https://doi.org/10.1145/2339530.2339556} {\emph
  {\bibinfo {booktitle} {Proceedings of the 18th ACM SIGKDD International
  Conference on Knowledge Discovery and Data Mining}}},\ \bibinfo {series and
  number} {KDD '12}\ (\bibinfo  {publisher} {Association for Computing
  Machinery},\ \bibinfo {address} {New York, NY, USA},\ \bibinfo {year}
  {2012})\ p.\ \bibinfo {pages} {150–158}\BibitemShut {NoStop}%
\bibitem [{\citenamefont {Zhou}\ \emph {et~al.}(2016)\citenamefont {Zhou},
  \citenamefont {Khosla}, \citenamefont {Lapedriza}, \citenamefont {Oliva},\
  and\ \citenamefont {Torralba}}]{ZhouKLOT16}%
  \BibitemOpen
  \bibfield  {author} {\bibinfo {author} {\bibfnamefont {B.}~\bibnamefont
  {Zhou}}, \bibinfo {author} {\bibfnamefont {A.}~\bibnamefont {Khosla}},
  \bibinfo {author} {\bibfnamefont {{\`{A}}.}~\bibnamefont {Lapedriza}},
  \bibinfo {author} {\bibfnamefont {A.}~\bibnamefont {Oliva}},\ and\ \bibinfo
  {author} {\bibfnamefont {A.}~\bibnamefont {Torralba}},\ }\bibfield  {title}
  {\bibinfo {title} {Learning deep features for discriminative localization},\
  }in\ \href
  {https://openaccess.thecvf.com/content_cvpr_2016/html/Zhou_Learning_Deep_Features_CVPR_2016_paper.html}
  {\emph {\bibinfo {booktitle} {{CVPR}}}}\ (\bibinfo  {publisher} {{IEEE}
  Computer Society},\ \bibinfo {year} {2016})\ pp.\ \bibinfo {pages}
  {2921--2929}\BibitemShut {NoStop}%
\bibitem [{\citenamefont {Vaswani}\ \emph {et~al.}(2017)\citenamefont
  {Vaswani}, \citenamefont {Shazeer}, \citenamefont {Parmar}, \citenamefont
  {Uszkoreit}, \citenamefont {Jones}, \citenamefont {Gomez}, \citenamefont
  {Kaiser},\ and\ \citenamefont {Polosukhin}}]{vaswani2017attention}%
  \BibitemOpen
  \bibfield  {author} {\bibinfo {author} {\bibfnamefont {A.}~\bibnamefont
  {Vaswani}}, \bibinfo {author} {\bibfnamefont {N.}~\bibnamefont {Shazeer}},
  \bibinfo {author} {\bibfnamefont {N.}~\bibnamefont {Parmar}}, \bibinfo
  {author} {\bibfnamefont {J.}~\bibnamefont {Uszkoreit}}, \bibinfo {author}
  {\bibfnamefont {L.}~\bibnamefont {Jones}}, \bibinfo {author} {\bibfnamefont
  {A.~N.}\ \bibnamefont {Gomez}}, \bibinfo {author} {\bibfnamefont
  {L.}~\bibnamefont {Kaiser}},\ and\ \bibinfo {author} {\bibfnamefont
  {I.}~\bibnamefont {Polosukhin}},\ }\bibfield  {title} {\bibinfo {title}
  {Attention is all you need},\ }in\ \href
  {https://papers.nips.cc/paper_files/paper/2017/hash/3f5ee243547dee91fbd053c1c4a845aa-Abstract.html}
  {\emph {\bibinfo {booktitle} {Proceedings of the 31st International
  Conference on Neural Information Processing Systems}}},\ \bibinfo {series and
  number} {NIPS'17}\ (\bibinfo  {publisher} {Curran Associates Inc.},\ \bibinfo
  {address} {Red Hook, NY, USA},\ \bibinfo {year} {2017})\ p.\ \bibinfo {pages}
  {6000–6010}\BibitemShut {NoStop}%
\bibitem [{\citenamefont {Chen}\ \emph {et~al.}(2018)\citenamefont {Chen},
  \citenamefont {Rubanova}, \citenamefont {Bettencourt},\ and\ \citenamefont
  {Duvenaud}}]{ChenRBD18}%
  \BibitemOpen
  \bibfield  {author} {\bibinfo {author} {\bibfnamefont {T.~Q.}\ \bibnamefont
  {Chen}}, \bibinfo {author} {\bibfnamefont {Y.}~\bibnamefont {Rubanova}},
  \bibinfo {author} {\bibfnamefont {J.}~\bibnamefont {Bettencourt}},\ and\
  \bibinfo {author} {\bibfnamefont {D.}~\bibnamefont {Duvenaud}},\ }\bibfield
  {title} {\bibinfo {title} {Neural ordinary differential equations},\ }in\
  \href
  {https://proceedings.neurips.cc/paper/2018/hash/69386f6bb1dfed68692a24c8686939b9-Abstract.html}
  {\emph {\bibinfo {booktitle} {NeurIPS}}}\ (\bibinfo {year} {2018})\ pp.\
  \bibinfo {pages} {6572--6583}\BibitemShut {NoStop}%
\bibitem [{\citenamefont {Sundararajan}\ \emph {et~al.}(2017)\citenamefont
  {Sundararajan}, \citenamefont {Taly},\ and\ \citenamefont
  {Yan}}]{sundararajan2017axiomatic}%
  \BibitemOpen
  \bibfield  {author} {\bibinfo {author} {\bibfnamefont {M.}~\bibnamefont
  {Sundararajan}}, \bibinfo {author} {\bibfnamefont {A.}~\bibnamefont {Taly}},\
  and\ \bibinfo {author} {\bibfnamefont {Q.}~\bibnamefont {Yan}},\ }\bibfield
  {title} {\bibinfo {title} {Axiomatic attribution for deep networks},\ }in\
  \href {https://proceedings.mlr.press/v70/sundararajan17a.html} {\emph
  {\bibinfo {booktitle} {International conference on machine learning}}}\
  (\bibinfo {organization} {PMLR},\ \bibinfo {year} {2017})\ pp.\ \bibinfo
  {pages} {3319--3328}\BibitemShut {NoStop}%
\bibitem [{\citenamefont {{\v{S}}trumbelj}\ and\ \citenamefont
  {Kononenko}(2010)}]{strumbelj2010efficient}%
  \BibitemOpen
  \bibfield  {author} {\bibinfo {author} {\bibfnamefont {E.}~\bibnamefont
  {{\v{S}}trumbelj}}\ and\ \bibinfo {author} {\bibfnamefont {I.}~\bibnamefont
  {Kononenko}},\ }\bibfield  {title} {\bibinfo {title} {An efficient
  explanation of individual classifications using game theory},\ }\href
  {http://jmlr.org/papers/v11/strumbelj10a.html} {\bibfield  {journal}
  {\bibinfo  {journal} {J. Mach. Learn. Res.}\ }\textbf {\bibinfo {volume}
  {11}},\ \bibinfo {pages} {1} (\bibinfo {year} {2010})}\BibitemShut {NoStop}%
\bibitem [{\citenamefont {Dombrowski}\ \emph {et~al.}(2023)\citenamefont
  {Dombrowski}, \citenamefont {Gerken}, \citenamefont {M{\"u}ller},\ and\
  \citenamefont {Kessel}}]{dombrowski2023diffeomorphic}%
  \BibitemOpen
  \bibfield  {author} {\bibinfo {author} {\bibfnamefont {A.-K.}\ \bibnamefont
  {Dombrowski}}, \bibinfo {author} {\bibfnamefont {J.~E.}\ \bibnamefont
  {Gerken}}, \bibinfo {author} {\bibfnamefont {K.-R.}\ \bibnamefont
  {M{\"u}ller}},\ and\ \bibinfo {author} {\bibfnamefont {P.}~\bibnamefont
  {Kessel}},\ }\bibfield  {title} {\bibinfo {title} {Diffeomorphic
  counterfactuals with generative models},\ }\href
  {https://ieeexplore.ieee.org/abstract/document/10345703} {\bibfield
  {journal} {\bibinfo  {journal} {IEEE Trans. Patt. Ana. Mach. Int.}\ }
  (\bibinfo {year} {2023})}\BibitemShut {NoStop}%
\bibitem [{\citenamefont {Nguyen}\ \emph {et~al.}(2016)\citenamefont {Nguyen},
  \citenamefont {Dosovitskiy}, \citenamefont {Yosinski}, \citenamefont {Brox},\
  and\ \citenamefont {Clune}}]{nguyen2016synthesizing}%
  \BibitemOpen
  \bibfield  {author} {\bibinfo {author} {\bibfnamefont {A.}~\bibnamefont
  {Nguyen}}, \bibinfo {author} {\bibfnamefont {A.}~\bibnamefont {Dosovitskiy}},
  \bibinfo {author} {\bibfnamefont {J.}~\bibnamefont {Yosinski}}, \bibinfo
  {author} {\bibfnamefont {T.}~\bibnamefont {Brox}},\ and\ \bibinfo {author}
  {\bibfnamefont {J.}~\bibnamefont {Clune}},\ }\bibfield  {title} {\bibinfo
  {title} {Synthesizing the preferred inputs for neurons in neural networks via
  deep generator networks},\ }\href {https://doi.org/10.48550/arXiv.1605.09304}
  {\bibfield  {journal} {\bibinfo  {journal} {Adv. Neu. Inf. Proc. Sys.
  (NeurIPS)}\ }\textbf {\bibinfo {volume} {29}},\ \bibinfo {pages} {3395}
  (\bibinfo {year} {2016})}\BibitemShut {NoStop}%
\bibitem [{\citenamefont {{Google Research}}(2015)}]{google2015deepdream}%
  \BibitemOpen
  \bibfield  {author} {\bibinfo {author} {\bibnamefont {{Google Research}}},\
  }\href@noop {} {\bibinfo {title} {Inceptionism: Going deeper into neural
  networks}},\ \bibinfo {howpublished}
  {\url{https://ai.googleblog.com/2015/06/inceptionism-going-deeper-into-neural.html}}
  (\bibinfo {year} {2015}),\ \bibinfo {note} {accessed: 2024-08-01}\BibitemShut
  {NoStop}%
\bibitem [{\citenamefont {Shapley}(1953)}]{shapley1953value}%
  \BibitemOpen
  \bibfield  {author} {\bibinfo {author} {\bibfnamefont {L.~S.}\ \bibnamefont
  {Shapley}},\ }\bibfield  {title} {\bibinfo {title} {A value for $n$-person
  games},\ }in\ \href {http://digital.casalini.it/9781400829156} {\emph
  {\bibinfo {booktitle} {Contributions to the Theory of Games (AM-28), Volume
  II}}}\ (\bibinfo  {publisher} {Princeton University Press},\ \bibinfo {year}
  {1953})\ pp.\ \bibinfo {pages} {307--317}\BibitemShut {NoStop}%
\bibitem [{\citenamefont {Lundberg}\ and\ \citenamefont
  {Lee}(2017)}]{lundbergUnifiedApproachInterpreting2017}%
  \BibitemOpen
  \bibfield  {author} {\bibinfo {author} {\bibfnamefont {S.~M.}\ \bibnamefont
  {Lundberg}}\ and\ \bibinfo {author} {\bibfnamefont {S.}~\bibnamefont {Lee}},\
  }\bibfield  {title} {\bibinfo {title} {A unified approach to interpreting
  model predictions},\ }\href {https://dl.acm.org/doi/10.5555/3295222.3295230}
  {\bibfield  {journal} {\bibinfo  {journal} {{NIPS}}\ ,\ \bibinfo {pages}
  {4765}} (\bibinfo {year} {2017})}\BibitemShut {NoStop}%
\bibitem [{\citenamefont {Sundararajan}\ and\ \citenamefont
  {Najmi}(2020)}]{sundararajanManyShapleyValues2020}%
  \BibitemOpen
  \bibfield  {author} {\bibinfo {author} {\bibfnamefont {M.}~\bibnamefont
  {Sundararajan}}\ and\ \bibinfo {author} {\bibfnamefont {A.}~\bibnamefont
  {Najmi}},\ }\bibfield  {title} {\bibinfo {title} {The many shapley values for
  model explanation},\ }in\ \href
  {https://proceedings.mlr.press/v119/sundararajan20b.html} {\emph {\bibinfo
  {booktitle} {International conference on machine learning}}}\ (\bibinfo
  {organization} {PMLR},\ \bibinfo {year} {2020})\ pp.\ \bibinfo {pages}
  {9269--9278}\BibitemShut {NoStop}%
\bibitem [{\citenamefont {Castro}\ \emph {et~al.}(2009)\citenamefont {Castro},
  \citenamefont {G{\'o}mez},\ and\ \citenamefont
  {Tejada}}]{CastroPolynomialCalculationShapley2009}%
  \BibitemOpen
  \bibfield  {author} {\bibinfo {author} {\bibfnamefont {J.}~\bibnamefont
  {Castro}}, \bibinfo {author} {\bibfnamefont {D.}~\bibnamefont {G{\'o}mez}},\
  and\ \bibinfo {author} {\bibfnamefont {J.}~\bibnamefont {Tejada}},\
  }\bibfield  {title} {\bibinfo {title} {Polynomial calculation of the
  {{Shapley}} value based on sampling},\ }\href
  {https://doi.org/10.1016/j.cor.2008.04.004} {\bibfield  {journal} {\bibinfo
  {journal} {Computers \& Operations Research}\ }\bibinfo {series} {Selected
  Papers Presented at the {{Tenth International Symposium}} on {{Locational
  Decisions}} ({{ISOLDE X}})},\ \textbf {\bibinfo {volume} {36}},\ \bibinfo
  {pages} {1726} (\bibinfo {year} {2009})}\BibitemShut {NoStop}%
\bibitem [{\citenamefont {Ancona}\ \emph {et~al.}(2019)\citenamefont {Ancona},
  \citenamefont {Oztireli},\ and\ \citenamefont
  {Gross}}]{AnconaExplainingDeepNeural2019a}%
  \BibitemOpen
  \bibfield  {author} {\bibinfo {author} {\bibfnamefont {M.}~\bibnamefont
  {Ancona}}, \bibinfo {author} {\bibfnamefont {C.}~\bibnamefont {Oztireli}},\
  and\ \bibinfo {author} {\bibfnamefont {M.}~\bibnamefont {Gross}},\ }\bibfield
   {title} {\bibinfo {title} {Explaining deep neural networks with a polynomial
  time algorithm for shapley value approximation},\ }in\ \href
  {https://proceedings.mlr.press/v97/ancona19a.html} {\emph {\bibinfo
  {booktitle} {International Conference on Machine Learning}}}\ (\bibinfo
  {organization} {PMLR},\ \bibinfo {year} {2019})\ pp.\ \bibinfo {pages}
  {272--281}\BibitemShut {NoStop}%
\bibitem [{\citenamefont {Rumelhart}\ \emph {et~al.}(1986)\citenamefont
  {Rumelhart}, \citenamefont {Hinton},\ and\ \citenamefont
  {Williams}}]{Rumelhart1986}%
  \BibitemOpen
  \bibfield  {author} {\bibinfo {author} {\bibfnamefont {D.~E.}\ \bibnamefont
  {Rumelhart}}, \bibinfo {author} {\bibfnamefont {G.~E.}\ \bibnamefont
  {Hinton}},\ and\ \bibinfo {author} {\bibfnamefont {R.~J.}\ \bibnamefont
  {Williams}},\ }\bibfield  {title} {\bibinfo {title} {Learning representations
  by back-propagating errors},\ }\href {https://doi.org/10.1038/323533a0}
  {\bibfield  {journal} {\bibinfo  {journal} {Nature}\ }\textbf {\bibinfo
  {volume} {323}},\ \bibinfo {pages} {533–536} (\bibinfo {year}
  {1986})}\BibitemShut {NoStop}%
\bibitem [{\citenamefont {Montavon}\ \emph {et~al.}(2020)\citenamefont
  {Montavon}, \citenamefont {Kauffmann}, \citenamefont {Samek},\ and\
  \citenamefont {M{\"{u}}ller}}]{MontavonKSM20}%
  \BibitemOpen
  \bibfield  {author} {\bibinfo {author} {\bibfnamefont {G.}~\bibnamefont
  {Montavon}}, \bibinfo {author} {\bibfnamefont {J.~R.}\ \bibnamefont
  {Kauffmann}}, \bibinfo {author} {\bibfnamefont {W.}~\bibnamefont {Samek}},\
  and\ \bibinfo {author} {\bibfnamefont {K.}~\bibnamefont {M{\"{u}}ller}},\
  }\bibfield  {title} {\bibinfo {title} {Explaining the predictions of
  unsupervised learning models},\ }in\ \href
  {https://doi.org/10.1007/978-3-031-04083-2_7} {\emph {\bibinfo {booktitle}
  {xxAI@ICML}}},\ \bibinfo {series} {Lecture Notes in Computer Science}, Vol.\
  \bibinfo {volume} {13200}\ (\bibinfo  {publisher} {Springer},\ \bibinfo
  {year} {2020})\ pp.\ \bibinfo {pages} {117--138}\BibitemShut {NoStop}%
\bibitem [{\citenamefont {Ali}\ \emph {et~al.}(2022)\citenamefont {Ali},
  \citenamefont {Schnake}, \citenamefont {Eberle}, \citenamefont {Montavon},
  \citenamefont {M{\"{u}}ller},\ and\ \citenamefont {Wolf}}]{AliSEMMW22}%
  \BibitemOpen
  \bibfield  {author} {\bibinfo {author} {\bibfnamefont {A.}~\bibnamefont
  {Ali}}, \bibinfo {author} {\bibfnamefont {T.}~\bibnamefont {Schnake}},
  \bibinfo {author} {\bibfnamefont {O.}~\bibnamefont {Eberle}}, \bibinfo
  {author} {\bibfnamefont {G.}~\bibnamefont {Montavon}}, \bibinfo {author}
  {\bibfnamefont {K.}~\bibnamefont {M{\"{u}}ller}},\ and\ \bibinfo {author}
  {\bibfnamefont {L.}~\bibnamefont {Wolf}},\ }\bibfield  {title} {\bibinfo
  {title} {{XAI} for transformers: Better explanations through conservative
  propagation},\ }in\ \href {https://proceedings.mlr.press/v162/ali22a.html}
  {\emph {\bibinfo {booktitle} {{ICML}}}},\ \bibinfo {series} {Proceedings of
  Machine Learning Research}, Vol.\ \bibinfo {volume} {162}\ (\bibinfo
  {publisher} {{PMLR}},\ \bibinfo {year} {2022})\ pp.\ \bibinfo {pages}
  {435--451}\BibitemShut {NoStop}%
\bibitem [{\citenamefont {Montavon}\ \emph {et~al.}(2019)\citenamefont
  {Montavon}, \citenamefont {Binder}, \citenamefont {Lapuschkin}, \citenamefont
  {Samek},\ and\ \citenamefont
  {Müller}}]{montavonLayerWiseRelevancePropagation2019}%
  \BibitemOpen
  \bibfield  {author} {\bibinfo {author} {\bibfnamefont {G.}~\bibnamefont
  {Montavon}}, \bibinfo {author} {\bibfnamefont {A.}~\bibnamefont {Binder}},
  \bibinfo {author} {\bibfnamefont {S.}~\bibnamefont {Lapuschkin}}, \bibinfo
  {author} {\bibfnamefont {W.}~\bibnamefont {Samek}},\ and\ \bibinfo {author}
  {\bibfnamefont {K.-R.}\ \bibnamefont {Müller}},\ }\bibfield  {title}
  {\bibinfo {title} {Layer-{{wise relevance propagation}}: {{An overview}}},\
  }in\ \href {https://doi.org/10.1007/978-3-030-28954-6_10} {\emph {\bibinfo
  {booktitle} {Explainable {{AI}}: {{Interpreting}}, {{Explaining}} and
  {{Visualizing Deep Learning}}}}},\ Vol.\ \bibinfo {volume} {11700},\ \bibinfo
  {editor} {edited by\ \bibinfo {editor} {\bibfnamefont {W.}~\bibnamefont
  {Samek}}, \bibinfo {editor} {\bibfnamefont {G.}~\bibnamefont {Montavon}},
  \bibinfo {editor} {\bibfnamefont {A.}~\bibnamefont {Vedaldi}}, \bibinfo
  {editor} {\bibfnamefont {L.~K.}\ \bibnamefont {Hansen}},\ and\ \bibinfo
  {editor} {\bibfnamefont {K.-R.}\ \bibnamefont {Müller}}}\ (\bibinfo
  {publisher} {Springer International Publishing},\ \bibinfo {year} {2019})\
  pp.\ \bibinfo {pages} {193--209}\BibitemShut {NoStop}%
\bibitem [{\citenamefont {Samek}\ \emph {et~al.}(2017)\citenamefont {Samek},
  \citenamefont {Binder}, \citenamefont {Montavon}, \citenamefont
  {Lapuschkin},\ and\ \citenamefont {M{\"{u}}ller}}]{SamekBMLM17}%
  \BibitemOpen
  \bibfield  {author} {\bibinfo {author} {\bibfnamefont {W.}~\bibnamefont
  {Samek}}, \bibinfo {author} {\bibfnamefont {A.}~\bibnamefont {Binder}},
  \bibinfo {author} {\bibfnamefont {G.}~\bibnamefont {Montavon}}, \bibinfo
  {author} {\bibfnamefont {S.}~\bibnamefont {Lapuschkin}},\ and\ \bibinfo
  {author} {\bibfnamefont {K.}~\bibnamefont {M{\"{u}}ller}},\ }\bibfield
  {title} {\bibinfo {title} {Evaluating the visualization of what a deep neural
  network has learned},\ }\href {https://doi.org/10.1109/TNNLS.2016.2599820}
  {\bibfield  {journal} {\bibinfo  {journal} {{IEEE} Trans. Neural Networks
  Learn. Syst.}\ }\textbf {\bibinfo {volume} {28}},\ \bibinfo {pages} {2660}
  (\bibinfo {year} {2017})}\BibitemShut {NoStop}%
\bibitem [{\citenamefont {Arras}\ \emph {et~al.}(2022)\citenamefont {Arras},
  \citenamefont {Osman},\ and\ \citenamefont {Samek}}]{Arras2022clevrxai}%
  \BibitemOpen
  \bibfield  {author} {\bibinfo {author} {\bibfnamefont {L.}~\bibnamefont
  {Arras}}, \bibinfo {author} {\bibfnamefont {A.}~\bibnamefont {Osman}},\ and\
  \bibinfo {author} {\bibfnamefont {W.}~\bibnamefont {Samek}},\ }\bibfield
  {title} {\bibinfo {title} {{CLEVR-XAI: A benchmark dataset for the ground
  truth evaluation of neural network explanations}},\ }\href
  {https://doi.org/10.1016/j.inffus.2021.11.008} {\bibfield  {journal}
  {\bibinfo  {journal} {Information Fusion}\ }\textbf {\bibinfo {volume}
  {81}},\ \bibinfo {pages} {14–40} (\bibinfo {year} {2022})}\BibitemShut
  {NoStop}%
\bibitem [{\citenamefont {Hedström}\ \emph {et~al.}(2023)\citenamefont
  {Hedström}, \citenamefont {Weber}, \citenamefont {Bareeva}, \citenamefont
  {Krakowczyk}, \citenamefont {Motzkus}, \citenamefont {Samek}, \citenamefont
  {Lapuschkin},\ and\ \citenamefont
  {Höhne}}]{hedstromQuantusExplainableAI2023}%
  \BibitemOpen
  \bibfield  {author} {\bibinfo {author} {\bibfnamefont {A.}~\bibnamefont
  {Hedström}}, \bibinfo {author} {\bibfnamefont {L.}~\bibnamefont {Weber}},
  \bibinfo {author} {\bibfnamefont {D.}~\bibnamefont {Bareeva}}, \bibinfo
  {author} {\bibfnamefont {D.}~\bibnamefont {Krakowczyk}}, \bibinfo {author}
  {\bibfnamefont {F.}~\bibnamefont {Motzkus}}, \bibinfo {author} {\bibfnamefont
  {W.}~\bibnamefont {Samek}}, \bibinfo {author} {\bibfnamefont
  {S.}~\bibnamefont {Lapuschkin}},\ and\ \bibinfo {author} {\bibfnamefont
  {M.~M.-C.}\ \bibnamefont {Höhne}},\ }\bibfield  {title} {\bibinfo {title}
  {Quantus: An explainable ai toolkit for responsible evaluation of neural
  network explanations and beyond},\ }\href
  {http://jmlr.org/papers/v24/22-0142.html} {\bibfield  {journal} {\bibinfo
  {journal} {J. Mach. Learn. Res.}\ }\textbf {\bibinfo {volume} {24}},\
  \bibinfo {pages} {1} (\bibinfo {year} {2023})}\BibitemShut {NoStop}%
\bibitem [{\citenamefont {Swartout}\ and\ \citenamefont
  {Moore}(1993)}]{swartout1993explanation}%
  \BibitemOpen
  \bibfield  {author} {\bibinfo {author} {\bibfnamefont {W.~R.}\ \bibnamefont
  {Swartout}}\ and\ \bibinfo {author} {\bibfnamefont {J.~D.}\ \bibnamefont
  {Moore}},\ }\bibfield  {title} {\bibinfo {title} {Explanation in second
  generation expert systems},\ }in\ \href
  {https://doi.org/10.1007/978-3-642-77927-5_24} {\emph {\bibinfo {booktitle}
  {Second generation expert systems}}}\ (\bibinfo  {publisher} {Springer},\
  \bibinfo {year} {1993})\ pp.\ \bibinfo {pages} {543--585}\BibitemShut
  {NoStop}%
\bibitem [{\citenamefont {Huang}\ \emph {et~al.}(2020)\citenamefont {Huang},
  \citenamefont {Kueng},\ and\ \citenamefont {Preskill}}]{Shadows}%
  \BibitemOpen
  \bibfield  {author} {\bibinfo {author} {\bibfnamefont {H.-Y.}\ \bibnamefont
  {Huang}}, \bibinfo {author} {\bibfnamefont {R.}~\bibnamefont {Kueng}},\ and\
  \bibinfo {author} {\bibfnamefont {J.}~\bibnamefont {Preskill}},\ }\bibfield
  {title} {\bibinfo {title} {Predicting many properties of a quantum system
  from very few measurements},\ }\href
  {https://doi.org/10.1038/s41567-020-0932-7} {\bibfield  {journal} {\bibinfo
  {journal} {Nature Phys.}\ }\textbf {\bibinfo {volume} {16}},\ \bibinfo
  {pages} {1050} (\bibinfo {year} {2020})}\BibitemShut {NoStop}%
\bibitem [{\citenamefont {Cirac}\ \emph {et~al.}(2021)\citenamefont {Cirac},
  \citenamefont {P\'erez-Garc\'{\i}a}, \citenamefont {Schuch},\ and\
  \citenamefont {Verstraete}}]{RevModPhys.93.045003}%
  \BibitemOpen
  \bibfield  {author} {\bibinfo {author} {\bibfnamefont {J.~I.}\ \bibnamefont
  {Cirac}}, \bibinfo {author} {\bibfnamefont {D.}~\bibnamefont
  {P\'erez-Garc\'{\i}a}}, \bibinfo {author} {\bibfnamefont {N.}~\bibnamefont
  {Schuch}},\ and\ \bibinfo {author} {\bibfnamefont {F.}~\bibnamefont
  {Verstraete}},\ }\bibfield  {title} {\bibinfo {title} {Matrix product states
  and projected entangled pair states: Concepts, symmetries, theorems},\ }\href
  {https://doi.org/10.1103/RevModPhys.93.045003} {\bibfield  {journal}
  {\bibinfo  {journal} {Rev. Mod. Phys.}\ }\textbf {\bibinfo {volume} {93}},\
  \bibinfo {pages} {045003} (\bibinfo {year} {2021})}\BibitemShut {NoStop}%
\bibitem [{\citenamefont {Eisert}\ \emph {et~al.}(2010)\citenamefont {Eisert},
  \citenamefont {Cramer},\ and\ \citenamefont {Plenio}}]{AreaReview}%
  \BibitemOpen
  \bibfield  {author} {\bibinfo {author} {\bibfnamefont {J.}~\bibnamefont
  {Eisert}}, \bibinfo {author} {\bibfnamefont {M.}~\bibnamefont {Cramer}},\
  and\ \bibinfo {author} {\bibfnamefont {M.~B.}\ \bibnamefont {Plenio}},\
  }\bibfield  {title} {\bibinfo {title} {Area laws for the entanglement
  entropy},\ }\href {https://doi.org/10.1103/RevModPhys.82.277} {\bibfield
  {journal} {\bibinfo  {journal} {Rev. Mod. Phys.}\ }\textbf {\bibinfo {volume}
  {82}},\ \bibinfo {pages} {277} (\bibinfo {year} {2010})}\BibitemShut
  {NoStop}%
\bibitem [{\citenamefont {Gottesman}(1998)}]{PhysRevA.57.127}%
  \BibitemOpen
  \bibfield  {author} {\bibinfo {author} {\bibfnamefont {D.}~\bibnamefont
  {Gottesman}},\ }\bibfield  {title} {\bibinfo {title} {Theory of
  fault-tolerant quantum computation},\ }\href
  {https://doi.org/10.1103/PhysRevA.57.127} {\bibfield  {journal} {\bibinfo
  {journal} {Phys. Rev. A}\ }\textbf {\bibinfo {volume} {57}},\ \bibinfo
  {pages} {127} (\bibinfo {year} {1998})}\BibitemShut {NoStop}%
\bibitem [{\citenamefont {Masot-Llima}\ and\ \citenamefont
  {Garcia-Saez}(2024)}]{Sergi}%
  \BibitemOpen
  \bibfield  {author} {\bibinfo {author} {\bibfnamefont {S.}~\bibnamefont
  {Masot-Llima}}\ and\ \bibinfo {author} {\bibfnamefont {A.}~\bibnamefont
  {Garcia-Saez}},\ }\bibfield  {title} {\bibinfo {title} {Stabilizer tensor
  networks: universal quantum simulator on a basis of stabilizer states},\
  }\href {https://arxiv.org/abs/2403.08724} {\bibfield  {journal} {\bibinfo
  {journal} {arXiv:2403.08724}\ } (\bibinfo {year} {2024})}\BibitemShut
  {NoStop}%
\bibitem [{\citenamefont {Chinzei}\ \emph {et~al.}(2024)\citenamefont
  {Chinzei}, \citenamefont {Yamano}, \citenamefont {Tran}, \citenamefont
  {Endo},\ and\ \citenamefont {Oshima}}]{chinzei2024tradeoff}%
  \BibitemOpen
  \bibfield  {author} {\bibinfo {author} {\bibfnamefont {K.}~\bibnamefont
  {Chinzei}}, \bibinfo {author} {\bibfnamefont {S.}~\bibnamefont {Yamano}},
  \bibinfo {author} {\bibfnamefont {Q.~H.}\ \bibnamefont {Tran}}, \bibinfo
  {author} {\bibfnamefont {Y.}~\bibnamefont {Endo}},\ and\ \bibinfo {author}
  {\bibfnamefont {H.}~\bibnamefont {Oshima}},\ }\bibfield  {title} {\bibinfo
  {title} {Trade-off between gradient measurement efficiency and expressivity
  in deep quantum neural networks},\ }\href {https://arxiv.org/abs/2406.18316}
  {\bibfield  {journal} {\bibinfo  {journal} {arXiv:2406.18316}\ } (\bibinfo
  {year} {2024})}\BibitemShut {NoStop}%
\bibitem [{\citenamefont {Gil-Fuster}\ \emph {et~al.}(2024)\citenamefont
  {Gil-Fuster}, \citenamefont {Gyurik}, \citenamefont {Pérez-Salinas},\ and\
  \citenamefont {Dunjko}}]{gilfuster2024relation}%
  \BibitemOpen
  \bibfield  {author} {\bibinfo {author} {\bibfnamefont {E.}~\bibnamefont
  {Gil-Fuster}}, \bibinfo {author} {\bibfnamefont {C.}~\bibnamefont {Gyurik}},
  \bibinfo {author} {\bibfnamefont {A.}~\bibnamefont {Pérez-Salinas}},\ and\
  \bibinfo {author} {\bibfnamefont {V.}~\bibnamefont {Dunjko}},\ }\bibfield
  {title} {\bibinfo {title} {On the relation between trainability and
  dequantization of variational quantum learning models},\ }\href
  {https://arxiv.org/abs/2406.07072} {\bibfield  {journal} {\bibinfo  {journal}
  {arXiv:2406.07072}\ } (\bibinfo {year} {2024})}\BibitemShut {NoStop}%
\bibitem [{\citenamefont {Ribeiro}\ \emph {et~al.}(2016)\citenamefont
  {Ribeiro}, \citenamefont {Singh},\ and\ \citenamefont
  {Guestrin}}]{ribeiro2016why}%
  \BibitemOpen
  \bibfield  {author} {\bibinfo {author} {\bibfnamefont {M.~T.}\ \bibnamefont
  {Ribeiro}}, \bibinfo {author} {\bibfnamefont {S.}~\bibnamefont {Singh}},\
  and\ \bibinfo {author} {\bibfnamefont {C.}~\bibnamefont {Guestrin}},\
  }\bibfield  {title} {\bibinfo {title} {{"Why should I trust you?": Explaining
  the predictions of any classifier}},\ }in\ \href
  {https://doi.org/10.1145/2939672.2939778} {\emph {\bibinfo {booktitle}
  {Proceedings of the 22nd ACM SIGKDD International Conference on Knowledge
  Discovery and Data Mining}}},\ \bibinfo {series and number} {KDD '16}\
  (\bibinfo  {publisher} {Association for Computing Machinery},\ \bibinfo
  {address} {New York, NY, USA},\ \bibinfo {year} {2016})\ p.\ \bibinfo {pages}
  {1135–1144}\BibitemShut {NoStop}%
\bibitem [{\citenamefont {Montavon}\ \emph {et~al.}(2018)\citenamefont
  {Montavon}, \citenamefont {Samek},\ and\ \citenamefont
  {Müller}}]{montavon2018methods}%
  \BibitemOpen
  \bibfield  {author} {\bibinfo {author} {\bibfnamefont {G.}~\bibnamefont
  {Montavon}}, \bibinfo {author} {\bibfnamefont {W.}~\bibnamefont {Samek}},\
  and\ \bibinfo {author} {\bibfnamefont {K.-R.}\ \bibnamefont {Müller}},\
  }\bibfield  {title} {\bibinfo {title} {Methods for interpreting and
  understanding deep neural networks},\ }\href
  {https://doi.org/https://doi.org/10.1016/j.dsp.2017.10.011} {\bibfield
  {journal} {\bibinfo  {journal} {Dig. Sign. Proc.}\ }\textbf {\bibinfo
  {volume} {73}},\ \bibinfo {pages} {1} (\bibinfo {year} {2018})}\BibitemShut
  {NoStop}%
\bibitem [{\citenamefont {Achtibat}\ \emph {et~al.}(2024)\citenamefont
  {Achtibat}, \citenamefont {Hatefi}, \citenamefont {Dreyer}, \citenamefont
  {Jain}, \citenamefont {Wiegand}, \citenamefont {Lapuschkin},\ and\
  \citenamefont {Samek}}]{AchICML24}%
  \BibitemOpen
  \bibfield  {author} {\bibinfo {author} {\bibfnamefont {R.}~\bibnamefont
  {Achtibat}}, \bibinfo {author} {\bibfnamefont {S.~M.~V.}\ \bibnamefont
  {Hatefi}}, \bibinfo {author} {\bibfnamefont {M.}~\bibnamefont {Dreyer}},
  \bibinfo {author} {\bibfnamefont {A.}~\bibnamefont {Jain}}, \bibinfo {author}
  {\bibfnamefont {T.}~\bibnamefont {Wiegand}}, \bibinfo {author} {\bibfnamefont
  {S.}~\bibnamefont {Lapuschkin}},\ and\ \bibinfo {author} {\bibfnamefont
  {W.}~\bibnamefont {Samek}},\ }\bibfield  {title} {\bibinfo {title}
  {{A}ttn{LRP}: Attention-aware layer-wise relevance propagation for
  transformers},\ }in\ \href
  {https://proceedings.mlr.press/v235/achtibat24a.html} {\emph {\bibinfo
  {booktitle} {Proceedings of the 41st International Conference on Machine
  Learning}}},\ \bibinfo {series} {Proceedings of Machine Learning Research},
  Vol.\ \bibinfo {volume} {235}\ (\bibinfo {year} {2024})\ pp.\ \bibinfo
  {pages} {135--168}\BibitemShut {NoStop}%
\bibitem [{\citenamefont {Bowles}\ \emph {et~al.}(2024)\citenamefont {Bowles},
  \citenamefont {Ahmed},\ and\ \citenamefont
  {Schuld}}]{bowles2024betterclassicalsubtleart}%
  \BibitemOpen
  \bibfield  {author} {\bibinfo {author} {\bibfnamefont {J.}~\bibnamefont
  {Bowles}}, \bibinfo {author} {\bibfnamefont {S.}~\bibnamefont {Ahmed}},\ and\
  \bibinfo {author} {\bibfnamefont {M.}~\bibnamefont {Schuld}},\ }\bibfield
  {title} {\bibinfo {title} {Better than classical? the subtle art of
  benchmarking quantum machine learning models},\ }\href
  {https://arxiv.org/abs/2403.07059} {\bibfield  {journal} {\bibinfo  {journal}
  {arXiv:2403.07059}\ } (\bibinfo {year} {2024})}\BibitemShut {NoStop}%
\bibitem [{\citenamefont {Schuld}(2021)}]{schuld2021supervised}%
  \BibitemOpen
  \bibfield  {author} {\bibinfo {author} {\bibfnamefont {M.}~\bibnamefont
  {Schuld}},\ }\bibfield  {title} {\bibinfo {title} {Supervised quantum machine
  learning models are kernel methods},\ }\href
  {https://arxiv.org/abs/2101.11020} {\bibfield  {journal} {\bibinfo  {journal}
  {arXiv:2101.11020}\ } (\bibinfo {year} {2021})}\BibitemShut {NoStop}%
\bibitem [{\citenamefont {Cybiński}\ \emph {et~al.}(2024)\citenamefont
  {Cybiński}, \citenamefont {Enouen}, \citenamefont {Georges},\ and\
  \citenamefont {Dawid}}]{cybinski2024speakphysicistunderstandyou}%
  \BibitemOpen
  \bibfield  {author} {\bibinfo {author} {\bibfnamefont {K.}~\bibnamefont
  {Cybiński}}, \bibinfo {author} {\bibfnamefont {J.}~\bibnamefont {Enouen}},
  \bibinfo {author} {\bibfnamefont {A.}~\bibnamefont {Georges}},\ and\ \bibinfo
  {author} {\bibfnamefont {A.}~\bibnamefont {Dawid}},\ }\bibfield  {title}
  {\bibinfo {title} {Speak so a physicist can understand you! tetriscnn for
  detecting phase transitions and order parameters},\ }\href
  {https://arxiv.org/abs/2411.02237} {\bibfield  {journal} {\bibinfo  {journal}
  {arXiv: 2411.02237}\ } (\bibinfo {year} {2024})}\BibitemShut {NoStop}%
\bibitem [{\citenamefont {Montavon}\ \emph {et~al.}(2017)\citenamefont
  {Montavon}, \citenamefont {Lapuschkin}, \citenamefont {Binder}, \citenamefont
  {Samek},\ and\ \citenamefont {Müller}}]{montavon2017explaining}%
  \BibitemOpen
  \bibfield  {author} {\bibinfo {author} {\bibfnamefont {G.}~\bibnamefont
  {Montavon}}, \bibinfo {author} {\bibfnamefont {S.}~\bibnamefont
  {Lapuschkin}}, \bibinfo {author} {\bibfnamefont {A.}~\bibnamefont {Binder}},
  \bibinfo {author} {\bibfnamefont {W.}~\bibnamefont {Samek}},\ and\ \bibinfo
  {author} {\bibfnamefont {K.-R.}\ \bibnamefont {Müller}},\ }\bibfield
  {title} {\bibinfo {title} {Explaining non-linear classification decisions
  with deep taylor decomposition},\ }\href
  {https://doi.org/https://doi.org/10.1016/j.patcog.2016.11.008} {\bibfield
  {journal} {\bibinfo  {journal} {Pattern Recognition}\ }\textbf {\bibinfo
  {volume} {65}},\ \bibinfo {pages} {211} (\bibinfo {year} {2017})}\BibitemShut
  {NoStop}%
\bibitem [{\citenamefont {Schuld}\ \emph {et~al.}(2020)\citenamefont {Schuld},
  \citenamefont {Bocharov}, \citenamefont {Svore},\ and\ \citenamefont
  {Wiebe}}]{Schuld2020circuit}%
  \BibitemOpen
  \bibfield  {author} {\bibinfo {author} {\bibfnamefont {M.}~\bibnamefont
  {Schuld}}, \bibinfo {author} {\bibfnamefont {A.}~\bibnamefont {Bocharov}},
  \bibinfo {author} {\bibfnamefont {K.~M.}\ \bibnamefont {Svore}},\ and\
  \bibinfo {author} {\bibfnamefont {N.}~\bibnamefont {Wiebe}},\ }\bibfield
  {title} {\bibinfo {title} {Circuit-centric quantum classifiers},\ }\href
  {https://doi.org/10.1103/physreva.101.032308} {\bibfield  {journal} {\bibinfo
   {journal} {Phys. Rev. A}\ }\textbf {\bibinfo {volume} {101}},\ \bibinfo
  {pages} {032308} (\bibinfo {year} {2020})}\BibitemShut {NoStop}%
\bibitem [{\citenamefont {Bergholm}\ \emph {et~al.}(2022)\citenamefont
  {Bergholm}, \citenamefont {Izaac}, \citenamefont {Schuld}, \citenamefont
  {Gogolin}, \citenamefont {Ahmed}, \citenamefont {Ajith}, \citenamefont
  {Alam}, \citenamefont {Alonso-Linaje}, \citenamefont {AkashNarayanan},
  \citenamefont {Asadi}, \citenamefont {Arrazola}, \citenamefont {Azad},
  \citenamefont {Banning}, \citenamefont {Blank}, \citenamefont {Bromley},
  \citenamefont {Cordier}, \citenamefont {Ceroni}, \citenamefont {Delgado},
  \citenamefont {Matteo}, \citenamefont {Dusko}, \citenamefont {Garg},
  \citenamefont {Guala}, \citenamefont {Hayes}, \citenamefont {Hill},
  \citenamefont {Ijaz}, \citenamefont {Isacsson}, \citenamefont {Ittah},
  \citenamefont {Jahangiri}, \citenamefont {Jain}, \citenamefont {Jiang},
  \citenamefont {Khandelwal}, \citenamefont {Kottmann}, \citenamefont {Lang},
  \citenamefont {Lee}, \citenamefont {Loke}, \citenamefont {Lowe},
  \citenamefont {McKiernan}, \citenamefont {Meyer}, \citenamefont
  {Montañez-Barrera}, \citenamefont {Moyard}, \citenamefont {Niu},
  \citenamefont {O'Riordan}, \citenamefont {Oud}, \citenamefont {Panigrahi},
  \citenamefont {Park}, \citenamefont {Polatajko}, \citenamefont {Quesada},
  \citenamefont {Roberts}, \citenamefont {Sá}, \citenamefont {Schoch},
  \citenamefont {Shi}, \citenamefont {Shu}, \citenamefont {Sim}, \citenamefont
  {Singh}, \citenamefont {Strandberg}, \citenamefont {Soni}, \citenamefont
  {Száva}, \citenamefont {Thabet}, \citenamefont {Vargas-Hernández},
  \citenamefont {Vincent}, \citenamefont {Vitucci}, \citenamefont {Weber},
  \citenamefont {Wierichs}, \citenamefont {Wiersema}, \citenamefont {Willmann},
  \citenamefont {Wong}, \citenamefont {Zhang},\ and\ \citenamefont
  {Killoran}}]{bergholm2022pennylane}%
  \BibitemOpen
  \bibfield  {author} {\bibinfo {author} {\bibfnamefont {V.}~\bibnamefont
  {Bergholm}}, \bibinfo {author} {\bibfnamefont {J.}~\bibnamefont {Izaac}},
  \bibinfo {author} {\bibfnamefont {M.}~\bibnamefont {Schuld}}, \bibinfo
  {author} {\bibfnamefont {C.}~\bibnamefont {Gogolin}}, \bibinfo {author}
  {\bibfnamefont {S.}~\bibnamefont {Ahmed}}, \bibinfo {author} {\bibfnamefont
  {V.}~\bibnamefont {Ajith}}, \bibinfo {author} {\bibfnamefont {M.~S.}\
  \bibnamefont {Alam}}, \bibinfo {author} {\bibfnamefont {G.}~\bibnamefont
  {Alonso-Linaje}}, \bibinfo {author} {\bibfnamefont {B.}~\bibnamefont
  {AkashNarayanan}}, \bibinfo {author} {\bibfnamefont {A.}~\bibnamefont
  {Asadi}}, \bibinfo {author} {\bibfnamefont {J.~M.}\ \bibnamefont {Arrazola}},
  \bibinfo {author} {\bibfnamefont {U.}~\bibnamefont {Azad}}, \bibinfo {author}
  {\bibfnamefont {S.}~\bibnamefont {Banning}}, \bibinfo {author} {\bibfnamefont
  {C.}~\bibnamefont {Blank}}, \bibinfo {author} {\bibfnamefont {T.~R.}\
  \bibnamefont {Bromley}}, \bibinfo {author} {\bibfnamefont {B.~A.}\
  \bibnamefont {Cordier}}, \bibinfo {author} {\bibfnamefont {J.}~\bibnamefont
  {Ceroni}}, \bibinfo {author} {\bibfnamefont {A.}~\bibnamefont {Delgado}},
  \bibinfo {author} {\bibfnamefont {O.~D.}\ \bibnamefont {Matteo}}, \bibinfo
  {author} {\bibfnamefont {A.}~\bibnamefont {Dusko}}, \bibinfo {author}
  {\bibfnamefont {T.}~\bibnamefont {Garg}}, \bibinfo {author} {\bibfnamefont
  {D.}~\bibnamefont {Guala}}, \bibinfo {author} {\bibfnamefont
  {A.}~\bibnamefont {Hayes}}, \bibinfo {author} {\bibfnamefont
  {R.}~\bibnamefont {Hill}}, \bibinfo {author} {\bibfnamefont {A.}~\bibnamefont
  {Ijaz}}, \bibinfo {author} {\bibfnamefont {T.}~\bibnamefont {Isacsson}},
  \bibinfo {author} {\bibfnamefont {D.}~\bibnamefont {Ittah}}, \bibinfo
  {author} {\bibfnamefont {S.}~\bibnamefont {Jahangiri}}, \bibinfo {author}
  {\bibfnamefont {P.}~\bibnamefont {Jain}}, \bibinfo {author} {\bibfnamefont
  {E.}~\bibnamefont {Jiang}}, \bibinfo {author} {\bibfnamefont
  {A.}~\bibnamefont {Khandelwal}}, \bibinfo {author} {\bibfnamefont
  {K.}~\bibnamefont {Kottmann}}, \bibinfo {author} {\bibfnamefont {R.~A.}\
  \bibnamefont {Lang}}, \bibinfo {author} {\bibfnamefont {C.}~\bibnamefont
  {Lee}}, \bibinfo {author} {\bibfnamefont {T.}~\bibnamefont {Loke}}, \bibinfo
  {author} {\bibfnamefont {A.}~\bibnamefont {Lowe}}, \bibinfo {author}
  {\bibfnamefont {K.}~\bibnamefont {McKiernan}}, \bibinfo {author}
  {\bibfnamefont {J.~J.}\ \bibnamefont {Meyer}}, \bibinfo {author}
  {\bibfnamefont {J.~A.}\ \bibnamefont {Montañez-Barrera}}, \bibinfo {author}
  {\bibfnamefont {R.}~\bibnamefont {Moyard}}, \bibinfo {author} {\bibfnamefont
  {Z.}~\bibnamefont {Niu}}, \bibinfo {author} {\bibfnamefont {L.~J.}\
  \bibnamefont {O'Riordan}}, \bibinfo {author} {\bibfnamefont {S.}~\bibnamefont
  {Oud}}, \bibinfo {author} {\bibfnamefont {A.}~\bibnamefont {Panigrahi}},
  \bibinfo {author} {\bibfnamefont {C.-Y.}\ \bibnamefont {Park}}, \bibinfo
  {author} {\bibfnamefont {D.}~\bibnamefont {Polatajko}}, \bibinfo {author}
  {\bibfnamefont {N.}~\bibnamefont {Quesada}}, \bibinfo {author} {\bibfnamefont
  {C.}~\bibnamefont {Roberts}}, \bibinfo {author} {\bibfnamefont
  {N.}~\bibnamefont {Sá}}, \bibinfo {author} {\bibfnamefont {I.}~\bibnamefont
  {Schoch}}, \bibinfo {author} {\bibfnamefont {B.}~\bibnamefont {Shi}},
  \bibinfo {author} {\bibfnamefont {S.}~\bibnamefont {Shu}}, \bibinfo {author}
  {\bibfnamefont {S.}~\bibnamefont {Sim}}, \bibinfo {author} {\bibfnamefont
  {A.}~\bibnamefont {Singh}}, \bibinfo {author} {\bibfnamefont
  {I.}~\bibnamefont {Strandberg}}, \bibinfo {author} {\bibfnamefont
  {J.}~\bibnamefont {Soni}}, \bibinfo {author} {\bibfnamefont {A.}~\bibnamefont
  {Száva}}, \bibinfo {author} {\bibfnamefont {S.}~\bibnamefont {Thabet}},
  \bibinfo {author} {\bibfnamefont {R.~A.}\ \bibnamefont {Vargas-Hernández}},
  \bibinfo {author} {\bibfnamefont {T.}~\bibnamefont {Vincent}}, \bibinfo
  {author} {\bibfnamefont {N.}~\bibnamefont {Vitucci}}, \bibinfo {author}
  {\bibfnamefont {M.}~\bibnamefont {Weber}}, \bibinfo {author} {\bibfnamefont
  {D.}~\bibnamefont {Wierichs}}, \bibinfo {author} {\bibfnamefont
  {R.}~\bibnamefont {Wiersema}}, \bibinfo {author} {\bibfnamefont
  {M.}~\bibnamefont {Willmann}}, \bibinfo {author} {\bibfnamefont
  {V.}~\bibnamefont {Wong}}, \bibinfo {author} {\bibfnamefont {S.}~\bibnamefont
  {Zhang}},\ and\ \bibinfo {author} {\bibfnamefont {N.}~\bibnamefont
  {Killoran}},\ }\bibfield  {title} {\bibinfo {title} {Pennylane: Automatic
  differentiation of hybrid quantum-classical computations},\ }\href
  {https://arxiv.org/abs/1811.04968} {\bibfield  {journal} {\bibinfo  {journal}
  {arXiv:1811.04968}\ } (\bibinfo {year} {2022})}\BibitemShut {NoStop}%
\bibitem [{\citenamefont {Frostig}\ \emph {et~al.}(2018)\citenamefont
  {Frostig}, \citenamefont {Johnson},\ and\ \citenamefont
  {Leary}}]{frostig2018Jax}%
  \BibitemOpen
  \bibfield  {author} {\bibinfo {author} {\bibfnamefont {R.}~\bibnamefont
  {Frostig}}, \bibinfo {author} {\bibfnamefont {M.}~\bibnamefont {Johnson}},\
  and\ \bibinfo {author} {\bibfnamefont {C.}~\bibnamefont {Leary}},\ }\href
  {https://mlsys.org/Conferences/doc/2018/146.pdf} {\bibinfo {title} {Compiling
  machine learning programs via high-level tracing}} (\bibinfo {year}
  {2018})\BibitemShut {NoStop}%
\bibitem [{\citenamefont {Smilkov}\ \emph {et~al.}(2017)\citenamefont
  {Smilkov}, \citenamefont {Thorat}, \citenamefont {Kim}, \citenamefont
  {Viégas},\ and\ \citenamefont {Wattenberg}}]{smilkov2017smoothgrad}%
  \BibitemOpen
  \bibfield  {author} {\bibinfo {author} {\bibfnamefont {D.}~\bibnamefont
  {Smilkov}}, \bibinfo {author} {\bibfnamefont {N.}~\bibnamefont {Thorat}},
  \bibinfo {author} {\bibfnamefont {B.}~\bibnamefont {Kim}}, \bibinfo {author}
  {\bibfnamefont {F.}~\bibnamefont {Viégas}},\ and\ \bibinfo {author}
  {\bibfnamefont {M.}~\bibnamefont {Wattenberg}},\ }\bibfield  {title}
  {\bibinfo {title} {Smoothgrad: removing noise by adding noise},\ }\href
  {https://arxiv.org/abs/1706.03825} {\bibfield  {journal} {\bibinfo  {journal}
  {arXiv:1706.03825}\ } (\bibinfo {year} {2017})}\BibitemShut {NoStop}%
\end{thebibliography}%

\onecolumngrid
\clearpage
\appendix

\setcounter{figure}{0}
\setcounter{theorem}{0}
\setcounter{footnote}{0}

\begin{center}
\large{Supplementary Material for \\ ``Opportunities and limitations of explaining quantum machine learning''
}
\end{center}

\section{Probabilistic functions for machine learning}\label{a:QMLvsPML}

    For didactic purposes, one could consider the closest fully-classical analogous to the hypothesis families of quantum functions introduced in Section~\ref{ss:QML}, as shown in Fig.~\ref{fig:QMLvsPML}.
    Such an exercise would involve, e.g., 
    specifying a parametrization of the set of doubly-stochastic matrices.
    A real-valued function could then arise from a probabilistic classical circuit as a three-step process: (1) start from an easy distribution (like the Kronecker delta on the all-$0$ string), (2) 
    apply a sequence of parametrized doubly-stochastic 
    gates (some of which would be data dependent, and some not), and 
    (3) estimate the expectation value of a fixed function on bits $b$.
    Doubly-stochastic transformations play the role of unitary matrices in this hypothetical scenario.
    Unitary matrices can safely be thought of as the generalization of rotations to complex-valued fields.
    This way, the action of a unitary on the matrix $D_p$ is akin to rotating the vector $p$ in a way that does not incur issues with the signs of the entries or the normalization.
    The main difference between doubly stochastic transformations and unitary transformations is that the former can only increase the Shannon entropy of the distribution (the inverse of a doubly stochastic matrix need not be stochastic), whereas the latter also allow for decreasing Shannon entropy, since the inverse of a unitary 
    matrix is also unitary.
    
    One immediately notices that even though one could define classical ML models based on these probabilistic circuits, these are really far from what practitioners use.
    With this presentation we hope to illustrate the conceptual simplicity of PQC-based QML models, yet in practice the function families PQC give rise to are rich and intricate.

    \begin{figure*}[h]
        \centering
        \includegraphics{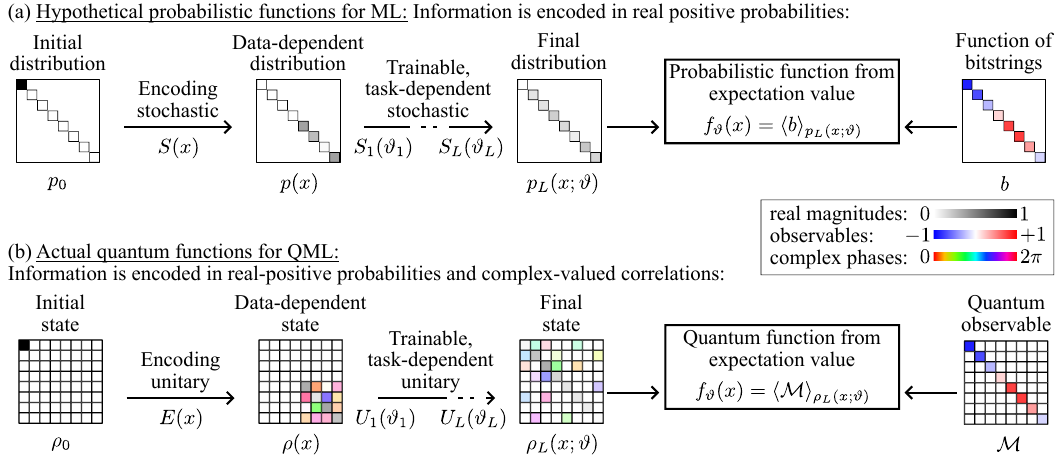}
        \caption{
            Sketch illustrating the relation between the quantum functions used in QML and the conceptually closest fully-classical probabilistic functions one could use for ML.
            (a) Classical probability distributions can be represented as diagonal matrices.
            Then, information on the data or trainable parameters can be encoded as a bi-stochastic transformation, resulting in a parametrized final distribution.
            A real-valued function can be recovered from evaluating the expectation value of any function of bitstrings with respect to the classical final distribution.
            (b) The same idea applies, just that now the classical information is stored in quantum states, represented as positive semi-definite, unit trace, Hermitian matrices.
            The core mechanisms are left unchanged, just now we have access to complex-valued entries outside the main diagonal.
            Again, a real-valued function can be recovered as the expectation value of a fixed quantum observable with respect to the final parametrized quantum state.
        }
        \label{fig:QMLvsPML}
    \end{figure*}

\section{Illustrative presentation of \LRP{}}\label{a:LRP}

    As introduced in Section~\ref{sss:divide}, we consider functions following a simple computational graph
    \begin{center}
        \begin{tikzcd}[column sep=0.5cm, row sep = small]
           x=z_0 \arrow[r, mapsto, "\phi_1"] & z_1 \arrow[r, mapsto, "\phi_2"] & \cdots \arrow[r, mapsto, "\phi_l"] & z_l \arrow[r, mapsto, "\phi_{l+1}"] & \cdots \arrow[r, mapsto, "\phi_L"] & z_L = f(x).
        \end{tikzcd}
    \end{center}
    Then, the first step in \LRP{} is to produce a relevance heat-map for the penultimate intermediate state, which we denote $R^{(L-1)}(x)$.
    This intermediate explanation has the same size and shape as $z_{L-1}$, so it ultimately depends on the structure of the model.
    In this line, we would say $R^{(L)}(x) = f(x)$, which for simplicity we take to be a single real number.
    We would write the explanation $R^{(L-1)}(x)$ as a function of $f(x)=z_L$ and $z_{L-1}$:
    \begin{center}
        \begin{tikzcd}
            f(x) = z_L = R^{(L)}(x) \arrow[r,mapsto] & R^{(L-1)}(x). \\
            z_{L-1} \arrow[u,mapsto,"\phi_{L}"] \arrow[ur, mapsto]
        \end{tikzcd}
    \end{center}
    For ease of notation, from now on we drop the $x$ dependence of the intermediate explanations, but in what follows we talk about explaining the model prediction for a single input (we are studying local explanations).

    Now we have an explanation for the penultimate hidden layer, and our goal is to, from here, produce an explanation for the second-to-last one.
    For the penultimate layer, though, in general we have $z_{L-1}\neq R^{(L-1)}$, which was the case for the last layer.
    This means that, in order to have an intermediate explanation for the next layer $R^{(L-2)}$, we must also take $R^{(L-1)}$ into account:
    \begin{center}
        \begin{tikzcd}[row sep = small]
              & R^{(L-1)} \arrow[r,mapsto] & R^{(L-2)}. \\
            R^{(L)} \arrow[ur,mapsto] \\
            & z_{L-1} \arrow[ul,mapsto,"\phi_{L}"'] \arrow[uu, mapsto] \arrow[uur, mapsto] & z_{L-1} \arrow[l, mapsto, "\phi_{L-1}"'] \arrow[uu, mapsto]
        \end{tikzcd}
    \end{center}
    This idea gives us a blueprint that can be followed throughout: the $l^\text{th}$ intermediate explanation is a simple function of $z_l, z_{l+1}$ and the $(l+1)^\text{th}$ intermediate explanation~\cite{montavon2017explaining}:
    \begin{center}
        \begin{tikzcd}[column sep = 0.5 cm, row sep = small]
              & R^{(L-1)} \arrow[r,mapsto] &\cdots \arrow[r,mapsto] & R^{(1)} \arrow[r, mapsto] & R^{(0)} \\
            z_L=R^{(L)} \arrow[ur,mapsto] \\
            & z_{L-1} \arrow[ul,mapsto,"\phi_{L}"'] \arrow[uu, mapsto] \arrow[uur, mapsto] & \cdots \arrow[l,mapsto, "\phi_{L-1}"'] \arrow[uur, mapsto] & z_1 \arrow[l, mapsto, "\phi_2"'] \arrow[uu, mapsto] \arrow[uur, mapsto] & z_0 \arrow[l, mapsto, "\phi_1"'] \arrow[uu, mapsto].
        \end{tikzcd}
    \end{center}
    At this point we recall that $z_0=x$ and it follows that $R^{(0)}$ is the last explanation, thus the one 
    we produce for this model $E{}(x)=R^{(0)}$.

\section{Derivation of the \taylorinf{} explanation}\label{a:taylor}

    \begin{proposition}[Taylor explanation]\label{prop:taylorrel}
        Let $f\colon\bbR^d\to\bbR$ be a trigonometric polynomial 
        \begin{align}
            f(x) &= \sum_{\omega\in\{0,\pm1\}^d}
        \left(    
a_\omega\cos\langle\omega,x\rangle+b_\omega\sin\langle\omega,x\rangle
\right)
        \end{align}
        of degree at most $1$ over each one of $d$ variables\footnote{
            Here we take multivariate polynomials to also allow for negative frequencies up to minus the degree.
        }.
        Then for any point $\tx\in\bbR^d$, $f$ can be written as
        \begin{align}
            f(x) &= f(\tx) + \sum_{i=1}^d T_i(x,\tx) + \varepsilon,
        \end{align}
        where $T_i$ is defined as 
        \begin{align}
            T_i(x,y) &:= \sin(x_i-y_i)\left.\pder[f(x)]{x_i}\right|_{x=y} - (1-\cos(x_i-y_i))\left.\pder[^2f(x)]{x_i^2}\right|_{x=y},
        \end{align}
        and $\varepsilon$ contains all the cross-derivative contributions, that is, terms of the form \begin{align}
        \left.\left(\prod_{i=1}^d\pder[^{n_i}]{x_i^{n_i}}\right)f(x)\right|_{x=y}
        \end{align}
        with $\lvert n\rvert_0\geq2$.
        Here $n=(n_i)_{i=1}^d$ is the vector of derivative orders, and $\lvert n\rvert_0\geq2$ means that it has at least two different non-$0$ entries.
    \end{proposition}
    \begin{proof}
        We show the identity directly by using the general formula for multivariate Taylor expansion 
        \begin{align}
            f(x) &= \sum_{\lvert n\rvert_1=0}^\infty \left.\pder[^{\vert n\rvert_1}f(x)]{x_1^{n_1}\ldots\partial x_d^{n_d}}\right|_{x=\tx} \prod_{i=1}^d\frac{(x_i-\tx_i)^{n_i}}{n_i!}
        \end{align}
        around an arbitrary point $\tx\in\bbR^d$.
        Again, we call $n=(n_i)_i\in\bbN^d$ the vector of derivative orders, and $\lvert n\rvert_1=\sum_{i=1}^d n_i$ the total derivative order across all variables.
        We use the known identities for even and odd order derivatives of trigonometric polynomials, to get
        \begin{align}
            \pder[^{2k}f]{x_i^{2k}} &= (-1)^{k+1} \pder[^2 f]{x_i^2}, \text{ for all } k\in\bbN\setminus\{0\}, \\
            \pder[^{2k+1}f]{x_i^{2k+1}} &= (-1)^k\pder[f]{x_i}, \text{ for all } k\in\bbN.
        \end{align}
        For these identities, we need only use the fact that $\omega_i^3=\omega_i$ for each component $\omega_i$ of each frequency vector $\omega\in\{0,\pm1\}^d$.
        Now we re-organize the terms of the Taylor formula as
        \begin{align}
            f(x) &= \left.\left[\prod_{i=1}^d\left(\sum_{n_i=0}^\infty\frac{(x_i-\tx_i)^{n_i}}{n_i!}\pder[^{n_i}]{x_i^{n_i}}\right)\right]f(x)\right|_{x=\tx}.
        \end{align}
        We split each of the sums of operators into 
        three separate terms
        \begin{align}
            \sum_{n_i=0}^\infty\frac{(x_i-\tx_i)^{n_i}}{n_i!}\pder[^{n_i}]{x_i^{n_i}} &= \bbI + \sum_{n_i=1}^\infty \frac{(x_i-\tx_i)^{2n_i}}{(2n_i)!}\pder[^{2n_i}]{x_i^{2n_i}} + \sum_{n_i=0}^\infty \frac{(x_i-\tx_i)^{2n_i+1}}{(2n_i+1)!}\pder[^{2n_i+1}]{x_i^{2n_i+1}}.
        \end{align}
        With this grouping of terms, we focus on the terms which overall contain derivatives (of any order) with respect to a single variable, and call all other terms $E$, to obtain
        \begin{align}
            f(x) &= \left.\left[\prod_{i=1}^d\left(\bbI + \sum_{n_i=1}^\infty \frac{(x_i-\tx_i)^{2n_i}}{(2n_i)!}\pder[^{2n_i}]{x_i^{2n_i}} + \sum_{n_i=0}^\infty \frac{(x_i-\tx_i)^{2n_i+1}}{(2n_i+1)!}\pder[^{2n_i+1}]{x_i^{2n_i+1}}\right)\right]f(x)\right|_{x=\tx} \\
            \nonumber
            &= \left.\left[\bbI+\sum_{i=1}^d\left(\sum_{n_i=1}^\infty \frac{(x_i-\tx_i)^{2n_i}}{(2n_i)!}\pder[^{2n_i}]{x_i^{2n_i}} + \sum_{n_i=0}^\infty \frac{(x_i-\tx_i)^{2n_i+1}}{(2n_i+1)!}\pder[^{2n_i+1}]{x_i^{2n_i+1}}\right)+E\right]f(x)\right|_{x=\tx} \\
            \nonumber
            &= f(\tx) + \sum_{i=1}^d\left(\sum_{n_i=1}^\infty \frac{(x_i-\tx_i)^{2n_i}}{(2n_i)!}\left.\pder[^{2n_i}f(x)]{x_i^{2n_i}}\right|_{x=\tx} + \sum_{n_i=0}^\infty \frac{(x_i-\tx_i)^{2n_i+1}}{(2n_i+1)!} \left.\pder[^{2n_i+1}f(x)]{x_i^{2n_i+1}}\right|_{x=\tx}\right) + \left.E(f(x))\right|_{x=\tx}.
            \nonumber
        \end{align}
        At this point, we only need to use the identities we introduced above, and 
        we call $\left.E(f(x))\right|_{x=\tx}=\varepsilon$,
        \begin{align}
            f(x) &= f(\tx) + \sum_{i=1}^d\left((1-\cos(x_i-\tx_i)) \left.\pder[^2f(x)]{x_i^2}\right|_{x=\tx} + \sin(x_i-\tx_i) \left.\pder[f(x)]{x_i}\right|_{x=\tx}\right) + \varepsilon \\
            &= f(\tx) + \sum_{i=1}^d T_i(x,\tx) + \varepsilon,
            \nonumber
        \end{align}
        which is what we had to show.
    \end{proof}

\section{Step-by-step derivation of a digital twin neural network}\label{a:twiNN}

    We consider an encoding-first parametrized quantum circuit as a three step process:
    \begin{enumerate}
        \item Create a data-dependent quantum state $x\mapsto\rho(x)$.
        \item Create a task-dependent parametrized observable $\vartheta\mapsto\calM(\vartheta)$.
        \item Take their inner product $f(x, \vartheta) = \Tr\{\rho(x)\calM(\vartheta)\}$.
    \end{enumerate}

    In this section, we derive a neuralization of such a PQC, in what we have dubbed a digital twiNN (twin Neural Network).
    Our approach corresponds to rewriting the PQC as a Neural Network, which involves rethinking the quantum objects $\rho(x)$ and $\calM(\vartheta)$ as real-valued, large matrices.
    Again, we consider a three step process:
    \begin{enumerate}
        \item Create a data-dependent feature matrix $x\mapsto A(x)$.
        \item Create a task-dependent linear layer $\vartheta\mapsto M(\vartheta)$.
        \item Take their inner product $f(x, \vartheta)=\frac{1}{2}\Tr\{ A(x)M(\vartheta)\}$.
    \end{enumerate}

    We first introduce the recipe to construct $A$ and $M$, and then show that the PQC and the twiNN indeed produce the same input-parameters-output relations.
    Keeping implementation in mind, we have different needs that $A$ and $M$ must fulfill.
    For $A$, we would like to have a closed-form expression that depends explicitly on $x$ for each entry of the feature map.
    For $M$, we are satisfied with a formula that depends explicitly on each entry of $\calM(\vartheta)$, thus being satisfied with only implicit dependence on the entries of $\vartheta$.

    We first introduce a natural way of expanding complex-valued matrices into slightly larger real-valued matrices in a way that respects the inner products.
    For an $N$-dimensional matrix $\bbC^{N\times N}$, consider the map
    \begin{align}
        \sfM\colon\bbC^{N\times N} &\to\bbR^{2N\times2N},\\
        U &\mapsto \begin{pmatrix} \Re(U) & -\Im(U) \\ \Im(U) & \Re(U) \end{pmatrix},
    \end{align}
    where $\Re$ and $\Im$ stand for the entrywise real and imaginary parts, respectively.

    \begin{lemma}[The map $\sfM$ is an isomorphism]\label{al:isomorphism}
        For any $U,V\in\bbC^{N\times N}$ complex matrices and $UV$ their product, it follows that $\sfM(U)\sfM(V) = \sfM(UV)$.
    \end{lemma}
    \begin{proof}
        We show the identity directly.
        First we identify the real and imaginary parts of $UV$ based on those of $U$ and $V$, and then we show that the block structure of $\sfM(UV)$ corresponds to the correct parts.

        Consider the expansion of $U$ and $V$ in their real and imaginary parts $U=\Re(U)+i\Im(U)$, $V=\Re(V)+i\Im(V)$.
        Then take their product $UV = \Re(U)\Re(V) - \Im(U)\Im(V) + i(\Re(U)\Im(V) + \Im(U)\Re(V))$, from where it follows that the real and imaginary parts of the product are $\Re(UV) = \Re(U)\Re(V) - \Im(U)\Im(V)$ and $\Im(UV) = \Re(U)\Im(V) + \Im(U)\Re(V)$.
        Next, consider the 
        matrix-matrix multiplication
        \begin{align}
            \sfM(U)\sfM(V) &= \begin{pmatrix} \Re(U) & -\Im(U) \\ \Im(U) & \Re(U) \end{pmatrix} \begin{pmatrix} \Re(V) & -\Im(V) \\ \Im(V) & \Re(V) \end{pmatrix} 
            \\
            \nonumber
            &= \begin{pmatrix} \Re(U)\Re(V) - \Im(U)\Im(V) & -\Re(U)\Im(V) -\Im(U)\Re(V) \\ \Re(U)\Im(V) + \Im(U)\Re(V) & \Im(U)\Im(V) + \Re(U)\Re(V) \end{pmatrix} \\
              \nonumber
            &= \begin{pmatrix} \Re(UV) & -\Im(UV) \\ \Im(UV) & \Re(UV) \end{pmatrix} = \sfM(UV).
              \nonumber
        \end{align}
        Indeed, we have shown that $\sfM(U)\sfM(V)=\sfM(UV)$.
    \end{proof}

    \begin{lemma}[Trace of $\sfM(H)$]\label{al:trace}
        For any $H\in\bbC^{N\times N}$ Hermitian, $H=H\dagg$, it holds that $\Tr\{\sfM(H)\} = 2 \Tr\{H\}$.
    \end{lemma}
    \begin{proof}
        We show this quickly and directly.
        One basic property of Hermitian matrices is that their main diagonal is real valued. Since the trace only involves the entries on the main diagonal, it follows that $\Tr\{H\}=\Tr\{\Re(H)\}$.
        From this, if we look at the block structure of $\sfM(H)$, it follows that the main diagonal of $\sfM(H)$ is nothing but two copies of the main diagonal of $H$, one after the other.
        And, hence, it is clear that $\Tr\{\sfM(H)\} = 2 \Tr\{\Re(H)\} = 2\Tr\{H\}$, thus proving the lemma.
    \end{proof}

    Next, for the feature map $A(x)$.
    Without loss of generality, we assume the PQC encodes $d$-dimensional data $x\in\bbR^d$ in $d$ qubits, each component being introduced once as the parameter of a parametrized Pauli-$X$ rotation $R_X(x_i)$, on the corresponding qubit.

    \begin{lemma}[Entries of $A$]\label{al:A_entries}
        We have efficient formulas for computing any entry of $A(x)$.
    \end{lemma}
    \begin{proof}
        We give the formula in the following.
        The encoding gate $U(x)$ is of the form $U(x) = \bigotimes_{j=1}^d R_X(x_j)$ and, when applied on the $\lvert0\rangle$ state vector, it produces a data-dependent state as 
        \begin{align}
            U(x)\lvert0\rangle\!\langle0\rvert U(x)\dagg &= \bigotimes_{j=1}^d R_X(x_j)\lvert0\rangle\!\langle0\rvert R_X(x_j)\dagg \\
            \nonumber
            &= \bigotimes_{j=1}^d \begin{pmatrix} \cos\frac{x_j}{2} \\ -i\sin\frac{x_j}{2}\end{pmatrix}\begin{pmatrix}\cos\frac{x_j}{2} & i\sin\frac{x_j}{2} \end{pmatrix} \\
             \nonumber
            &= \bigotimes_{j=1}^{d}\begin{pmatrix}\cos^2\frac{x_j}{2} & i\cos\frac{x_j}{2}\sin\frac{x_j}{2} \\ -i\cos\frac{x_j}{2}\sin\frac{x_j}{2} & \sin^2\frac{x_j}{2}\end{pmatrix} \\
             \nonumber
            &= \bigotimes_{j=1}^d\begin{pmatrix}\frac{1+\cos x_j}{2} & \frac{i\sin x_j}{2} \\ -\frac{i\sin x_j}{2} & \frac{1-\cos x_j}{2}\end{pmatrix}.
             \nonumber
        \end{align}
        At this point, we introduce short-hand notation
        \begin{align}
            c^+_j &:= 1+\cos x_j, \\
            c^-_j &:= 1-\cos x_j ,\\
            s_j &:= \sin x_j.
        \end{align}
        We also introduce bit-string indices $k,l\in\{0,1\}^d$, whose entries are indexed by $j$, $k=(k_j)_j$, with $k_j\in\{0,1\}$, and similarly for $l$.
        We make use of the kronecker delta $\delta_{a,b}$, which is the indicator function for the condition $a=b$.
        Introducing these new symbols in the equations we obtain a formula for the entries of $\rho(x)$ which should be particularly friendly to implement on a classical computer,
        \begin{align}
            U(x)\lvert0\rangle\!\langle0\rvert U(x)\dagg &= \ldots = \frac{1}{2^d}\bigotimes_{j=1}^d \begin{pmatrix} c^+_j & is_j \\ -is_j & c^-_j\end{pmatrix} \\
             \nonumber
            &= \frac{1}{2^d}\bigotimes_{j=1}^d\left[c^+_j\lvert0\rangle\!\langle0\rvert + c^-_j\lvert1\rangle\!\langle1\rvert +is_j(\lvert0\rangle\!\langle1\rvert + \lvert1\rangle\!\langle0\rvert)\right] \\
             \nonumber
            &= \frac{1}{2^d}\bigotimes_{j=1}^d \left(\sum_{k_j,l_j=0}^1\left[\delta_{k_j,l_j}(c^+_j\delta_{k_j,0}+c^-_j\delta_{k_j,1}) + (1-\delta_{k_j,l_j})(-1)^{k_j}s_j\right] \lvert k_j\rangle\!\langle l_j\rvert\right) \\
            &= \frac{1}{2^d}\sum_{k,l\in\{0,1\}^d}\left(\prod_{j=1}^d\delta_{k_j,l_j}(c^+_j\delta_{k_j,0}+c^-_j\delta_{k_j,1}) + (1-\delta_{k_j,l_j})(-1)^{k_j}s_j\right)\lvert k\rangle\!\langle l\rvert \\
            &= \rho(x).
             \nonumber
        \end{align}
    
        In order to obtain $A(x)=\sfM(\rho(x))$, we need to separate the real and imaginary parts of $\rho(x)$.
        Luckily, each entry of $\rho(x)$ is either 
        pure real or pure imaginary, with no mixed complex entries.
        Upon visual inspection of the formula for the entries of $\rho(x)$, we observe a pattern that relates each of the trigonometric functions $c^+_j, c^-_j, s_j$ to each of the computational basis elements $\lvert 0/1\rangle\!\langle 0/1\rvert$.
        We want to introduce a map which captures this pattern in the form of a look-up table
        \begin{align}
            \lvert0_j\rangle\!\langle0_j\vert &\mapsto c^+_j ,\\
            \lvert0_j\rangle\!\langle1_j\vert &\mapsto is_j ,\\
            \lvert1_j\rangle\!\langle0_j\vert &\mapsto -is_j ,\\
            \lvert1_j\rangle\!\langle1_j\vert &\mapsto c^-_j.
        \end{align}
        We call it $g_j$ for ease of notation, and its formula reads as 
        \begin{align}
            g_j\colon\{0,1\}^2 &\to\{c^+_j,c^-_j,s_j\} ,\\
            k_j,l_j &\mapsto g_j(k_j,l_j) = i^{3k_j+l_j}\left(\delta_{k_j,l_j}+\cos(x_j-\frac{\pi}{2}(\delta_{k_j,1}+\delta_{l_j,1}))\right).
        \end{align}
        Similarly, we can introduce a global map $g$ which corresponds to $g_j$ on each qubit, so for any $k,l\in\{0,1\}^d$,
        \begin{align}
            g(k,l) &= \prod_{j=1}^d g_j(k_j,l_j) \\
            \nonumber
            &= \prod_{j=1}^d i^{3k_j+l_j}\left(\delta_{k_j,l_j}+\cos(x_j-\frac{\pi}{2}(\delta_{k_j,1}+\delta_{l_j,1}))\right) \\
             \nonumber
            &= i^{\sum_{j=1}^d3k_j+l_j}\prod_{j=1}^d\left(\delta_{k_j,l_j}+\cos(x_j-\frac{\pi}{2}(\delta_{k_j,1}+\delta_{l_j,1}))\right) \\
             \nonumber
            &= i^{3\lvert k\rvert + \lvert l\rvert}\prod_{j=1}^d\left(\delta_{k_j,l_j}+\cos(x_j-\frac{\pi}{2}(\delta_{k_j,1}+\delta_{l_j,1}))\right),
             \nonumber
        \end{align}
        We used the notation $\lvert k\rvert$ for the parity of the bit-string $k$ (the number of $1$s).
        With this we reached yet another implementation-friendly formula for the entries of $\rho(x)$, and in fact one where we can very quickly extract the real and imaginary parts
        \begin{align}
            \rho(x) &= \sum_{k,l\in\{0,1\}^d} \frac{g(k,l)}{2^d} \lvert k\rangle\!\langle l\rvert ,\\
            \Re(\rho(x)) &= \frac{1}{2^d} \sum_{k,l\in\{0,1\}^d} \Re(g(k,l)) \lvert k\rangle\!\langle l\rvert \\
            \nonumber
            &= \frac{1}{2^d} \sum_{k,l\in\{0,1\}^d} \delta_{\lvert k\oplus l\rvert,0} g(k,l) \lvert k\rangle\!\langle l\rvert ,\\
            \Im(\rho(x)) &= \frac{1}{2^d} \sum_{k,l\in\{0,1\}^d} \Im(g(k,l)) \lvert k\rangle\!\langle l\rvert \\
            \nonumber
            &= \frac{1}{2^d} \sum_{k,l\in\{0,1\}^d} \delta_{\lvert k\oplus l\rvert,1} (-i) g(k,l) \lvert k\rangle\!\langle l\rvert.
            \nonumber
        \end{align}
        What these formulas relate is that the \enquote{global even} terms (taking addition modulo $2$ for both bit-strings $k\oplus l$) are pure real, and the \enquote{global odd} terms are pure imaginary.
    
        Bringing everything together, we are left with a formula for computing each entry of the data-dependent matrix $A(x)$, that corresponds to the first layer of our twiNN.
        If we call $G(x)$ the matrix whose entries are $(g(k,l))_{k,l\in\{0,1\}^d}$, we obtain
        \begin{align}
            A(x) &= \frac{1}{4^n}\begin{pmatrix}\Re(G(x)) & -\Im(G(x)) \\ \Im(G(x)) & \Re(G(x)) \end{pmatrix}.
        \end{align}

        For any entry of $A(x)$, we need only compute $g(k,l)$ for the corresponding bit-strings $k,l\in\{0,1\}^d$, whose complexity is linear in $d$, so the formula is efficient for any entry.
    \end{proof}

    Finally, for the task-dependent linear layer, we take $M(\vartheta)=\sfM(\calM(\vartheta))$.
    In order to compute the entries of $\calM(\vartheta)$, we need to have a decomposition of $\calM(\vartheta)$ as 
    a fixed observable and a variational circuit $\calM(\vartheta)=V(\vartheta)\dagg\calM_0V(\vartheta)$.
    This makes sense for example if we assume the fixed Hamiltonian admits an efficient classical description, and if the variational block is composed of several parametrized local unitaries, for which we also have efficient classical descriptions.

    \begin{lemma}[Expectation values]
        With the above definitions, it holds that 
        \begin{align}
            \Tr\{\rho(x)\calM(\vartheta)\} &= \frac{1}{2} \Tr\{A(x) M(\vartheta)\}.
        \end{align}
    \end{lemma}
    \begin{proof}
        We prove this statement directly.
        Since we have $A(x)=\sfM(\rho(x))$ and $M(\vartheta)=\sfM(\calM(\vartheta))$, it follows from Lemma~\ref{al:isomorphism} that $A(x)M(\vartheta)=\sfM(\rho(x)\calM(\vartheta))$.
        Subsequently, Lemma~\ref{al:trace} confirms that $\Tr\{A(x) M(\vartheta)\} =2\Tr\{\rho(x)\calM(\vartheta)\}$.
    \end{proof}

\section{Derivation of quantum layerwise relevance 
propagation for the twiNN, linear rule, and encoding rule}\label{a:QLRP}

    The twiNN produces an output as a 
    two-step process
    \begin{align}
        x \mapsto  A \mapsto f = \Tr\{AM\}.
    \end{align}
    We drop the explicit dependence on $x$ and $\vartheta$ for ease of notation.
    Correspondingly, we are after a relevance attribution method that works in two steps, in reversed order.
    For a given tensor (real number, vector, or matrix) $T$, we use the notation $R(T)$ for the relevance attribution of $T$.
    The relevance tensor $R(T)$ has the same shape and size as $T$, and we denote its components with the same set of indices.
    For example, for a matrix with two indices $T_{ab}$, the relevance of the $(a,b)$ entry of $T$ is denoted as $R_{ab}(T)$.

    With this, we propose a relevance propagation algorithm that, at every step, takes into account the relevance tensor of the previous step, and the tensors of that and the previous step.
    Diagrammatically, this approach looks like follows:
    \begin{center}
        \begin{tikzcd}[column sep=2cm, row sep = small]
            R(f) \arrow[r, mapsto, "\text{linear rule}"] & R(A) \arrow[r, mapsto, "\text{encoding rule}"] & R(x) \\
            f \arrow[u, mapsto]\arrow[r, mapsfrom, "\text{linear step}"'] & A \arrow[u, mapsto]\arrow[r, mapsfrom, "\text{encoding step}"']\arrow[ur,mapsto] & x\arrow[u, mapsto]
        \end{tikzcd}
    \end{center}

    In this section we derive the formula for the linear rule and for the encoding rule.
    As is usual, we start from $R(f) = f$.

\subsection{Linear rule}

    The linear rule is almost immediate.
    We just need to look at the formula for $f$ in terms of $A$, and then visual inspection gives us the rule
    \begin{align}
        f(A) &= \Tr\{AM\} = \sum_{i,j} A_{i,j}M_{i,j}.
    \end{align}
    Since $f$ is a linear function of the entries of $A$, the relevance associated to the entries of $A$ is immediately read out: $R_{i,j}(A) = A_{i,j}M_{i,j}$.
    The linear step ensures exact conservation by construction $R(f) = \sum_{i,j}R_{i,j}(A)$.

\subsection{Encoding rule}

    The encoding step is non-linear, so we have a bit of work to do.
    From the formula for the entries of $A$ in terms of $x$, we realize that each entry of $A$ is almost a trigonometric monomial of degree $1$ on each of the components of $x$.
    Only the $1$ summands contained in the $c^\pm$ terms cause the product to turn into a sum of several monomial terms.
    Yet, for the function basis $t_k\in\{c_j^\pm,s_j\}$, we have that each entry $A_{i,j}$ is either $0$ or of the form
    \begin{align}
        A_{i,j} &= \frac{1}{2^d}\prod_{k=1}^d t_k^{(i,j)}.
    \end{align}
    The crucial fact here is that each $A_{i,j}$ entry as a function is a trigonometric polynomial of degree at most $1$ on each of the $d$ components of $x$.
    Then, following Proposition~\ref{prop:taylorrel}, we have seen they admit the following approximation, for any \emph{root point} $\tx^{(i,j)}$,
    \begin{align}
        A_{i,j}(x) &\approx A_{i,j}(\tx^{(i,j)}) + \sum_{k=1}^d T_k^{(i,j)}(x,\tx^{(i,j)}) + \varepsilon,\\
        &\text{with}\nonumber \\
        T_k^{(i,j)}(x,\tx^{(i,j)}) &= \sin(x_k-\tx_k^{(i,j)})\left.\pder[A_{i,j}(x)]{x_k}\right|_{x=\tx^{(i,j)}} + (1-\cos(x_k-\tx_k^{(i,j)}))\left.\pder[^2A_{i,j}(x)]{x_k^2}\right|_{x=\tx^{(i,j)}},
    \end{align}
    and where $\varepsilon$ contains all cross-order derivative terms.
    With this approximation, and assuming we have suitable root points for each $(i,j)$, for which $A_{i,j}(\tx^{(i,j)})\approx0$, then the encoding rule reads
    \begin{align}
        R_k(x) &= \sum_{i,j}T_k^{(i,j)}(x,\tx^{(i,j)})\frac{R_{i,j}(A)}{A_{i,j}}.
    \end{align}
    This rule is approximately conservative by construction, with the amount of conservation being based on the quality of approximation 
    of the Fourier decomposition
    \begin{align}
        \sum_kR_k(x) &= \sum_k\sum_{i,j}T_k^{(i,j)}(x,\tx^{(i,j)})\frac{R_{i,j}(A)}{A_{i,j}} \\
        \nonumber &= \sum_{i,j}\left(\sum_kT_k^{(i,j)}(x,\tx^{(i,j)})\right)\frac{R_{i,j}(A)}{A_{i,j}} \\
        \nonumber
        &\approx \sum_{i,j}A_{i,j}\frac{R_{i,j}(A)}{A_{i,j}} \\
        \nonumber&= \sum_{i,j}R_{i,j}(A).
        \nonumber
    \end{align}
    In order to implement the encoding rule, one needs to evaluate not only each entry of $A(x)$, but also the first and second partial derivatives of $A(x)$ with respect to each component of $x$.
    Fortunately, these are really straightforward to evaluate, again due to the simple trigonometric structure of $A(x)$ we have been exploiting.
    In particular, we need only use the parameter-shift rule.
    We give it first for an individual qubit, and then see how it generalizes right away to several qubits, provided each qubit uploads a different component
    \begin{align}
        \rho^1(x_1) &= \frac{1}{2}\begin{pmatrix} 1+\cos(x_1) & i\sin(x_1) \\ -i\sin(x_1) & 1-\cos(x_1) \end{pmatrix} ,\\
        \der[\rho^1(x_1)]{x_1} &= \frac{1}{2}\begin{pmatrix} -\sin(x_1) & i\cos(x_1) \\ -i\cos(x_1) & +\sin(x_1) \end{pmatrix} \\
        &= \frac{\rho^1\left(x_1+\frac{\pi}{2}\right)-\rho^1\left(x_1-\frac{\pi}{2}\right)}{2}.
        \nonumber
    \end{align}
    Here $\rho^1(x_1)$ represented a single-qubit encoded state.
    The case is not much harder 
    for the whole $d$-qubit state
    \begin{align}
        \rho(x) &= \rho^1(x_1) \otimes\cdots\otimes\rho^d(x_d) ,\\
        \partial_k\left[\rho(x)\right] &= \partial_k\left[\rho^1(x_1) \otimes\cdots\otimes \rho^k(x_k) \otimes\cdots\otimes\rho^d(x_d)\right] \\
        \nonumber
        &= \rho^1(x_1) \otimes\cdots\otimes \partial_k\rho^k(x_k) \otimes\cdots\otimes\rho^d(x_d)\\
        \nonumber
        &= \rho^1(x_1) \otimes\cdots\otimes \frac{\rho^k\left(x_k+\frac{\pi}{2}\right)-\rho^k\left(x_k-\frac{\pi}{2}\right)}{2} \otimes\cdots\otimes\rho^d(x_d) \\
        \nonumber
        &= \frac{\rho\left(x+\frac{\pi}{2}\he_k\right)-\rho\left(x-\frac{\pi}{2}\he_k\right)}{2},
        \nonumber
    \end{align}
    where $\he_k$ denotes the $k^\text{th}$ basis vector.
    The formula can be used recursively for higher-order derivatives
    \begin{align}
        \partial^2_k\rho(x) &= \partial_k\left[\partial_k\rho(x)\right] \\
        \nonumber
        &= \partial_k\left[\frac{\rho\left(x+\frac{\pi}{2}\he_k\right)-\rho\left(x-\frac{\pi}{2}\he_k\right)}{2}\right] \\
        \nonumber
        &= \partial_k\left[\frac{\rho\left(x+\pi\he_k\right)-\rho(x)-\rho(x)+\rho\left(x-\pi\he_k\right)}{4}\right] \\
        \nonumber
        &= -\frac{\rho(x)-\rho(x+\pi\he_k)}{2},
        \nonumber
    \end{align}
    in the last line we used the $2\pi$-periodicity of $\rho(x)$ to say $\rho(x+\pi\he_k)=\rho(x-\pi\he_k)$.
    The identities we found for $\rho(x)$ apply directly to $A(x)$, since the map from $\rho(x)$ to $A(x)$ is linear.
    That means also we do not need to give explicit closed forms for the derivatives of the $g(k,l)$ functions introduced above, since we have shown in order to compute those we need only evaluate the functions themselves on shifted locations.

    So, there is only one piece missing in this puzzle, and that is how we find the root points $\tx^{(i,j)}$.
    Note that the root points in this case are slightly different in nature as the root points we have talked about in other approaches.
    Here we are not looking for a single point where the overall function $f$ nullifies.
    Rather, we are looking for as many points as components $A$ has (exponential in the input dimension).
    Luckily, since we have closed formulas for the entries of $A$, we can take advantage of the shared structure among all of them.
    
    Each of the entries is a product of trigonometric functions of a single variable 
    \begin{align}
    A_{i,j}(x) &= \prod_k t_k^{(i,j)}(x_k).
    \end{align}
    So, in order to achieve $A_{i,j}(\tx^{(i,j)})=0$, it suffices to find a zero of the trigonometric function $t_k^{(i,j)}$ corresponding to a single component $t_k^{(i,j)}(\tx_k^{(i,j)})=0$.
    
    With this knowledge, we here propose an algorithm to find root points for each component.
    Because the property that must be fulfilled by the root point only affects a single component, it could be that a single point is a good root point for several entries of $A$ simultaneously.
    In the algorithm we propose we concentrate on finding the \emph{closest} root point to a given $x$ for each $A_{i,j}$ in the Euclidean sense.
    We first give the procedure in plain language, and at the end we attach the pseudo-code.
    As we did in the proof of Lemma~\ref{al:A_entries}, we consider again the map
    \begin{align}
        \lvert0_k\rangle\!\langle0_k\vert &\mapsto c^+_k ,\\
        \lvert0_k\rangle\!\langle1_k\vert &\mapsto is_k , \\
        \lvert1_k\rangle\!\langle0_k\vert &\mapsto -is_k ,\\
        \lvert1_k\rangle\!\langle1_k\vert &\mapsto c^-_k.
    \end{align}
    From here, we read the corresponding root points for each trigonometric function
    \begin{align}
        \lvert0_k\rangle\!\langle0_k\vert &\mapsto 1+\cos(\tx_k) = 0 \quad \iff \tx_k = \pm\pi ,\\
        \lvert0_k\rangle\!\langle1_k\vert &\mapsto i\sin(\tx_k) =0 \quad \iff \tx_k \in \{0,\pm\pi\} ,\\
        \lvert1_k\rangle\!\langle0_k\vert &\mapsto -i\sin(\tx_k) =0 \quad \iff \tx_k \in \{0,\pm\pi\}, \\
        \lvert1_k\rangle\!\langle1_k\vert &\mapsto 1-\cos(\tx_k) = 0 \quad \iff \tx_k = 0.
    \end{align}
    Each of the conditions defines a high-dimensional grid (a partiton via periodical hyperplanes in every direction) on input space.
    The question becomes, given $x$ and $(i,j)$, what is the closest root point $\tx^{(i,j)}$?

    From $x$, we start by building three related vectors $x^{(1)}=x$, $x^{(2)}=x+\pi \Vec{1}$, and $x^{(3)}=x-\pi\Vec{1}$.
    Here $\Vec{1}$ represents 
    a 
    vector of ones.
    Now, we consider the matrix $(x^{(n)}_m)_{\substack{m\in\{1,\ldots,d\},\\n=1,2,3}}$ and consider the sequence $(m_l,n_l)_{l\in\{1,\ldots, 3d\}}$, which corresponds to the sorting of the elements of $(x^{(n)}_m)_{\substack{m\in\{1,\ldots,d\},\\n=1,2,3}}$ according to their magnitude.
    For instance, the first few elements are
    \begin{align}
        (m_1,n_1) &= \argmin_{(m,n)\in\{1,\ldots,d\}\times\{1,2,3\}} \{\lvert x^{(n)}_m\rvert\}, \\
        (m_2,n_2) &= \argmin_{(m,n)\in\{1,\ldots,d\}\times\{1,2,3\}\setminus\{(m_1,n_1)\}} \{\lvert x^{(n)}_m\rvert\} ,\\
        &\vdots\nonumber \\
        (m_l,n_l) &= \argmin_{(m,n)\in\{1,\ldots,d\}\times\{1,2,3\}\setminus \cup_{s=1}^{l-1}\{(m_s,n_s)\}} \{\lvert x^{(n)}_m\rvert\} ,\\
        &\vdots\nonumber .
    \end{align}

    Now, take $(m_1,n_1)$, which identifies the hyperplane to which $x$ is closest.
    The first index $m_1$ refers to the component $x_{m_1}$, and the second index $n_1$ identifies whether $x_{m_1}$ is close to $0$, $\pi$, or $-\pi$.
    We identify the specific root point $\tx^{(i,j)}=(\tx_k^{(i,j)})_k$ from the value of $n_1$:
    \begin{align}
        \tx_k &= \begin{cases}
            x_k &\text{for } k\neq m_1 ,\\
            \begin{cases}
                0 &\text{if }n_1=1,\\
                -\pi&\text{if }n_1=2,\\
                \pi&\text{if }n_1=3,
            \end{cases}&\text{for } k=m_1.
        \end{cases}
    \end{align}
    Indeed, $x$ differs from $\tx$ only on a single component, $k=m_1$.
    We have identified the root point, but we have not said which component this is a root point of.
    For that, we go back to the identification above, and again reach an answer based on the value of $n_1$.
    Instead of starting from an $(i,j)$ and finding the matching root point, we go the other way: we start from the root point we just defined, and check for which choices of $(i,j)$ it is a valid root point.

    We consider the bit-string representation of $i$ and $j$, and then focus on the $m_1^\text{th}$ bits $i_{m_1},j_{m_1}$.
    We reverse the correspondence from above, now taking the indexing of $n_1$, for which we have the look-up table
    \begin{align}
        n_1 &\mapsto \lvert i_{m_1}\rangle\!\langle j_{m_1}\rvert ,\\
        1 &\mapsto \{\lvert0\rangle\!\langle1\rvert,\lvert1\rangle\!\langle0\vert,\lvert1\rangle\!\langle1\rvert\} ,\\
        \{2,3\} & \mapsto \lvert0\rangle\!\langle0\vert.
    \end{align}
    The way to read this map is: if $n_1=1$, then $\tx$ as defined above is the root point for all choices of $(i,j)$ where the $m_1^\text{th}$ bit of either $i$ or $j$ is $1$.
    Conversely, if $n_1$ is either $2$ or $3$, then $\tx$ as defined above is the root point for all choices of $(i,j)$ where the $m_1^\text{th}$ bit of both $i$ and $j$ is $0$.

    After this step, we have found a root point that works for either $1/4$ or $3/4$ of all choices of $(i,j)$.
    Notice in the earlier paragraph we talk about \emph{all} choices of $(i,j)$ for which a particular bit is a particular value.

    After this step, we proceed to use the next element in the sequence, $(m_2,n_2)$.
    If it happens that $m_2=m_1$, then there is a chance this is the last step.
    If either $n_1$ or $n_2$ are $1$, then it follows that $\tx$ as defined above is the root point for all choices of $(i,j)$.
    If neither of them is $1$, then the point corresponding to $(m_2,n_2)$ does not provide new root points, and we move further up the sequence.

    If, on the contrary, we have that $m_2\neq m_1$, then we find a new root point.
    In this case, we repeat all the steps replacing $(m_1,n_1)$ by $(m_2,n_2)$ up until the point where we identify for which choices of $(i,j)$ the newly defined $\tx$ is a valid root point.
    As by now we already had a fraction of all possible choices associated with the first root point, now we only look at the remaining ones, those for which we do not have a root point yet.
    At this point, in order to see which choices of $(i,j)$ have the new $\tx$ as a root point, we will check the $m_2^\text{th}$ bit of $i$ and $j$.
    Again, this results on finding a root point for either one or three quarters of the remaining choices of $(i,j)$.
    
    As one can expect, we need to repeat this process at most $\calO(d)$ times before we are guaranteed to have found a root point for each possible choice of $(i,j)$.
    The search finishes only when a new $m_l$ is equal to one of the previous ones, with either $n_l=1$ or the previous one being equal to $1$.

    \begin{algorithm}[H]
        \caption{Root points for the encoding rule}
        \label{alg:rootpoints}
        \begin{algorithmic}[1]
            \Require
            \State $A\in\bbR^{2^d\times 2^d}$ \Comment Real-valued matrix for data-dependent state.
            \State $x\in\bbR^d$ \Comment Input vector.
            \Ensure $\tx\in\bbR^d\times\bbR^{2^d\times2^d}$ vector of root points for each entry of $A$.
            \State
            \For {$i
            \in\{0,1\}^d, j\in\{0,1\}^d$}
                \State $\tx^{(i,j)}_k \gets \texttt{void}$ \Comment Initialize empty root point vector.
            \EndFor
            \State
            \State $x^{(1)}\gets x$
            \State $x^{(2)}\gets x+\pi \vec{1}$
            \State $x^{(3)}\gets x-\pi\vec{1}$
            \State $(M,N)\gets$ sort $(m,n)$ according to ascending $\lvert x^{(n)}_m\rvert$ \Comment Sort indices according to distance to grid.
            \For {$(m_l,n_l)\in(M,N)$}
                \If {$n_l=1$}
                    \For {$(i,j) \in \{(i_{m_l},j_{m_l})\in\{(0,1),(1,0),(1,1)\}\}$} \Comment Check $3/4$ of all remaining bitstrings.
                        \If {$\tx^{(i,j)} = \texttt{void}$} \Comment Assign root point only if it had not already been assigned.
                            \State $\tx^{(i,j)}_k \gets \begin{cases}
                                0 & \text{if } k = m_l \\
                                x_k & \text{else}
                            \end{cases}$ \Comment Assign root point differing from $x$ only on one component.
                        \EndIf
                    \EndFor
                \Else { $n_l\in\{2,3\}$}
                    \For {$(i,j) \in \{(i_{m_l},j_{m_l})=(0,0)\}$} \Comment Check $1/4$ of all remaining bitstrings.
                        \If {$\tx^{(i,j)} = \texttt{void}$} \Comment Assign root point only if it had 
                        not already been assigned.
                            \State $\tx^{(i,j)}_k \gets \begin{cases}
                                \begin{cases}
                                    -\pi & \text{if } n_l=2 \\
                                    \pi & \text{if } n_l=3
                                \end{cases} & \text{if } k = m_l \\
                                x_k & \text{else}
                            \end{cases}$ \Comment Assign root point differing from $x$ only on one component, depending on $n_l$.
                        \EndIf
                    \EndFor
                \EndIf
                \For {$i\in\{0,1\}^d, j\in\{0,1\}^d$} \Comment{Early-stopping criterion.}
                    \If {$\tx^{(i,j)} \neq\texttt{void}$} \Comment Check if all entries of $A$ have a root point assigned to them at the end of each iteration.
                        \State \Return $\tx$ \Comment If all entries have a root point, the algorithm is finished.
                    \EndIf
                \EndFor
            \EndFor
            \State \Return $\tx$
        \end{algorithmic}
    \end{algorithm}
    
    As a potential limitation, this algorithm finds root points which differ from the original point in only one direction.
    That means we are restricting the relevance propagation rule to distribute relevance only onto one component.
    This is a negative property because the component onto which relevance gets distributed is not decided based on task-dependent information, but rather only on the geometry of the encoding functions that arise from the encoding step.

\section{Evaluation metrics} \label{a:evaluation}

    The first evaluation metric measures the alignment between the proposed explanation and the mask, and we call it \emph{explanation alignment} $Q_A$.
    To have a well-behaved score, we compute the fraction of the absolute relevance that is assigned to the correct components:

    \begin{align}
        Q_A(x) &= \frac{\sum_{i=1}^d\left| E_i(x) \right|M_i(x)}{\sum_{i=1}^{d}\left| E_i(x) \right|}.
    \end{align}
    We note that whether $M_i(x)$ is $1$ or $0$ depends on the correct label for $x$.
    Also, the explanation alignment is indirectly related to the conservativity of the explanation $E(x)$, as even if $\sum_i E_i(x)\approx f(x)$, here we sum over the absolute values of the explanation.

    The second quality-of-explanation metric considers a different kind of alignment between the explanation $E$ and the mask $M$, namely their correlation coefficient, to reach the \emph{Pearson correlation} metric $Q_P$.
    This metric deals not in absolute values, but in deviations from the mean, also normalized to produce a well-behaved score:
    \begin{align}
        Q_P(x) &= \operatorname{Corr}_{i\in[d]}(E_i(x),M_i(x)) = \frac{\bbE_i\left[(E_i(x) - \bbE_j\left[E_j(x)\right])(M_i(x) - \bbE_j\left[M_j(x)\right])\right]}{\sqrt{\bbE_i\left[E_i(x) - \bbE_j\left[E_j(x)\right]\right]^2\bbE_i\left[M_i(x) - \bbE_j\left[M_j(x)\right]\right]^2}},
    \end{align}
    here we took $[d]=\{1,\ldots,d\}$, and $\bbE_i$ means the expectation value over the indices $i\in[d]$.

    Finally, we introduce a quality-of-explanation metric using the framework of \emph{receiver operator characteristics} (ROC).
    This metric produces a score given a set of explanations, and not for a single input.
    Introducing a real-valued threshold $\alpha\in[0,1]$, we momentarily define as \enquote{well-explained} all inputs whose explanation alignment is higher than $\alpha$, and as \enquote{wrong-explained} all inputs whose alignment to the unimportant components (for example the opposite of the mas $1-M(x)$) is higher than $\alpha$.
    This way, an input could be well-explained and wrong-explained at the same time, for the purposes of this metric.

    Computing this third-and-last evaluation metric follows a two-step process.
    First, given a set of inputs $S$ we compute the fraction of well-explained $r_+(\alpha)$ and wrong-explained $r_-(\alpha)$ inputs, for different values $\alpha\in[0,1]$:
    \begin{align}
        r_+(\alpha) &= \frac{\lvert\{x\in S\,|\, x\text{ well-explained}\}\rvert}{\lvert S\rvert} \\
        r_-(\alpha) &= \frac{\lvert\{x\in S\,|\, x\text{ wrong-explained}\}\rvert}{\lvert S\rvert},
    \end{align}
    and from here, we consider the parametrized curve $\gamma(\alpha) = (r_-(\alpha),r_+(\alpha))$ for $\alpha\in[0,1]$.
    From this parametrized curve, we numerically extract the \emph{receiver operator characteristic} curve (ROC), where the dependence between $r_+$ and $r_-$ is made explicit, removing the parameter $\alpha$.
    That is, we take the fraction of well-explained inputs as a function of the fraction of wrong-explained inputs $r_+(r_-)\in[0,1]$, for $r_-\in[0,1]$.
    Next, we compute the \emph{area under the curve} (AUC) for the ROC curve, which gives us the metric:
    \begin{align}
        Q_{\text{ROC}}(S) &= \int_0^1 r_+(r_-) \,\mathrm d r_-.
    \end{align}

    Given that both $r_+,r_-\in[0,1]$, it follows that $Q_\text{ROC}\in[0,1]$.
    The ROC curve has an interpretation coming from binary classification tasks, below we show an example of how such curves look like in practice.
    In the main text we keep only the area under the curve, and not the whole curve.
    
    Comparing the ROC curves for different attribution methods thus allows for a qualitative assessment of which methods perform better at accurately attributing relevance to the main dimensions while minimizing irrelevant attributions. This behavior is then succinctly summarized in the  area under the ROC curve (AUC) as a single number.

\section{Experimental setup}\label{a:experiments}
    
    In this part of the appendix, we give an overview of the employed data-set and task, the different XAI and XQML methods we compare, and the evaluation metrics we use.

\subsection{Dataset and learning task} \label{app:data}

    \begin{figure*}[b]
    \centering
    \includegraphics[width = .9\linewidth]{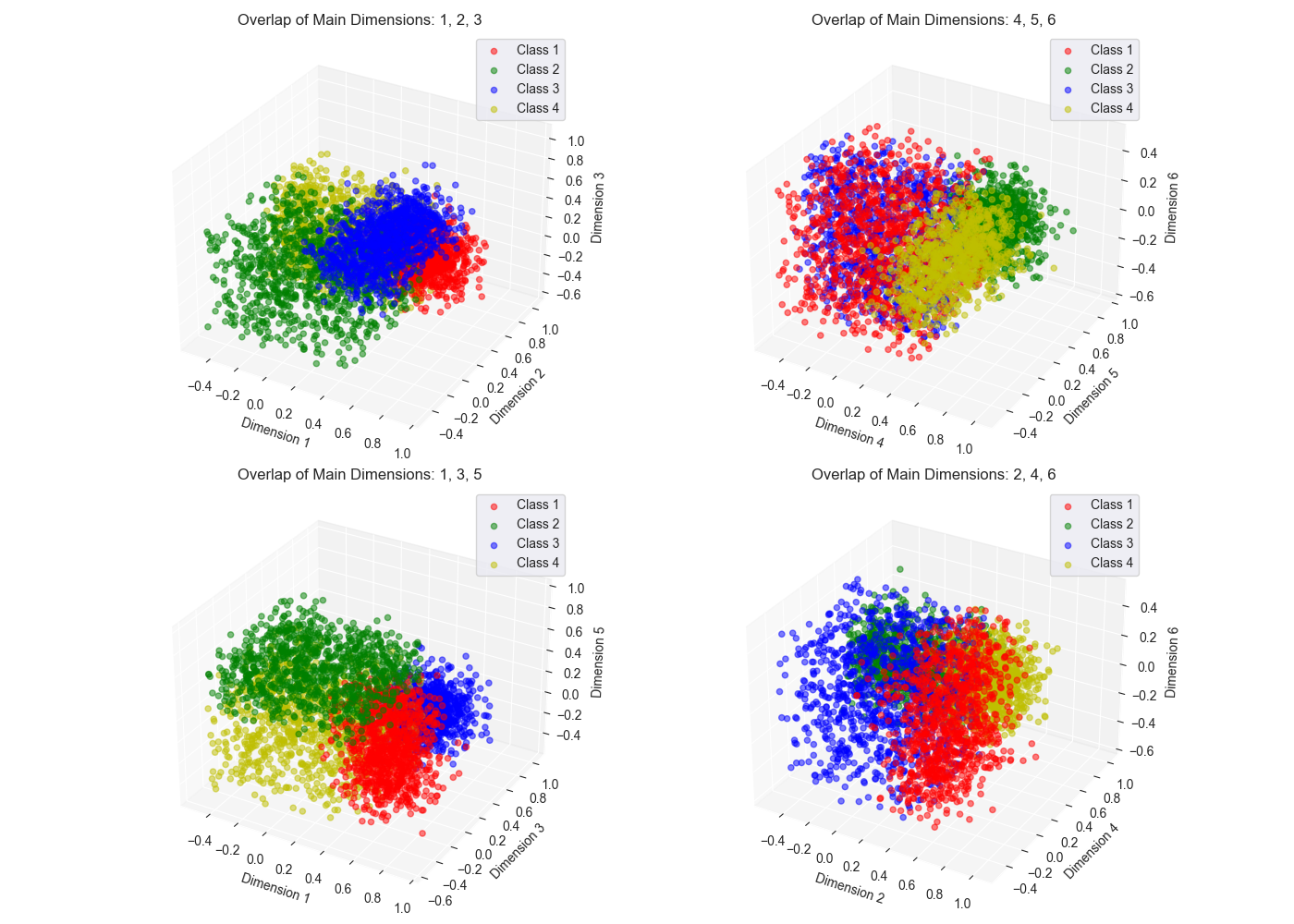}
    \caption{
        Visualization of the data-set across different subsets of $\mathbb{R}^6$.
    }
    \label{fig:data_overlap}
    \end{figure*}

    To compare the individual methods, we train a Parametrized Quantum Circuit (PQC) on the following setup.
    We generate a synthetic data-set for a $6$-dimensional classification task with four classes.
    For each class, we select three main dimensions as follows:
    \begin{table}[h!]
    \centering
    \begin{tabular}{lcc} \hline\hline
        \textbf{Class} & \textbf{Main Dimensions} & \textbf{Remaining Dimensions} \\\hline\hline
        0 & \{0, 1, 2\} & \{3, 4, 5\} \\\hline
        1 & \{3, 4, 5\} & \{0, 1, 2\} \\\hline
        2 & \{0, 2, 4\} & \{1, 3, 5\} \\\hline
        3 & \{1, 3, 5\} & \{0, 2, 4\} \\\hline\hline
    \end{tabular}
    \caption{Class specifications for main and remaining dimensions.}
    \label{tab:class-specs}
    \end{table}

    Next, we generate data for each class by sampling entries from one of two distributions, depending on the main dimensions.
    For each class, the main dimensions follow a Gaussian distribution, with different means for different classes.
    The remaining dimensions are sampled uniformly at random in an interval centered around the origin.
    Specifically, the data-generating distribution looks as follows:
    \begin{align}
    (\mathbf{x}^{(c)}, y^{(c)}) &= \left(\left\{(x_{1}^{(c)}, x_{2}^{(c)}, x_{3}^{(c)}) \sim \mathcal{N}(\mu, R^2)\right\} \cup \left\{(x_{4}^{(c)}, x_{5}^{(c)}, x_{6}^{(c)}) \sim \text{Unif}([-m, m])\right\}, \quad y^{(c)} = i\right), \\
    &\text{where} \quad c \in \{0, 1, 2, 3\}.
    \end{align}
    We use superscripts to refer to inputs in the $c^\text{th}$ class.
    To tackle this problem in a supervised learning fashion, we are given a training set of labeled data: $ S = \{(\mathbf{x}_i, \mathbf{y}_i)\}_{i=1}^{N}$ sampled according to the distribution above.
    
    Thus, each class' data forms a cluster characterized by a specific subset of three dimensions being drawn from a normal distribution, while the remaining dimensions are sampled uniformly at random.
    The normal distribution is characterized by its mean $\mu \in \mathbb{R}^3$ and the standard deviation $\sigma \in \mathbb{R}$.
    Table~\ref{tab:class-specs} below lists the main dimensions and remaining dimensions used for each class:
    \begin{table}[h!]
    \centering
    \begin{tabular}{lcc} \hline\hline
        \textbf{Distribution} & \textbf{Parameter} & \textbf{Value} \\\hline\hline
        Normal & Mean ($ \boldsymbol{\mu} $) & $ \left(\frac{1}{2}, \frac{1}{2}, 0\right) $ \\\hline
        Normal & Standard deviation ($ \sigma $) & $ \frac{1}{\sqrt{2}} \times 0.2 $ \\\hline
        Uniform & Interval endpoint (m) & $\{0.1, 0.5, \pi\}$ \\\hline\hline
    \end{tabular}
    \caption{Parameters for the normal and uniform distributions.}
    \label{tab:dist-params}
\end{table}

    Table~\ref{tab:dist-params} lists the parameters for the normal and uniform distributions used in the data generation process.
    In total, we draw $n=1000$ samples for each respective class.

\subsection{Parametrized quantum circuits}

    \begin{figure*}[b]
    \centering
    \includegraphics[width = .9\linewidth]{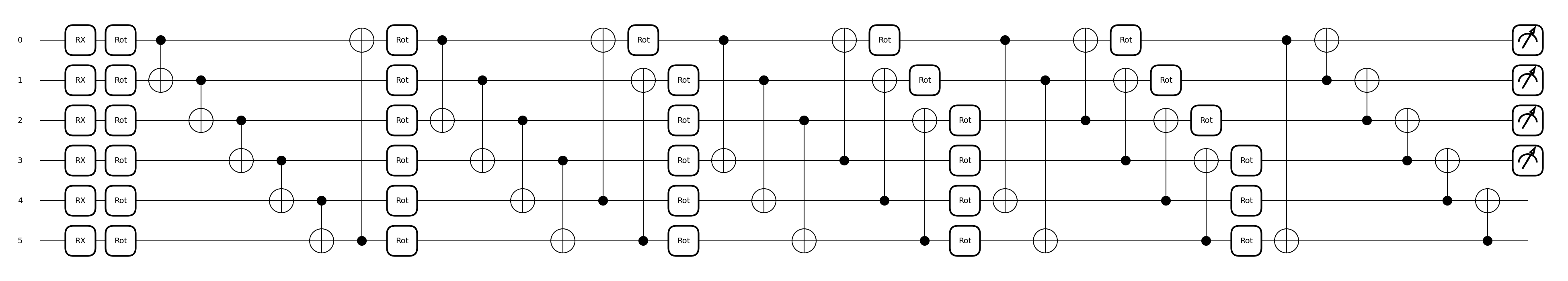}
    \caption{
        Schematic of the employed parametrized quantum circuit.
    }
    \label{fig:pqc_layout}
    \end{figure*}

    As an exemplary learning model, we employ a simple, rotational encoding based parameterized quantum circuit together with a trainable entangling gate set.
    The embedding of the numerical data is achieved by encoding the data points' individual values on differing qubits using Pauli-$X$ rotations.
    The trainable part of the PQC is given by $n_\text{layer}$ repetitions of blocks of strongly entangling layers.
    The latter consist of single qubit rotational gates followed by entanglers \cite{Schuld2020circuit}.
    We compute the prediction of a label by measuring the 
    expectation values of the single-qubit Pauli-$Z$ observable on the first 4 qubits, respectively.
    An exemplary circuit for $n_\text{layer} = 5$ is shown in Fig.~\ref{fig:pqc_layout}.

\subsubsection{Training}

    We use a categorical cross-entropy loss function to train the model.
    For a given batch of input data and associated labels $ \mathcal{B} = \{(\mathbf{x}_i, \mathbf{y}_i)\}_{i=1}^{N} $ and resulting model output $\mathbf{z}_i = f(\mathbf{x}_i)$, it is given by
    \begin{align}
        \mathcal{L}_{\text{CE}} = -\frac{1}{N} \sum_{i=1}^{N} \sum_{c=1}^{C} y_{ic} \log(\hat{y}_{ic} + \epsilon)
    \end{align}  
    where $ N $ is the number of samples in the batch, $ C $ is the number of classes, $ \epsilon = 10^{-10} $ is a small constant added for numerical stability, and 
    \begin{align}
    \hat{\mathbf{y}}_i &= \text{softmax}(\mathbf{z}_i), \\
    y_{ic} &= \begin{cases} 
    1 & \text{if class } c \text{ is the true class for sample } i, \\ 
    0 & \text{otherwise.}
    \end{cases}
    \end{align}
    We have implemented the quantum machine learning model using the PennyLane \cite{bergholm2022pennylane} library and use the JAX \cite{frostig2018Jax} backend to simulate it efficiently and train it.
    As an optimization paradigm, Adam is used in conjunction with a cosine decay schedule.
    Table~\ref{tab:hyperparameters} lists the empirically chosen hyperparameters.

    \begin{table}[ht]
    \centering
    \begin{tabular}{lc}
        \hline\hline
        Hyperparameter & Value \\\hline\hline
        learning rate $\alpha$ & $1$ \\\hline
        $n_\text{epoch}$ & 200 \\\hline
        $n_\text{layer}$ & 5 \\\hline
        batch size & 1000\\\hline\hline
    \end{tabular}
    \caption{Hyperparameters for the QML experiment.}
    \label{tab:hyperparameters}
\end{table}

    \begin{figure*}[t]
        \centering
        \includegraphics{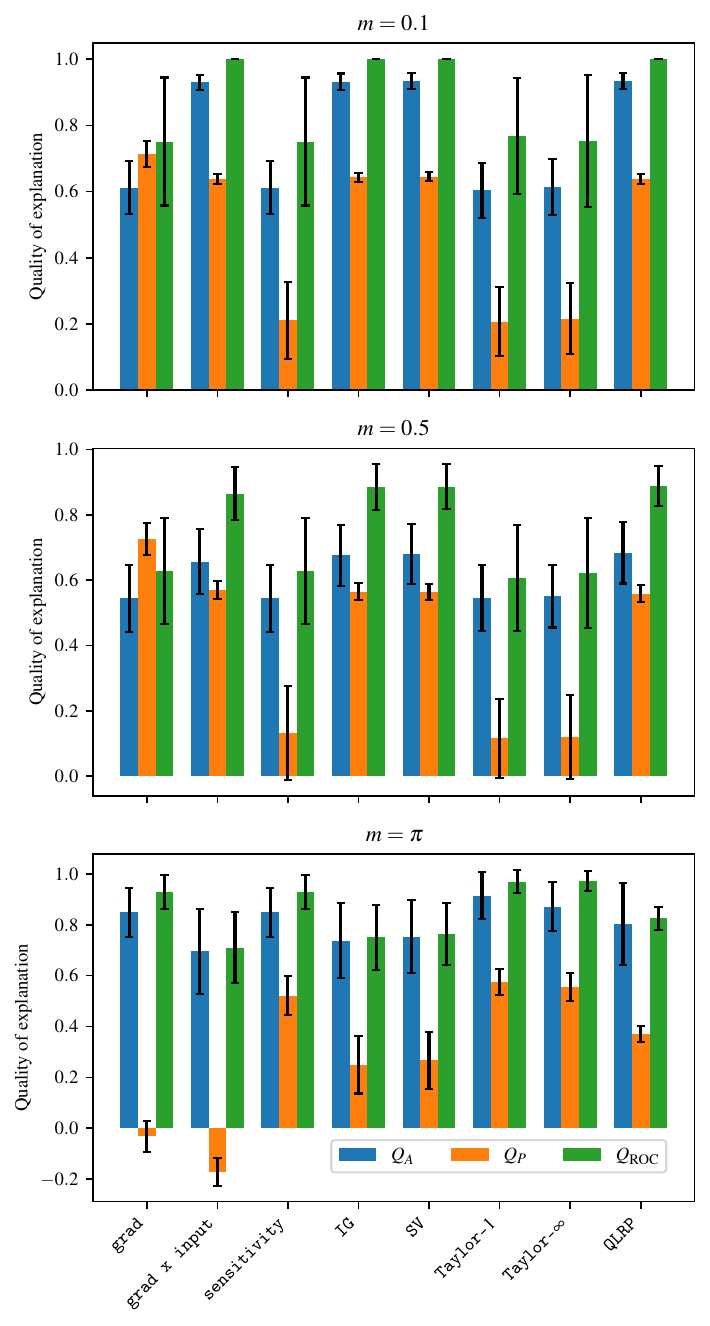}
        \caption{Full evaluation metrics}
        \label{fig:full-eval}
    \end{figure*}

\subsubsection{Evaluation metrics}

    We included more local attribution methods than introduced in the main text, we briefly present them here.
    As a starting point, one could simply take the local information present in the gradients. This leads to the explanation we refer to as \grad{}: $E_i(x) = \partial_i f(x)$, and its averaged version with respect to local perturbations $E_i(x) = \bbE_{\xi\sim\calN(0, \sigma^2 \bbI)}[\partial_i f(x+\xi)]$, which we call \smoothgrad{}~\cite{smilkov2017smoothgrad}.
    \smoothgrad{} was initially introduced to provide robustness of explanations against noise, as for instance when the gradients of the model are discontinuous.
    In general neither \grad{} nor \smoothgrad{} give rise to positive explanations\footnote{
        Requiring \emph{positive} heat-maps $E_i(x)\geq0$ might seem meaningful, as \enquote{positive causation} can be easier to interpret than \enquote{negative causation}.
        Yet, in practice that is too stringent an axiom to fulfill, so it is commonly not taken into account.
    }, which is addressed straightforwardly by the next explanation method, called \sensitivity{}: $E_i(x) = \lvert\partial_i f(x)\rvert$\footnote{
        Also standard is to call sensitivity the square of the gradient, instead of its magnitude.
    }.
    Since the positivity of \sensitivity{} comes from measuring the magnitude of the gradient, it may be the case that \sensitivity{} does not relate to positive causation.

    With the aim of comparing and evaluating the performance of the different attribution algorithms, we compute different metrics based on a desired ground-truth input mask and the local feature attributions across the data samples.
    To this end, we exploit the structure of the synthetic data-set.
    As discussed above in Appendix~\ref{app:data}, each of the four classes' input data exhibits three important and three unimportant directions in $\mathbb{R}^6$, by design.
    Thus, one can generate the ground truth masks for each class accordingly.
    The elements of the ground truth mask $\mathbf{M}^c(x) \in \mathbb{R}^6$ for a sample $(\mathbf{x}_i, \mathbf{y}_i)$ belonging to class $c$ is consequently given by 
    \begin{align}
    M_j^c(x_i) = \begin{cases} 
    1 & \text{if } j \in \{\text{main dimensions of class } c\},\\ 
    0 & \text{otherwise.}
    \end{cases}
    \end{align}
    We compute the three evaluation metrics $Q_{\{A,P,\text{ROC}\}}$ for each input $x_i$, and then average over all inputs.
    The quality scores we report are thus $\bbE_x\left[Q_A(x)\right], \bbE_x\left[Q_P(x)\right],$ and $Q_\text{ROC}(T)$, where $T$ is a set of $200$ inputs for each class.

    Fig.~\ref{fig:full-eval} summarized our numerical results across all explanations and evaluation metrics.
    
\end{document}